\PassOptionsToPackage{dvipsnames}{xcolor}
\documentclass[a4paper, UKenglish, cleveref, autoref, thm-restate,
]{lipics-v2021}
\nolinenumbers

\pdfoutput=1

\newtoggle{techreport}
\toggletrue{techreport}%

\usepackage{xargs}
\usepackage[normalem]{ulem} 

\usepackage{xifthen}        
\newcommand{\ifempty}[3]{%
  \ifthenelse{\isempty{#1}}{#2}{#3}%
}

\DeclareGraphicsExtensions{%
  .png,.PNG,%
  .pdf,.PDF,%
  .jpg,.mps,.jpeg,.jbig2,.jb2,.JPG,.JPEG,.JBIG2,.JB2
}

\newif\ifemi

\newcommandx{\preprint}[3][1=preprint,2=Springer]{
  \ifempty{#1}{}{
	 \ \\[1em]\noindent
	 \textbf{Disclaimer}
    The published version of this paper is~\cite{#1} (\copyright\ #2).
  }
}

\usepackage[textsize=tiny,colorinlistoftodos]{todonotes}

\usepackage[most]{tcolorbox}
\newtcolorbox{markbox}{
  enhanced,
  breakable,
  size=minimal,
  parbox=false,
  after={\par},
  before upper={\indent},
  colback=white,
  overlay = {
	 \draw[line width=2pt]
	 (frame.north east) -| ([xshift=3mm]frame.east) |-(frame.south east);
  },
  overlay first={\draw[line width=2pt] (frame.north east) -| ([xshift=3mm]frame.south east);},
  overlay middle={\draw[line width=2pt] ([xshift=3mm]frame.north east) -- ([xshift=3mm]frame.south east);},
  overlay last={\draw[line width=2pt] ([xshift=3mm]frame.north east) |- (frame.south east);},
}

\newcommand{\eMcomm}[2][check]{%
  \ifthenelse{\equal{#1}{new}}{{\color{red}#2}}{%
	 \ifthenelse{\equal{#1}{changed}}{{\color{teal}{#2}}}{%
		\ifthenelse{\equal{#1}{rm}}{\todo[color=black!3]{\tiny eM: removed\\'{#2}'}}{%
		  \todo[color=orange!20]{\tiny eM: \color{NavyBlue}#1}%
		  {\color{OliveGreen}{#2}}%
		}%
	 }%
  }%
}

\usepackage[]{soul}
\sethlcolor{OliveGreen!40}
\providebool{showChanges}
\newcommand{\changed}[1]{\ifthenelse{\boolean{showChanges}}{\hl{#1}}{#1}}

\newcommand{\ecoop}[1]{{#1}}

\newcommand{\hidden}[1]{}
\newcommand{\hide}[1]{}
\usepackage{xspace}

\newcommand{\cf}[2]{
  \fontsize{#1}{#1}{\selectfont{#2}}
}

\makeatletter
\newcommand{\dolist}[2]{%
  \def\nextitem{\def\nextitem{#1}}%
  \@for \el:=#2\do{\nextitem\textbf{\el}}%
}

\newcommand{\domathlist}[2]{%
  \def\nextitem{\def\nextitem{\ensuremath{#1}}}%
  \@for \el:=#2\do{\nextitem\ensuremath{\el}}%
}

\def\mktest#1{
  \def\transform##1+##2+##3{##1 piu' ##2 * ##3\penalty0}%
  \do{\expandafter\transform#1}%
}
\def\mksubscript#1{
  \def\transform##1[##2]{##1_{##2}\penalty0}%
  \do{\expandafter\transform#1}%
}
\newcommand{\mapcmd}[3][{, }]{%
  \def\nextitem{\def\nextitem{#1}}%
  \@for \el:=#3\do{\nextitem{#2{\el}}}%
}
\makeatother

\ifemi
\usepackage{showlabels}

\newcommand{\emi}[2]{
  \marginpar{\fcolorbox{red}{shadecolor}{\cf{#1}{{#2}}}}
}
\newcommand{\emic}[2]{\par
  \fcolorbox{red}{shadecolor}{\parbox{\linewidth}{
      \color{gray}
      \begin{description}
      \item[{\color{blue} #2}]{\sf #1}
      \end{description}}}
}
\else
\newcommand{\emi}[2]{}
\newcommand{\emic}[2]{}{}
\fi

\newcommand{\HSLtag}{\scalebox{1.25}{%
  \begin{tikzpicture}
  \draw (0.1,0.2) -- (0.2,0.115) -- (0.3,0.2) ;
  \draw (0.1,0.2) -- (0.15,0.1) ;
  \draw (0.3,0.2) -- (0.25,0.1) ;
  \draw (0.12,0.1) -- (0.2,0.05) -- (0.28,0.1) ;
  \draw (0.2,0) circle (0.12) ;
  \draw (0.16,0.04) circle (0.01) ;
  \draw (0.24,0.04) circle (0.01) ;
  \draw (0.15,-0.07) -- (0.25,-0.07) ;
  \end{tikzpicture}
}}

\newcommand{\hsl}[1][]{\par
  {\color{red}\vbox{\medskip\noindent\hrulefill \\[5pt]
  \HSLtag \hspace{\stretch{1}}\ifempty{#1}{HIC SUNT
  LEONES}{{#1}\hspace{\stretch{1}}} \HSLtag \\ \smallskip\noindent\hrulefill \\}}\par
}

\def\colorFun{\color{orange}}
\newcommand{\noarg}{}
\newcommand{\mkfun}[4][\colorFun]{
  \newcommand{#2}[1][#4]{%
    {#1\ensuremath{\mathsf{#3}}}%
    \ifempty{##1}{\noarg}{%
      ({##1})}%
  }%
}

\mkfun{\head}{hd}{}
\mkfun{\tail}{tl}{}

\newcommand{\sst}{\;\big|\;}
\newcommand{\dom}[1]{\operatorname{dom} {#1}}

\newcommand{\conf}[1]{\ensuremath{\langle {#1} \rangle}}

\newcommand{\qqand}[1][and]{\qquad\text{#1}\qquad}

\newcommand{\upd}[3]{{#1}[{#2} \myupd {#3}]}

\newcommand{\bnfdef}{\ensuremath{\ ::=\ }}

{\bfseries}{\rmfamily}
{\bfseries}{\rmfamily}

\makeatletter
\newcommand*{\da@rightarrow}{\mathchar"0\hexnumber@\symAMSa 4B }
\newcommand*{\da@leftarrow}{\mathchar"0\hexnumber@\symAMSa 4C }
\newcommand*{\xdashrightarrow}[2][]{%
  \mathrel{%
    \mathpalette{\da@xarrow{#1}{#2}{}\da@rightarrow{\,}{}}{}%
  }%
}
\newcommand{\xdashleftarrow}[2][]{%
  \mathrel{%
    \mathpalette{\da@xarrow{#1}{#2}\da@leftarrow{}{}{\,}}{}%
  }%
}
\newcommand*{\da@xarrow}[7]{%
  \sbox0{$\ifx#7\scriptstyle\scriptscriptstyle\else\scriptstyle\fi#5#1#6\m@th$}%
  \sbox2{$\ifx#7\scriptstyle\scriptscriptstyle\else\scriptstyle\fi#5#2#6\m@th$}%
  \sbox4{$#7\dabar@\m@th$}%
  \dimen@=\wd0 %
  \ifdim\wd2 >\dimen@
    \dimen@=\wd2 %
  \fi
  \count@=2 %
  \def\da@bars{\dabar@\dabar@}%
  \@whiledim\count@\wd4<\dimen@\do{%
    \advance\count@\@ne
    \expandafter\def\expandafter\da@bars\expandafter{%
      \da@bars
      \dabar@
    }%
  }%
  \mathrel{#3}%
  \mathrel{%
    \mathop{\da@bars}\limits
    \ifx\\#1\\%
    \else
      _{\copy0}%
    \fi
    \ifx\\#2\\%
    \else
      ^{\copy2}%
    \fi
  }%
  \mathrel{#4}%
}
\makeatother

\newcommand{\quo}[1]{\lq\lq {#1}\rq\rq}
\def\finex{{\unskip\nobreak\hfil
	 \penalty50\hskip1em\null\nobreak\hfil$\diamond$
	 \parfillskip=0pt\finalhyphendemerits=0\endgraf}
}

\AtBeginEnvironment{example}{%
	\pushQED{\qed}%
}
\AtEndEnvironment{example}{\popQED\endexample}

\definecolor{shadecolor}{rgb}{1,0.99,0.9}
\definecolor{bg}{rgb}{0.95,0.95,0.95}

\newcommand{\abcattr}[1][a]{\textsf{#1}}
\newcommand{\abccond}[1][\rho]{#1}

\def\colorExp{\color{NavyBlue}}

\newcommand{\abcexp}[1][e]{\colorExp #1}
\newcommandx{\abctuple}[1][1 = t]{\llparenthesis{#1}\rrparenthesis}
\newcommandx{\abcget}[2][1=a,2={id},usedefault=@]{
  \ptp[{#1}]{\colorOp .}\abcattr[{#2}]
}
\newcommandx{\abcptp}[2][1=a,2=\abccond,usedefault=@]{\ifempty{#1}{}{\ptp[{#1}] \ifempty{#2}{}{{\colorOp \shortmid}}} {#2}}
\newcommandx{\abcint}[6][1=a,2=\abccond,3=e,4=e',5=b,6=\abccond',usedefault=@]{
  \abcptp[{#1}][{#2}]
  \ {\colorOp \xrightarrow{\scriptstyle \abcexp[#3]\quad\abcexp[#4]}}\
  \abcptp[{#5}][{#6}]
}
\newcommandx{\mkabcint}[8][3=a,4=\abccond,5=e,6=e',7=b,8=\abccond',usedefault=@]{
  \node[bblock, #1] (#2) {$\abcint[{#3}][{#4}][{#5}][{#6}][{#7}][{#8}]$};
}

\tikzset{
    abccallout/.style={
      fill=green!10,
		opacity=.5,
		overlay,
		align=center,
      cloud callout,
		cloud puffs=15,
		aspect=2.5,
		cloud ignores aspect,
		cloud puff arc=100,
		shading=ball
    }
  }

\newcommandx{\abcP}[6][1=P,2=K,3=.1cm,4=1cm,5=north east,6=proc,usedefault=@]{
  \begin{tikzpicture}
	 \node[fill=blue!10, shape=circle] (#6) {$\p[#1]$};
	 \node[abccallout, above = #3 of #6, xshift=#4, callout absolute pointer={(#6.#5)}] {$#2$}
	 ;
	 \draw[decorate,decoration={expanding waves,angle=7,segment length = .05cm}] (#6.east) -- ++(.5cm,0)
	 ;
  \end{tikzpicture}
}

\ExplSyntaxOn
\NewDocumentCommand{\ucgreek}{m}
 {
  \str_case:nn { #1 }
   {
    {A}{\mathrm{A}}
    {B}{\mathrm{B}}
    {C}{\Sigma}
    {D}{\Delta}
    {E}{\mathrm{E}}
    {F}{\Phi}
    {G}{\Gamma}
    {H}{\mathrm{H}}
    {I}{\mathrm{I}}
    {J}{\Theta}
    {K}{\mathrm{K}}
    {L}{\Lambda}
    {M}{\mathrm{M}}
    {N}{\mathrm{N}}
    {O}{\mathrm{O}}
    {P}{\Pi}
    {Q}{\mathrm{X}}
    {R}{\mathrm{P}}
    {S}{\Sigma}
    {T}{\mathrm{T}}
    {U}{\Upsilon}
    {W}{\Omega}
    {X}{\Xi}
    {Y}{\Psi}
    {Z}{\mathrm{Z}}
   }
 }
\NewDocumentCommand{\lcgreek}{m}
 {
  \str_case:nn { #1 }
   {
    {a}{\alpha}
    {b}{\beta}
    {c}{\varsigma}
    {d}{\delta}
    {e}{\varepsilon}
    {f}{\varphi}
    {g}{\gamma}
    {h}{\eta}
    {i}{\iota}
    {j}{\vartheta}
    {k}{\kappa}
    {l}{\lambda}
    {m}{\mu}
    {n}{\nu}
    {o}{o}
    {p}{\pi}
    {q}{\chi}
    {r}{\rho}
    {s}{\sigma}
    {t}{\tau}
    {u}{\upsilon}
    {w}{\omega}
    {x}{\xi}
    {y}{\psi}
    {z}{\zeta}
   }
 }
\ExplSyntaxOff

\usepackage{wrapfig}
\usepackage{hyperref}
\usepackage{amsmath}
\usepackage{mathtools}
\usepackage{extarrows}
\usepackage{thm-restate}
\usepackage{graphicx}
\usepackage{ifthen}
\usepackage{tikz}
\usetikzlibrary{arrows,automata,calc}
\usetikzlibrary{shapes.symbols,shapes.callouts,shadows,decorations}
\tikzset{
  gt/.style={
         ->,
         >=stealth',
         shorten >=.1pt,
         auto,
         node distance=5cm,
         scale = 1,
         every state/.style={inner sep = 2pt, minimum size = 0pt, font=\footnotesize},
         transform shape
  },
lt/.style={
	 ->,>=stealth',shorten >=1pt,auto,node distance=5cm,scale = 0.97,transform shape,
  },
  initial/.style={
	 state,initial by arrow, initial text={}
  }
}
\usepackage[capitalise]{cleveref}
\usepackage{xargs}

\usepackage{listings}
\lstdefinelanguage{JavaScript}{
  morekeywords=[1]{break, continue, delete, else, for, function, if, in,
    new, return, this, typeof, var, void, while, with, type, constructor, public},
morekeywords=[2]{false, null, true, boolean, number, undefined,
    Array, Boolean, Date, Math, Number, String, Object},
morekeywords=[3]{eval, parseInt, parseFloat, escape, unescape},
  sensitive,
  morecomment=[s]{/*}{*/},
  morecomment=[l]//,
  morecomment=[s]{/**}{*/}, morestring=[b]',
  morestring=[b]"
}[keywords, comments, strings]

\lstdefinelanguage[ECMAScript2015]{JavaScript}[]{JavaScript}{
  morekeywords=[1]{await, async, case, catch, class, const, default, do,
    enum, export, extends, finally, from, implements, import, instanceof,
    let, static, super, switch, throw, try},
  morestring=[b]` }
\lstalias[]{ES6}[ECMAScript2015]{JavaScript}

\definecolor{mediumgray}{rgb}{0.3, 0.4, 0.4}
\definecolor{mediumblue}{rgb}{0.0, 0.0, 0.8}
\definecolor{forestgreen}{rgb}{0.13, 0.55, 0.13}
\definecolor{darkviolet}{rgb}{0.58, 0.0, 0.83}
\definecolor{royalblue}{rgb}{0.25, 0.41, 0.88}
\definecolor{crimson}{rgb}{0.86, 0.8, 0.24}

\lstdefinestyle{JSES6Base}{
  basicstyle=\ttfamily\footnotesize\lstextra,
  breakatwhitespace=false,
  breaklines=false,
  columns=fullflexible,
  commentstyle=\color{mediumgray}\upshape,
  emph={},
  emphstyle=\color{crimson},
  extendedchars=true,  fontadjust=true,
  identifierstyle=\color{black},
  keepspaces=true,
  keywordstyle=\color{mediumblue},
  keywordstyle={[2]\color{darkviolet}},
  keywordstyle={[3]\color{royalblue}},
  numbers=left,
  numbersep=5pt,
  numberstyle=\tiny\color{black},
  showlines=true,
  showspaces=false,
  showstringspaces=false,
  showtabs=false,
  stringstyle=\color{forestgreen},
  tabsize=2,
  title=\lstname,
  upquote=true  }
\let\lstextra=\relax

\lstdefinestyle{JavaScript}{
  language=JavaScript,
  style=JSES6Base
}
\lstdefinestyle{ES6}{
  language=ES6,
  style=JSES6Base
}

\usepackage{mfirstuc}

\usepackage{xspace}

\usepackage{float}
\newfloat{listing}{htbp}{lol}
\floatname{listing}{Listing}

\usepackage{amsmath}
\usepackage{amsthm}
\usepackage{hyperref}
\usepackage{xcolor}
\usepackage[utf8]{inputenc}
\usepackage{relsize}
\usepackage{tikz-cd}
\usepackage{graphicx}
\usepackage{tabulary}
\usepackage{url}
\usepackage{multirow}
\usetikzlibrary{automata, arrows,shadows,shapes,matrix,decorations,calc,positioning,shapes,fit}
\usepackage{caption}
\usepackage{subcaption}
\usepackage{algorithm}
\usepackage{algorithmic}
\usepackage{stmaryrd}

\providebool{ecoop}
\newcommand{\typeSet}{{\colorElem T}}
\newcommand{\state}[1]{{#1}}

\newcommand{\receiving}[1]{\xrightarrow{#1?}}
\newcommand{\mydot}{\cdot}

\newcommand{\short}[1]{\ifbool{ecoop}{#1}{}}
\newcommand{\longer}[1]{\ifbool{ecoop}{}{#1}}

\newcommand{\myupd}{\mapsto}
\newcommand{\mydots}{, \ldots, }
\newcommand{\Gn}[1]{{\colorProto\mathsf{#1}}}
\newcommand{\G}{\agt}

\renewcommand{\S}{\asysrt}

\newcommand{\M}{\afish}
\newcommand{\machine}{\afish}

\newcommand{\JB}{\operatorname{UP}}%
\newcommand{\lbje}{\mathsf{last_{up}}}
\newcommand{\initlbje}{\emptyset}
\newcommandx{\parSt}[2][1=\M,2=\lbje,usedefault=@]{({#1}, #2)}
\newcommand{\evP}[1]{#1.\lbjePointer}
\newcommand{\lbjePointer}{\mathsf{last_{up}}}%

\newcommand{\cmdname}[1]{ \text{#1}}
\newcommand{\concSet}{\mathtt{conc}}

\newcommand{\ifr}{\mathtt{ifr}}
\newcommand{\role}{\mathtt{R}}

\newcommand{\rolename}[1]{\mathtt{#1}}
\newcommand{\lgt}{{\colorLog\mathtt{l}}}
\newcommand{\evt}{{\colorElem\mathtt{t}}}
\newcommand{\lgtn}[1]{{\colorLog \mathtt{#1}}}
\renewcommand{\lg}{{\colorLog\mathit{l}}}
\newcommand{\ev}{{\colorElem e}}
\newcommand{\lgname}[1]{{\colorLog\mathit{#1}}}
\newcommand{\cmd}{\role \langle \evt \rangle }
\newcommand{\cmdindx}[1]{\role_{#1} \langle \evt_{#1} \rangle }
\newcommand{\cstmcmd}[2]{ {\rolename{#1}} \langle {\lgtn{#2}} \rangle }
\newcommand{\othercstmcmd}[2]{ {#1} \langle {#2} \rangle }
\newcommand{\mycmd}[2]{ \rolename{#1} \langle {\lgtn{#2}} \rangle }
\newcommand{\cmdnrindx}[1]{ \evt_{#1}  }
\newcommand{\cmdnr}{\evt}

\newcommand{\map}{{\colorElem{\kappa}}}
\newcommandx{\oldroles}[2][1=\G,2=\subscription,usedefault=@]{\mathit{subscribers}(#1,#2)}
\newcommandx{\newroles}[3][1=\evt,2=\G,3=\subscription,usedefault=@]{\mathit{roles}(#1,#2,#3)}
\newcommand{\proterm}{\project^{\subscription}_{\role}}

\newcommand{\conc}{\, || \,}
\newcommand{\subscription}{{\colorFun \sigma}}
\newcommand{\project}{\mathbin{\,\downarrow}}

\newcommand{\define}{:=}

\newcommand{\projshort}[2]{\project^{#1}_{\rolename{#2}}}

\newcommand{\eff}{\mathtt{eff}}
\newcommandx{\effof}[3][3=\lbje,usedefault=@]{\eff (#1,#2,#3)}
\newcommandx{\projeffof}[3][3=\R,usedefault=@]{\eff(#1,#2 \projshort{\subscription}{#3})}
\newcommandx{\srteffof}[2][1=\lg,2=\G,usedefault=@]{\eff (#1,#2)}

\newcommand{\branch}{\mathsf{branch}}
\newcommand{\mbranch}[1]{\branch({#1})}
\newcommand{\mup}[1]{{#1}.\!\JB}
 \usetikzlibrary {arrows.meta}
\usetikzlibrary{arrows}
\usetikzlibrary{shapes}
\tikzset{every picture/.style={>=stealth,bend angle=20}}
\tikzset{every label/.style={font=\scriptsize}}
\tikzset{every node/.style={font=\scriptsize}}
\tikzset{play/.style={circle,draw,minimum size=#1}}
\tikzset{play/.default=0.625cm}

\tikzset{pla/.style={circle,draw,fill=red!30,minimum size=#1}}
\tikzset{pla/.default=0.5cm}
\tikzset{plb/.style={circle,draw,fill=yellow!30,minimum size=#1}}
\tikzset{plb/.default=0.5cm}
\tikzset{player0/.style={circle,draw,fill=green!30,minimum size=#1}}
\tikzset{player0/.default=0.5cm}

\tikzset{prob/.style={diamond,draw,fill=blue!20,minimum size=#1}}
\tikzset{prob/.default=0.7cm}
\tikzset{end/.style={rectangle,draw,minimum size=#1}}
\tikzset{end/.default=0.55cm}

\newcommand{\concM}[1]{\mathbin{{\conc}_{#1}}}

\newcommand{\wfInf}[3]{{#1}  \vdash {#2} \in {#3}}

\newcommand{\roleMap}[1][]{\operatorname{r}_{#1}}
\newcommand{\roleMapApp}[2][]{\operatorname{r}_{#1}\!\left({#2}\right)}

\newcommand{\swarmAdd}{\mathbin{{:}{:}}}

\newcommandx{\asumnewnew}[2][1=\acmdrel,2={\areactionnew[]},usedefault=@]{
  \def\tmp{\dolist{\ \branchsep\ }{#2}}
  \ifempty{#1}{\tmp}{{#1} \mdot \ifempty{#2}{\zero}{\left[ \tmp \right]}}
}
\newcommandx{\gsumprefixnewnew}[3][1=i,2=\evt,3=R,usedefault=@]{
\arole[{#1}][{#3}]\conf{{\def\rkl{l}\def\rkx{#2}\ifx\rkl\rkx\aeventtypenew[{#1}][{#3}]\else\aeventtypenew[{#1}][{#2}]\fi}}
}
\newcommandx{\agsumnewnew}[5][1=i,2=I,3=\evt,4=G,5={\arole[]},usedefault=@]{
    \textstyle\sum_{#1 \in #2}\gsumprefixnewnew[{#1}][{#3}][{#5}] . \agt[{#1}][{#4}]}
    
\definecolor{myblue}{rgb}{0,0,0.6}
\definecolor{myred}{rgb}{0.4,0,0}
\definecolor{otherred}{rgb}{0.5,0,0.2}

\definecolor{mygreen}{rgb}{0.3,0.4,0.3}

\def\colorProto{\color{myred}}
\def\colorSet{\color{otherred}}
\def\colorFish{\color{myred}}
\def\colorVar{\color{myred}}
\def\colorElem{\color{otherred}}
\def\colorCmd{\color{myblue}}
\def\colorLog{\color{otherred}}
\def\colorFun{\color{mygreen}}
\def\colorSymb{\color{mygreen}}

\DeclareMathOperator{\branchsep}{{\&}}
\DeclareMathOperator{\logcat}{{\colorSymb\cdot}}

\DeclareMathSymbol{\mdotsymb}{\mathord}{symbols}{"01}
\DeclareMathOperator{\mdot}{\!{\colorSymb \mdotsymb}\!}
\DeclareMathOperator{\cte}{{\colorSymb /}}

\newcommandx{\realise}[2][1=\reactsto,2=\agt,usedefault=@]{$(#1,#2)$-realisation}

\newcommand{\eventset}{{\colorSet{\mathcal{E}}}}
\newcommand{\typeset}{{\colorSet{\mathcal{T}}}}

\newcommandx{\aevent}[2][1={},2=e,usedefault=@]{
	{\colorElem\mathit{\MakeLowercase{#2}}}_{#1}
}
\newcommandx{\acevent}[3][1={},2=e, 3=S, usedefault=@]{({#3}_{#1},\aevent[][{#2}]_{#1})}
\newcommandx{\aclog}[2][1={},2=\lambda,usedefault=@]{{\colorElem {#2}}_{#1}}

\newcommandx{\aeventtype}[2][1={},2=t,usedefault=@]{
	\mathtt{\colorElem {{#2}}}_{#1}
}
\newcommandx{\eventtyping}[3][1={}, 2=e, 3=e, usedefault=@]{
	\vdash \aevent[{#1}][{#2}] \, {\colorSymb :}\, \aeventtype[{#1}][{#3}]
}
\newcommand{\emptylog}{\epsilon}

\newcommandx{\eqg}[3][1=\aclog,2=\aclog',3=\agt]{{#1}\equiv_{#3}{#2}}

\newcommandx{\alogtype}[2][1={},2=l,usedefault=@]{{{\colorLog\mathtt{#2}}_{#1}}}
\newcommandx{\alog}[2][1={},2=l,usedefault=@]{
	{\colorLog\mathit{#2}}_{#1}
}
\newcommandx{\logof}[1][1={},usedefault=@]{\mathsf{\colorFun log}\ifempty{#1}{}{({#1})}}
\newcommandx{\reachable}[1][1={},usedefault=@]{\mathsf{\colorFun Reach}\ifempty{#1}{}{({#1})}}

\newcommand{\subterm}{\sqsubseteq}

\newcommandx{\eqlog}[3][3=\aspecdef,usedefault=@]{#1\equiv_{#3}#2}

\newcommandx{\leqlog}[1][1=\alog]{\, <_{#1} \,}
\newcommandx{\geqlog}[1][1=\alog]{\, >_{#1} \,}

\newcommandx{\acmd}[2][1={},2=c,usedefault=@]{{\colorCmd\mathsf{{#2}}}_{#1}}
\newcommand{\inputsymbol}{{\colorSymb ?}}

\newcommandx{\apref}[2][1={},2=e,usedefault=@]{\aeventtype[{#1}][{#2}]\!\inputsymbol}
\newcommandx{\inp}[2][1={}, 2=t, usedefault=@]{
	\aeventtype[{#1}][{#2}]\inputsymbol\,
}
\makeatletter
\newcommandx{\inps}[2][1={}, 2=e, usedefault=@]{
	\def\nextitem{\def\nextitem{#1}}\@for \el:=#2 \do{\aeventtype[{#1}][{\el}]\inputsymbol\, \nextitem}}
\makeatother

\newcommandx{\outp}[2][1={\acmd},2=\aeventtype]{\ifempty{#1}{#2}{#1/{{#2}}}\colorSymb{!}}

\newcommand{\aC}{\colorElem{C}}
\newcommand{\acmdrel}{{\colorElem{\kappa}}}

\newcommandx{\afish}[4][1={},2=M, 3=\aC, 4=\acmdrel, usedefault=@]{\ifempty{#2}{{#3}_{#1} \cdot {#4}_{#1}} {{{\colorFish \mathsf{#2}}_{#1}}}}
\newcommandx{\afishact}[3][1={},2=c,3=,usedefault=@]{
	\acmd[{#1}][{#2}] \cte \ifempty{#3}{\aeventtype[{#1}]}{\aeventtype[{#1}][{#3}]}
}
\newcommandx{\obscmd}[3][1=\afish, 2={\acmd[]}, 3 = {\alogtype[]}, usedefault=@]{
	{\big(#1\big)}\!\downarrow_{\afishact[][{#2}][{#3}]}
}

\newcommandx{\agt}[2][1={},2=G,usedefault=@]{\ifempty{#1}{{\colorProto \mathsf{#2}}}{{\colorProto \mathsf{#2}}_{#1}}}

\newcommandx{\fishid}[1][1=f]{{\colorFish {\mathtt{#1}}}}

\newcommandx{\asys}[2][1={},2=S,usedefault=@]{{{\colorElem\mathbb{#2}}_{#1}}}
\newcommandx{\apond}[2][1={},2=M,usedefault=@]{{{\color{myred}\mathcal{#2}}_{#1}}}

\newcommand{\roleset}{{\color{myred}\mathcal{R}}}

\newcommandx{\asysrt}[2][1={},2=S,usedefault=@]{{{\color{myred}\mathsf{#2}}_{#1}}}
\newcommandx{\metasys}[2][1={},2=S,usedefault=@]{{{\colorElem\mathfrak{#2}}_{#1}}}

\newcommandx{\fishes}[2][1={},2=L,usedefault=@]{{\colorElem\mathbb{#2}}_{#1}}

\newcommandx{\arsys}[2][1={},2=S,usedefault=@]{{\colorElem{\mathbb {#2}_{#1}}}}

\newcommandx{\areaction}[3][1=i,2=e,3=\afish,usedefault=@]{
	\aeventtype[{#1}][{#2}] {\colorSymb ?}\, {{#3}_{#1}}
}

\makeatletter
\def\mklogtype#1{
	\def\nextitem{\def\nextitem{\logcat\penalty0}}\@for\el:=#1\do{\nextitem\aeventtype[][{\el}]}}

\def\mklog#1{
	\ifempty{#1}{\emptylog}{
		\def\nextitem{\def\nextitem{\logcat\penalty0}}\@for\el:=#1\do{\nextitem{\el}}}
}

\def\mksys#1{
	\def\transform##1[##2]{(\afish[][##1],\mklog{##2})\parop\penalty0}\@for\el:=#1\do{\expandafter\transform\el}}
\makeatother

\newcommandx{\asumnew}[2][1=\acmdrel,2={\areaction[]},usedefault=@]{
	\def\tmp{\dolist{\ \branchsep\ }{#2}}
	\ifempty{#1}{\tmp}{{#1} \mdot \ifempty{#2}{\zero}{\left[ \tmp \right]}}
}
\newcommand{\atrole}{\texttt{\colorSymb \upshape@}}

\newcommandx{\gsumpref}[5][1={},2=c,3=e,4=R,5={}]{\acmd[{#1}][{#2}] \cte \alogtype[{#1}][{#3}] \atrole \arole[{#1}][{#4}]{\gtpref}\ifempty{#5}{\zero}{{#5}_{#1}}}

\newcommand{\gtpref}{\,{\colorSymb \cdot}\penalty0\,}

\newcommandx{\gsumprefix}[4][1=i,2=c,3=l,4=R,usedefault=@]{
	\acmd[{#1}][{#2}]{\atrole\arole[{#1}][{#4}]}\conf{{\def\rkl{l}\def\rkx{#3}\ifx\rkl\rkx\alogtype[{#1}][{#3}]\else\mklogtype{#3}\fi}}
}

\newcommandx{\arole}[2][1=i,2=R,usedefault=@]{{\mathtt{#2}_{#1}}}
\newcommandx{\agsum}[6][1=i,2=I,3=\acmd,4=l,5=G,6={\arole[]},usedefault=@]{
	\sum_{#1 \in #2}\gsumprefix[{#1}][{#3}][{#4}][{#6}] \gtpref \agt[{#1}][{#5}]}

\newcommandx{\agprodrun}[5][1=i,2=I,3=\acmd,4=\alogtype,5=\agt,usedefault=@]{
	\pi_{#1 \in #2}{#3_#1/#4_#1;#5}}

\mkfun[\colorFun]{\reactsto}{$\sigma$}{}
\mkfun[\colorFun]{\achi}{$\chi$}{}
\newcommand{\alambda}{{\reactsto}}

\newcommandx{\aspecdef}[3][1=\agt,2=\achi,3=\reactsto,usedefault=@]{#1,#3}
\mkfun[\colorFun]{\cmds}{cmds}{}
\mkfun[\colorFun]{\events}{events}{}

\mkfun[\colorFun]{\norec}{norec}{}
\mkfun[\colorFun]{\auxcr}{cr}{}
\mkfun[\colorFun]{\cmdrdy}{crdy}{}
\mkfun[\colorFun]{\rdy}{guards}{}
\mkfun[\colorFun]{\cmpts}{roles}{}
\mkfun[\colorFun]{\filter}{filter}{}
\mkfun[\colorFun]{\activer}{active}{}
\mkfun[\colorFun]{\passiver}{passive}{}
\mkfun[\colorFun]{\observed}{obs}{}

\newcommandx{\aproj}[3][1={\arole[]},2=\reactsto,3=\achi,usedefault=@]{\downarrow_{#1}^{{#2}}}

\newcommandx{\asem}[1][1={}]{[\![ #1 ]\!]}

\newcommandx{\atripg}[3][1=\agt,2=\alambda,3=\achi]{\langle#1,#2,#3\rangle}
\renewcommand{\atripg}{\agt}

\newcommand{\zero}{{{\colorSymb \mathbf{0}}}}
\newcommandx{\avar}[1][1=X]{{\colorVar \mathtt{#1}}}
\newcommandx{\avartype}[1][1=X]{{\colorVar \mathtt{#1}}}

\newcommandx{\arec}[2][1=\avar,2=\afish,usedefault=@]{
	{{#1} = {#2}}
}

\newcommand{\parop}{\,{\colorSymb \mid }\,}
\newcommand{\stchange}[1][]{{\colorFun{\delta}_{#1}}}

\newcommand{\stchangeBT}{\stchange[BT]}

\newcommandx{\red}[1][1={},usedefault=@]{\xlongrightarrow{#1}}
\newcommandx{\weakred}[1][1={},usedefault=@]{\xLongrightarrow{#1}}
\newcommandx{\absred}[1][1={},usedefault=@]{\xmapsto{#1}}
\makeatletter
\def\rightarrowfill@@{\arrowfill@@\relax\relbar\rightarrow}
\def\arrowfill@@#1#2#3#4{$\m@th\thickmuskip0mu\medmuskip\thickmuskip\thinmuskip\thickmuskip
	\relax#4#1
	\xleaders\hbox{$#4#2$}\hfill
	#3$}
\newcommand{\metared}[2][]{\ext@arrow 0359\rightarrowfill@@{#1}{#2}}
\makeatother
\newcommandx{\noBTMachineStep}[1][1={},usedefault=@]{\red[#1]}
\newcommandx{\noBTSwarmStep}[1][1={},usedefault=@]{\red[#1]}

\newcommand{\irule}[2]{\frac{\textstyle\rule[-1.3ex]{0cm}{3ex}#1}{\textstyle\rule[-.5ex]{0cm}{3ex}#2}}

\newcommand{\rulename}[1]{{\sc [#1]}}

\def \mathrule #1#2#3{
	\irule{#1}{#2}\ifempty{#3}{}{\hspace{0em}\mbox{\footnotesize\rulename{#3}}}
}

\newcommandx{\acmdnew}[2][1={},2=c,usedefault=@]{
	{\colorCmd\mathit{\MakeLowercase{#2}}}_{#1}
}

\newcommandx{\aeventtypenew}[2][1={},2=e,usedefault=@]{
	\mathsf{\colorElem {{#2}}}_{#1}
}

\newcommandx{\afishactnew}[3][1={},2=c,3=,usedefault=@]{
	\acmd[{#1}][{#2}] \cte \ifempty{#3}{\aeventtypenew[{#1}]}{\aeventtype[{#1}][{#3}]}
}

\newcommandx{\areactionnew}[3][1=i,2=e,3=\afish,usedefault=@]{
	\aeventtypenew[{#1}][{#2}] {\colorSymb ?}\, {{#3}_{#1}}
}

\newcommandx{\gsumprefixnew}[4][1=i,2=c,3=e,4=R,usedefault=@]{
	\acmd[{#1}][{#2}]{\atrole\arole[{#1}][{#4}]}\conf{{\def\rkl{l}\def\rkx{#3}\ifx\rkl\rkx\aevent[{#1}][{#3}]\else\aeventtypenew[{#1}][{#3}]\fi}}
}

\newcommandx{\agsumnew}[6][1=i,2=I,3=\acmd,4=e,5=G,6={\arole[]},usedefault=@]{
	\textstyle\textstyle\sum_{#1 \in #2}\gsumprefixnew[{#1}][{#3}][{#4}][{#6}] \gtpref \agt[{#1}][{#5}]}

\newcommandx{\aprojnew}[3][1={\arole[]},2=\reactsto,3=\achi,usedefault=@]{\downarrow_{#1}^{{#2}}}

\mkfun[\colorFun]{\afterset}{after}{}

\newcommandx{\denv}[1][1={},usedefault=@]{\Gamma\ifempty{#1}{}{.{#1}}}

\newcommandx{\comptwo}[2][1={\agt[1]},2={\agt[2]},usedefault=@]{{#1} \; || \; {#2}}

\usepackage{pgfplots}
\pgfplotsset{compat=newest}

\definecolor{dtured}    {rgb/cmyk}{0.6,0,0 / 0,0.91,0.72,0.23}
\definecolor{blue}      {rgb/cmyk}{0.1843,0.2431,0.9176 / 0.88,0.76,0,0}
\definecolor{brightgreen}{rgb/cmyk}{0.1216,0.8157,0.5098 / 0.69,0,0.66,0}
\definecolor{navyblue}  {rgb/cmyk}{0.0118,0.0588,0.3098 / 1,0.9,0,0.6}
\definecolor{yellow}    {rgb/cmyk}{0.9647,0.8157,0.3019 / 0.05,0.17,0.82,0}
\definecolor{orange}    {rgb/cmyk}{0.9882,0.4627,0.2039 / 0,0.65,0.86,0}
\definecolor{pink}      {rgb/cmyk}{0.9686,0.7333,0.6941 / 0,0.35,0.26,0}
\definecolor{grey}      {rgb/cmyk}{0.8549,0.8549,0.8549 / 0,0,0,0.2}
\definecolor{red}       {rgb/cmyk}{0.9098,0.2471,0.2824 / 0,0.86,0.65,0}
\definecolor{green}     {rgb/cmyk}{0,0.5333,0.2078 / 0.89,0.05,1,0.17}
\definecolor{purple}    {rgb/cmyk}{0.4745,0.1373,0.5569 / 0.67,0.96,0,0}
\definecolor{aBlue}	    {RGB/cmyk}{0,158,220/1.0,0.0,0.0,0.0}
\definecolor{aRed}		{RGB/cmyk}{220,85,82/0.0,0.87,0.8,0.0}
\definecolor{orchid}	{RGB/cmyk}{197, 174, 207/0.05, 0.16, 0, 0.19} %
\definecolor{aPurple}	{RGB/cmyk}{70, 41, 90/0.22, 0.54, 0, 0.65}
\definecolor{lightblue1} {RGB/cmyk}{173, 190, 211/0.18,0.1,0,0.17}
\definecolor{lightblue2} {RGB/cmyk}{110, 129, 190/0.42,0.32,0,0.25}

\pgfplotscreateplotcyclelist{DTU}{%
dtured,         fill=dtured,        \\%
blue,           fill=blue,          \\%
brightgreen,    fill=brightgreen    \\%
navyblue,       fill=navyblue       \\%
yellow,         fill=yellow         \\%
orange,         fill=orange         \\%
grey,           fill=grey           \\%
red,            fill=red            \\%
green,          fill=green          \\%
purple,         fill=purple         \\%
aBlue,          fill=aBlue          \\%
aRed,           fill=aRed           \\%
orchid,         fill=orchid         \\%
aPurple,        fill=aPurple        \\%
lightblue1,     fill=lightblue1     \\%
lightblue2,     fill=lightblue2     \\%
}
\usetikzlibrary{spy}    %
\usepgfplotslibrary{statistics} %
\pgfplotsset{ %
compat=newest,
major x grid style={line width=0.5pt,draw=grey},
major y grid style={line width=0.5pt,draw=grey},
legend style={at={(0.5,-0.1)}, anchor=north,fill=none,draw=none,legend columns=-1,/tikz/every even column/.append style={column sep=10pt}},
axis line style={draw=none},
tick style={draw=none},
every axis/.append style={ultra thick},
tick label style={/pgf/number format/assume math mode}, %
}
\tikzset{every mark/.append style={scale=1.5}}

\usepackage{mathpartir}

\crefname{definition}{Def.{}}{Definitions}
\Crefname{definition}{Def.{}}{Definitions}
\crefname{figure}{Fig.{}}{Figures}
\Crefname{figure}{Fig.{}}{Figures}

\usepackage{textcomp}%
\Crefname{section}{\textsection{}}{\textsection{}\textsection{}}
\crefname{section}{\textsection{}}{\textsection{}\textsection{}}

\booltrue{ecoop}%

\booltrue{showChanges}

\title{Compositional Design, Implementation, and Verification of Swarms%
	\iftoggle{techreport}{ (Technical Report)}{}%
}

\author{Florian Furbach}{University of Surrey, UK and Technical University of Denmark, Denmark}{f.furbach@surrey.ac.uk}{https://orcid.org/0009-0008-4922-9363}{}%
\author{Lucas Clorius}{Technical University of Denmark, Denmark}{luccl@dtu.dk}{https://orcid.org/0009-0005-2224-2419}{}%
\author{Roland Kuhn}{Actyx AG, Germany \and \url{https://rolandkuhn.com/}}{roland@actyx.io}{https://orcid.org/0000-0003-1582-6238}{}%
\author{Hern\'an Melgratti}{%
	University of Buenos Aires \& Conicet, Argentina
	\and
	\url{https://lafhis.dc.uba.ar/~melgratti}}{hmelgra@dc.uba.ar}{https://orcid.org/0000-0003-0760-0618}{}%
\author{Alceste Scalas}{Technical University of Denmark, Denmark \and \url{https://people.compute.dtu.dk/alcsc}}{alcsc@dtu.dk}{https://orcid.org/0000-0002-1153-6164}{}%
\author{Emilio Tuosto}{
	Gran Sasso Science Institute, %
	Italy
	\and
	\url{https://cs.gssi.it/emilio.tuosto}}{emilio.tuosto@gssi.it}{https://orcid.org/0000-0002-7032-3281}
	{}%

\authorrunning{F.~Furbach, L.~Clorius, R.~Kuhn, H.~Melgratti, A.~Scalas, and E.~Tuosto}

\Copyright{Florian Furbach, Lucas Clorius, Roland Kuhn, Hern\'an Melgratti, Alceste Scalas, Emilio Tuosto} 



\category{} %

\iftoggle{techreport}{%
	\relatedversiondetails{ECOOP'26 paper}{https://doi.org/10.4230/LIPIcs.ECOOP.2026.21}
}{}

\funding{%
  Research partially supported by the Horizon Europe grant 101093006 (TaRDIS) %
  and by the Italian Ministero dell'Università e della Ricerca ``Dipartimenti di eccellenza'' initiative.%
}%

\acknowledgements{This work was inspired by the discussions at the Dagstuhl Seminar 24051 ``Next Generation Protocols for Heterogeneous Systems'' \cite{DBLP:journals/dagstuhl-reports/BalzerC0024}. We thank the organisers of the meeting and Schloss Dagstuhl -- Leibniz Center for Informatics for making this work possible.}%

\iftoggle{techreport}{%
  \hideLIPIcs%
}{}

\EventEditors{Robbert Krebbers and Alexandra Silva}
\EventNoEds{2}
\EventLongTitle{40th European Conference on Object-Oriented Programming (ECOOP 2026)}
\EventShortTitle{ECOOP 2026}
\EventAcronym{ECOOP}
\EventYear{2026}
\EventDate{June 29--July 3, 2026}
\EventLocation{Brussels, Belgium}
\EventLogo{}
\SeriesVolume{372}
\ArticleNo{21}

\newtheorem{result}[theorem]{Result}

\newcommand{\changeNoMargin}[1]{#1}
\newcommand{\changeOursNoMargin}[1]{#1}

\makeatletter
  \makeatletter
  \patchcmd{\@addmarginpar}{\ifodd\c@page}{\ifodd\c@page\@tempcnta\m@ne}{}{}
  \makeatother
  	\reversemarginpar
  \newcounter{change}
  \renewcommand{\changeNoMargin}[1]{{\color{blue}{#1}}}%
  \renewcommand{\changeOursNoMargin}[1]{{\color{red}{#1}}}%
\makeatother

\begin{document}

\maketitle

\begin{abstract}
  \emph{Swarm protocols} are a recently introduced formalism for
  specifying, implementing, and verifying peer-to-peer systems called
  \emph{swarms}.  A swarm consists of distributed agents called
  \emph{machines} that communicate by asynchronous event propagation.
  Following a \emph{local-first} model, each machine can progress
  without requiring continuous connectivity to other machines.
  Existing models of swarms are not compositional, making the modular
  development of large and complex swarm applications as well as the
  reuse of code difficult.
  We address these issues by presenting novel theory and techniques
  for the compositional specification, verification, and
  implementation of swarms.  These results enable the correct
  compositional reuse of pre-existing swarm protocols and machine
  implementations.
  We implement these contributions in a companion software artifact
  which enables the automatic integration of independently designed
  and verified swarm components.
  
\end{abstract}

\section{Introduction}
\label{sec:intro}

Modern distributed systems face significant challenges in ensuring reliability, availability, and scalability.
Recent work on \emph{swarm protocols}~\cite{DBLP:conf/ecoop/KuhnMT23,10.1145/3597926.3604917}
addresses these challenges by formally specifying the behavior of interacting agents capable of independent yet coordinated progress.
In this approach, a \emph{swarm} is modelled as an ensemble of distributed interacting agents, called \emph{machines}.
Each machine can coordinate with others by emitting \emph{events} (i.e., messages) that propagate asynchronously through the swarm, and by \emph{subscribing} to (some of) the events emitted by other machines. Essentially, swarms are type-based publish/subscribe systems with distributed asynchronous delivery~\cite{Eugster03PubSub}.
Unlike most publish/subscribe systems that aim for causal-consistent delivery (i.e.,  machines share a consistent view of the global system state), swarms are designed to operate over inconsistent views, often corresponding to inconsistent cuts of the global state~\cite{babaoglu1993consistent}.
This design choice prioritises availability over consistency~\cite{CAP}, following a local-first paradigm \cite{lfc, localfirst}: each machine can make progress independently, without requiring up-to-date global information or a continuously available network connection.
Such an approach enables decentralized decision-making, reducing reliance on central coordination and improving robustness. However, it also introduces potential inconsistencies: machines may take decisions based on stale or divergent state, which may later require reconciliation~\cite{parker1983detection}.

The approach in
\cite{DBLP:conf/ecoop/KuhnMT23,10.1145/3597926.3604917} advocates the
usage of behavioral types~\cite{hlvccdmprt16,gay2017behavioural,ancona2016behavioral} to ensure
that a swarm correctly reconciles inconsistencies.
Roughly, the intended global behaviour of the swarm (i.e., the
expected interactions among machines) is formalised as a \emph{swarm
  protocol}: a specification that provides a bird's-eye view of how
machines playing different \emph{roles} should interact at runtime, %
without causing irreconcilable inconsistencies in the overall state of the swarm.
A swarm protocol $\G$ can then be \emph{projected} onto the different roles to obtain the corresponding machine specifications.
When an ensemble of machines, each running according to a projection of $\G$, is deployed to form a swarm $\S$, then (under \emph{well-formedness} conditions discussed later) the swarm $\S$ is guaranteed to execute in accordance with $\G$.

\subparagraph{A Use Case of Swarm Protocols and Swarms.}%
To illustrate the theory and practice of swarm protocols, we present a
simplified use case in factory automation (\quo{Industry
  4.0}).\footnote{We borrow the use case of an automatic warehouse
  operated by a swarm which has been implemented in the open source
  Actyx toolkit~\cite{actyx}.}
The intended global behaviour of the swarm is specified by the
$\Gn{Warehouse}$ protocol illustrated in \Cref{fig:Warehouse}.
A machine playing the role of a transport vehicle $\rolename{T}$ requests a part by emitting an event of type $\lgtn{partReq}$. A machine playing the forklift role $\rolename{FL}$ responds by delivering the requested part to the specified pick-up position $\lgtn{pos}$. A machine playing role $\rolename{T}$ then picks up the part and emits an event of type $\lgtn{partOK}$. Between consecutive requests, the door $\rolename{D}$ may close, emitting an event of type $\lgtn{closingTime}$.

\begin{figure}[t]
  \short{\vspace{-3mm}}%
  \centering
  \begin{minipage}[b]{.425\textwidth}
    \centering
    \begin{tikzpicture}[gt, scale=\ifbool{ecoop}{.85}{.85}, transform shape, label distance=1mm, node distance=4mm and 22mm]
        \node (0) [pla]{0};
        \node (begin) [left =0.5cm of 0] {};
        \node (middle) [right = of 0]{};
        \node (2) [pla, right = of middle] {2};
        \node (1) [pla, above = of middle] {1};
        \node (3)  [pla, below = 1.7mm of middle] {3};

        \draw[->] (begin) to node[sloped, anchor=center, above=1mm]{}  (0);
        \draw[->] (0) to node[sloped, anchor=center, above=0.5mm]{$ \cstmcmd{T}{partReq}$}  (1);
        \draw[->] (1) to node[sloped, anchor=center, above=0.5mm]{$ \cstmcmd{FL}{pos}$}  (2);
        \draw[->] (2) to node[sloped, anchor=center, above=0.2mm]{$ \cstmcmd{T}{partOK}$}  (0);
        \draw[->] (0) to node[sloped, anchor=center, below=0.5mm]{$ \cstmcmd{D}{closingTime}$}  (3);
      \end{tikzpicture}
      \short{\vspace{-3mm}}%
      \captionof{figure}{The swarm protocol $\Gn{Warehouse}$, involving the roles $\rolename{T}$ (transport), $\rolename{FL}$ (forklift), and $\rolename{D}$ (door). }
      \label{fig:Warehouse}
      \smallskip

      \begin{tikzpicture}[gt, scale=\ifbool{ecoop}{.85}{.85}, transform shape, label distance=1mm, node distance=3.5mm and 22mm]
          \node (0) [plb]{0};
          \node (begin) [below =0.3cm of 0] {};
          \node (middle) [right = of 0]{};
          \node (2) [plb, right = of 0] {1};
          \node (3)  [plb, left = of 0] {3};

          \draw[->] (begin) to node[sloped, anchor=center, above=1mm]{}  (0);
          \draw[->,out=110,in=70,loop] (0) to node[sloped, anchor=center, above=0mm, bend left]{$\lgtn{closingTime}!$}  (0);
          \draw[->,bend left] (2) to node[sloped, anchor=center, below=0mm, bend left]{$  \lgtn{partOK}?$}  (0);
          \draw[->,bend left] (0) to node[sloped, anchor=center, above=0mm, bend left]{$ \lgtn{partReq}?$}  (2);
          \draw[->] (0) to node[above]{$ \lgtn{closingTime}?$}  (3);
        \end{tikzpicture}
        \short{\vspace{-3mm}}%
        \captionof{figure}{Machine $\M_\rolename{D}$ obtained by projecting the swarm protocol $\Gn{Warehouse}$ (\Cref{fig:Warehouse}) onto role $\rolename{D}$ (door).}
        \label{fig:Door}
      \end{minipage}
      \hfill
      \begin{minipage}[b]{.555\textwidth}
        \scalebox{1.25}{%
          \hspace{5mm}%
\begin{lstlisting}[style=ES6, backgroundcolor=\color{white}, xleftmargin=.05\textwidth, basicstyle=\tiny\ttfamily, caption={}, label=lst:ts-TT]
// Machine implementation using the Actyx toolkit
const door = Warehouse.makeMachine('D')

// States
const s0 = door.designEmpty('s0')
    .command('close', [closingTime], () => {
        ...
        return [closingTime.make({ time: new Date() })]
    })
    .finish()
const s1 = door.designEmpty('s1').finish()
const s3 = door.designEmpty('s3').finish()

// Accept events and change states
s0.react([partReq], s1, () => { ...; return s1.make() })
s0.react([closingTime], s3, () => { ...; return s3.make() })
s1.react([partOK], s0, () => { ...; return s0.make() })
\end{lstlisting}
        }
        \captionof{figure}{Implementation of a $\rolename{D}$oor Machine in Actyx.}
        \label{fig:code_machine_door}
      \end{minipage}%
    \end{figure}
    \begin{figure}[t]
      \centering
      \begin{tikzpicture}[gt, scale=\ifbool{ecoop}{.8}{1}, transform shape, label distance=2mm, node distance=26mm]

          \node (0) [pla]{0};
          \node (begin) [left =0.5cm of 0] {};
          \node (1) [pla, right = of 0] {1};
          \node (2) [pla,, right = of 1] {2};
          \node (3)  [pla, right = of 2] {3};

          \draw[->] (begin) to node[sloped, anchor=center, above=0.5mm]{}  (0);
          \draw[->] (0) to node[sloped, anchor=center, above=0.5mm]{$ \cstmcmd{T}{partReq}$}  (1);
          \draw[->] (1) to node[sloped, anchor=center, above=0.5mm]{$ \cstmcmd{T}{partOK}$}  (2);
          \draw[->] (2) to node[sloped, anchor=center, above=0.5mm]{$ \cstmcmd{A}{car}$}  (3);
        \end{tikzpicture}
        \short{\vspace{-2mm}}%
        \caption{A $\Gn{Factory}$ protocol that can interface with the $\Gn{Warehouse}$ (\Cref{fig:Warehouse}) using the role $\rolename{T}$.}
        \label{fig:Factory}
      \end{figure}

The swarm protocol $\Gn{Warehouse}$ is then \emph{projected} onto its roles $\rolename{T}$, $\rolename{FL}$, and $\rolename{D}$, yielding a machine specification for each role. For example, the projection of $\Gn{Warehouse}$ onto the role $\rolename{D}$ (door) results in the state machine shown in \Cref{fig:Door}.
For each event type, say $\lgtn{closingTime}$, a transition labelled $\lgtn{closingTime}!$
indicates that the machine \emph{emits} an event of type $\lgtn{closingTime}$,
while a transition labelled $\lgtn{closingTime}?$ indicates that the machine \emph{accepts} an
event of that type.
During the execution of a swarm, when a machine emits an event, the event is appended to the emitting machine's \emph{local event log};
then, the event asynchronously propagates to the logs of other machines. A machine changes its state only when it
inspects its log and accepts an event from it;
an event emission does not directly cause a state change
-- rather, after an event is emitted, its acceptance triggers a state
transition. For example, in \Cref{fig:Door}, state $\state 0$ of machine $\M_\rolename{D}$ has a self-loop
emitting an event of type $\lgtn{closingTime}$, as well as a
transition that accepts the same event type.
The non-atomic nature of event emission ($\lgtn \_!$) and acceptance ($\lgtn \_
?$) implies that events from different machines may interleave: i.e., after a
machine emits an event, and before that event is accepted, other events from
other machines may be emitted, propagated, and accepted.

\medskip%

In practice, machines can be implemented using the open source Actyx toolkit~\cite{actyx}, which provides a TypeScript API for defining swarm protocols, implementing machines, and deploying swarms. %
\Cref{fig:code_machine_door} shows the TypeScript implementation of the state machine
$\M_\rolename{D}$ in \Cref{fig:Door}: %
the machine is called \lstinline[style=ES6]|door| and plays the role \lstinline[style=ES6]|'D'| (line 2).
Lines 5--12 define the states \lstinline[style=ES6]|s0|, \lstinline[style=ES6]|s1|, \lstinline[style=ES6]|s3| that map to states $\state 0$, $\state 1$, and $\state 3$ in $\M_\rolename{D}$, respectively, as follows.
Line 5 introduces the state named \lstinline[style=ES6]|s0| while lines 6--9 define the transition that emits an event of type \lstinline[style=ES6]|closingTime|, which is declared using the \lstinline[style=ES6]|command| method. The elided code on line 7 implements the business logic associated with the event emission, while line 8 specifies the information attached to the emitted event (in this case, a timestamp).
Lines 11 and 12 introduce the states  \lstinline[style=ES6]|s1| and  \lstinline[style=ES6]|s3|.
Finally, lines 15–17 define the accepting transitions. In particular, line 16 specifies that when the machine is in state \lstinline[style=ES6]|s0| and receives an event of type \lstinline[style=ES6]|closingTime|, it accepts the event and transitions to state \lstinline[style=ES6]|s3|. The omitted code implements the corresponding application-specific business logic.

Once the code is written, the Actyx toolkit enables static verification to check that a machine's implementation adheres to an intended behaviour. Specifically, it provides APIs for defining swarm protocols, tools for projecting protocols onto roles, a type extraction tool that reconstructs machines from TypeScript code, and a type-checking tool to check whether the projected and extracted machines are consistent~\cite{actyx}. For instance, the Actyx toolkit can verify that the code in \Cref{fig:code_machine_door} conforms to the machine in \Cref{fig:Door}, that is projected from \Cref{fig:Warehouse}.

\medskip%

Going back to the theory, the use of swarm protocols ensures a key correctness property: if a swarm $\S$ consists of machines projected onto
the roles of the same swarm protocol $\G$ (such as $\Gn{Warehouse}$ in \Cref{fig:Warehouse}), then the swarm $\S$ enjoys \emph{eventual
fidelity} to $\G$ \cite{DBLP:conf/ecoop/KuhnMT23}. Intuitively, this means that, once every event propagates to every
machine in the swarm $\S$, all machines will reach a consensus on which events reflect a transition specified in $\G$,
and how they conform to an execution of $\Gn{G}$. %
This property also holds if, to increase reliability, the swarm includes multiple replicas of machines projected from the same role (e.g., in the example above, multiple machines may play the forklift role $\rolename{FL}$). %
The formulation and proof of eventual fidelity are quite challenging, as they must take into account the shape of $\Gn{G}$, the swarm event propagation mechanics, and the \emph{subscriptions} of each machine (i.e., which events it sees).

\subparagraph{Our Objective: Compositional Design and Verification of Swarm Protocols.}
The previously-published theory \cite{DBLP:conf/ecoop/KuhnMT23} and implementation \cite{actyx,10.1145/3597926.3604917} of swarm protocols only support \emph{monolithic} swarm design and implementation: i.e., they do not support developing a large and complex
swarm by combining modular, reusable swarm protocols and swarms that are designed, implemented, and verified independently.

Consider the (simplified) protocol  $\Gn{Factory}$ shown in  \Cref{fig:Factory}  where a transport ($\rolename{T}$) $\cmdname{request}$s and picks a part, which is then utilised by {an assembly robot
($\rolename{A}$)} to $\cmdname{build}$ a car (emitting an event of type $\lgtn{car}$). Notably, the $\Gn{Factory}$ swarm protocol only specifies $\rolename{T}$'s events $\lgtn{partReq}$ and $\lgtn{partOK}$, and does not specify how the part is obtained (i.e., what happens ``in between'' $\rolename{T}$'s events); such a detail could be part of another swarm protocol --- such as $\Gn{Warehouse}$ (\Cref{fig:Warehouse}), where $\rolename{T}$ \ecoop{gets the part from a forklift} $\rolename{FL}$. Intuitively, the $\Gn{Factory}$ protocol is a high-level view of a production process, whereas the composition of $\Gn{Factory}$ and $\Gn{Warehouse}$ describes the complete process where both protocols run concurrently with the transport role acting as an interface between them.

Now, suppose that the $\Gn{Factory}$ protocol changes into an updated swarm protocol $\Gn{Factory'}$ (where, e.g., the assembly robot $\rolename{A}$ coordinates with other factory devices): it would be desirable to compose $\Gn{Factory'}$ and $\Gn{Warehouse}$ (unchanged), obtaining a complete, updated production process protocol. Correspondingly, the machines projected from the $\Gn{Warehouse}$ swarm protocol should be \emph{reusable:} for example, it should be possible to deploy the implementation of the door role $\rolename{D}$ from $\Gn{Warehouse}$ (shown in \Cref{fig:code_machine_door}) into any swarm that is ``compatible'' with $\Gn{Warehouse}$, without requiring manual changes depending on whether other machines implement the protocol $\Gn{Factory}$ or $\Gn{Factory'}$.

Unfortunately, this scenario cannot be addressed with swarm protocol results in literature. There is no method to design $\Gn{Factory}$ and $\Gn{Warehouse}$ as separate protocols and compose them at a later stage: a whole $\Gn{FactoryWithWarehouse}$ swarm protocol must be designed at once, and then used to project swarm machines. Moreover, if a detail of the ``factory'' part of the $\Gn{FactoryWithWarehouse}$ swarm protocol changes and leads to an updated protocol $\Gn{FactoryWithWarehouse'}$, the change may also impact the machines of the ``warehouse'' part --- e.g., the machine obtained by projecting $\Gn{FactoryWithWarehouse'}$ onto role $\rolename{D}$ may differ from the projection of $\Gn{FactoryWithWarehouse}$ onto $\rolename{D}$ in non-trivial ways, \ecoop{meaning that the implementation of the $\rolename{D}$oor machine in~\Cref{fig:code_machine_door} may not be reusable and must be manually revised}. %
As a consequence, developing new swarms, or adapting existing swarms to new requirements, is a challenging and time-consuming endeavour. 
It would be highly desirable to have methods and tools to assemble libraries of swarm protocols and corresponding ready-to-use swarm machines that can be safely and easily composed, reused, and deployed to tackle new scenarios. In this work, we introduce such methods and tools.

\subparagraph{Contributions and Outline of the Paper.}
We present a novel theory and implementation of compositional swarm design, development, and verification.
Our approach is based on \emph{interfacing roles}:
given $n$ swarm protocols $\G_1, \ldots, \G_n$ with suitable roles acting as interfaces, we define the composition $\G_1 \conc \ldots \conc \G_n$.
Also, given $n$ swarms $\S_{1}, \ldots, \S_{n}$ whose machines are projected from $\G_1, \ldots, \G_n$, we introduce a method to compose $\S_{1}, \ldots, \S_{n}$.
This allows for creating large and complex swarms by reusing and composing machines that were previously projected from the individual simpler swarm protocols $\G_1, \ldots, \G_n$.
The details behind this broad intuition are quite subtle, due to the highly dynamic nature of swarm protocols, swarms, and their properties %
-- and to achieve these results, we develop novel swarm-specific notions of \emph{well-formedness} and \emph{compositionality}.

In \Cref{sec:overview} we provide an overview of our approach.
\Cref{sec:swarm-protocols-composition} presents our new swarm protocol composition and well-formedness, with a method for computing
well-formed subscriptions. \Cref{sec:projecting-machines-comp-swarm-proto} defines a projection of swarm protocols onto machines,
a \emph{branch-tracking} semantics for such machines, and proves their eventual fidelity to well-formed swarm protocols. 
\Cref{sec:swarmcomposition} presents our results on correctly composing swarms which realise different
swarm protocols.
In \Cref{sec:implementation} we discuss the companion software artifact of this paper \cite{furbach_2026_18459720}: %
an extension of the open source Actyx toolkit that implements the theory of \Crefrange{sec:swarm-protocols-composition}{sec:swarmcomposition}
to allow the practical compositional development of swarms; %
we compare experimental results on different strategies for computing the subscriptions of composed swarms.
We discuss related work in \Cref{sec:related work} and conclude in \Cref{sec:conclusion}.

\section{Overview and Main Results}
\label{sec:overview}

Before presenting our development, in
\cref{sec:overview:swarms} we give an intuitive account of the swarm framework
in~\cite{DBLP:conf/ecoop/KuhnMT23}.
Then, in \Cref{sec:overview:approach} we
outline our approach and summarise our main results.

\subsection{Swarms and Swarm Protocols}
\label{sec:overview:swarms}

As mentioned in \Cref{sec:intro}, a \emph{swarm} is an ensemble of
\emph{machines} that can make progress independently without requiring
up-to-date global information or an always-active network connection.
Each machine $\afish$ maintains a \emph{local log} $\lgname{l}_{\afish}$ of
\emph{events} received from the surrounding swarm.
These events \ecoop{are unique} and have a globally agreed-upon total
order~\cite{burckhardt2014principles} (usually defined in a coordination-free manner using timestamps,
which we leave implicit) and are used by $\afish$ to update its state when
newly arriving events alter $\lgname{l}_{\afish}$.

A \emph{swarm protocol $\G$} provides a
global view of the intended interactions that should take place during the swarm
execution.
To ensure correct interaction at run-time, each swarm machine should play a
specific \emph{role $\role$} in $\G$.
To this end, a machine can be obtained by \emph{projecting} $\G$ onto $\role$; the
projection operation is based on the \emph{subscription} of role $\role$, i.e.,
which event types a machine playing role $\role$ observes while running.
E.g., the machine in \Cref{fig:Door} is obtained by projecting
the $\Gn{Warehouse}$ protocol in \Cref{fig:Warehouse} onto role
$\rolename{D}$ using a subscription that associates $\rolename D$ with
all the event types of the protocol except $\lgtn{pos}$ (emitted by
$\rolename{FL}$).
The event emission $\lgtn{closingTime}!$ in state $\state 0$ adds an event of type
$\lgtn{closingTime}$ to the machine's local log and
leaves
the machine in state $\state 0$, whereas accepting
an event of type $\lgtn{closingTime}$ from the log, denoted as
$\lgtn{closingTime}?$, transitions the machine to state $\state 3$.

For an intuitive presentation of the execution of a swarm realising
the $\Gn{Warehouse}$ protocol, let us consider a possible run of a
swarm of three machines ---$\M_\rolename{T}$, $\M_\rolename{FL}$, and
$\M_\rolename{D}$--- respectively implementing roles $\rolename{T}$,
$\rolename{FL}$, and $\rolename{D}$:
\begin{enumerate}%
\item Machine $\M_\rolename{T}$ performs a $\cmdname{request}$ and emits the event $\lgname{partReq_1}$ (of type $\lgtn{partReq}$), 
  which it adds to its local log $\lgname l_\rolename{T}$.
\item This event is asynchronously propagated to the local log $\lgname{l}_\rolename{FL}=\lgname{partReq_1}$ of $\M_\rolename{FL}$, thus making
   $\M_\rolename{FL}$ reach state 1 from where it delivers the part
	 and emits $\lgname{pos_X}$ (of type $\lgtn{pos}$).
  \item The event $\lgname{pos_X}$ reaches $\M_\rolename{T}$, which picks up the part 
	 and emits the event $\lgname{partOK_1}$ (of type $\lgtn{partOK}$).
  \item These events propagate to machine $\M_\rolename{D}$, who adds them to its local log
  $\lgname{l}_\rolename{D} =\mklog{\lgname{partReq_1},\lgname{pos_X},\lgname{partOK_1}}$.
  It then closes the door emitting the event $\lgname{8PM}$ (of type $\lgtn{closingTime}$).
\end{enumerate}

Remarkably, when deploying a swarm, multiple machines may assume the same role
in $\G$, in a way that generalises \emph{replicated state
machines}~\cite{DBLP:journals/csur/Schneider90}.
In our $\Gn{Warehouse}$ protocol for instance, other machines may play role
$\rolename{FL}$ in addition to $\M_\rolename{FL}$. Such a machine may concurrently deliver another part by emitting an event, say $\lgname{pos_Y}$;
eventually, all events propagate to all machines, and all machines
reach the same log.
For instance, the  log
	 $\mklog{\lgname{partReq_1},\lgname{pos_X},\lgname{pos_Y},\lgname{partOK_1},\lgname{8PM},%
   }$
  represents a case where event $\lgname{pos_Y}$ is emitted after
  event $\lgname{pos_X}$; therefore $\lgname{pos_Y}$ will be ignored
  by all machines.

Following~\cite{DBLP:conf/ecoop/KuhnMT23}, the key correctness
property of a swarm $\S$ with respect to a protocol $\G$ is
\emph{eventual fidelity}.
Intuitively, eventual fidelity means that once the events emitted by the
machines in $\S$ propagate to the whole swarm, all machines in $\S$
reach a consensus corresponding to a valid execution of $\G$.
Notably, \cite{DBLP:conf/ecoop/KuhnMT23} proves that eventual fidelity is guaranteed by construction whenever $\G$ is
\emph{well-formed} and the machines in $\S$ are correct projections of the roles in $\G$;
in other words, the swarm $\S$ does not need any runtime orchestration to be eventually faithful to $\G$.

\subsection{Achieving Compositional Swarms}
\label{sec:overview:approach}

Unfortunately, as discussed in \Cref{sec:intro}, the theory and results in~\cite{DBLP:conf/ecoop/KuhnMT23} do not feature any
form of compositionality and do not support modular swarm development.
In this work, our key contribution is a new approach enabling the compositional design,
implementation and verification of swarms.
This is a challenging endeavour that we tackle by introducing a brand new
notion of well-formedness that ($i$) ensures eventual fidelity of a swarm to a protocol,
and ($ii$) is preserved by composition (unlike the well-formedness in~\cite{DBLP:conf/ecoop/KuhnMT23}).
In \cref{fig:framework} we provide an overview of our approach
showing how its main constructions are related and how they interact;
the overview leverages the $\Gn{Warehouse}$ and $\Gn{Factory}$ composition use cases in \Cref{sec:intro}.

\tikzset{
  mycallout/.style={
	 fill=yellow!20, %
	 align=center,
	 rectangle callout
  },
  box/.style={
	 align = center,
	 text width = 2.5cm
  },
  merge/.style={
	 align = center,
	 fill=teal!30,
	 circle,
	 font = \tiny,
	 inner sep =1pt
  },
}
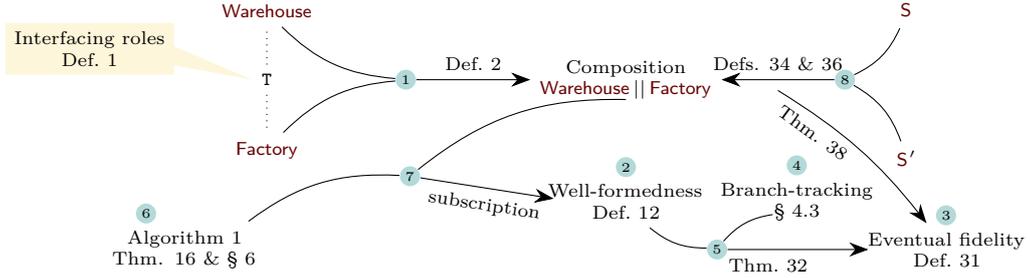
\begin{figure}[t]
  \centering
  \begin{tikzpicture}[node distance = .5cm and 1.5cm, >={Stealth[scale=1.5]}]
	 \node(G) {$\Gn{Warehouse}$};
	 \node(S) [right = 7.5cm of G] {$\S$};
	 \node(R) [below = of G] {$\rolename T$};
	 \node(G') [below = of R] {$\Gn{Factory}$};
	 \node(S') [below = of S, yshift = -1cm] {$\S'$};
	 \node(dummy) [merge,right = of R] {1};
	 \node(comp) [box, right = of dummy, inner sep = 0pt] {Composition \\ $\Gn{Warehouse} \conc \Gn{Factory}$}; %
	 \node(wf) [box, below = of comp, yshift = -.5cm] {Well-formedness \\ \Cref{def:wf}};
	 \node(bt) [box, right = -.5cm of wf] {Branch-tracking \\ \Cref{sec:branch-tracking-machines}};
	 \node(dummy2) [merge, below  = of wf.center, xshift = 1.2cm] {5};
	 \node(ef) [box, right = 1.5cm of dummy2] {Eventual fidelity \\ \Cref{def:eventual-fidelity}};
	 \node(alg) [box, left = of dummy2, xshift=-4cm] {Algorithm 1\\Thm. \ref{lem:subscription_composition_WF} \& \cref{sec:implementation}};
	 \node(dummy3) [merge, above right = of alg] {7};
	 \node(dummy4) [merge, right = 5.5cm of dummy] {8};

	 \node[mycallout, left = .5cm of G, yshift = -.5cm, callout absolute pointer={(R.west)}] {Interfacing roles\\\Cref{def:swarm-proto-interface}};
	 \node[merge, above = -1pt of wf] {2};
	 \node[merge, above = -1pt of ef] {3};
	 \node[merge, above = -1pt of bt] {4};
	 \node[merge, above = -1pt of alg, xshift = -.5cm] {6};

	 \path[dotted] (R) edge	 (G);
	 \path[dotted] (R) edge (G');
	 \path[] (G) edge[bend right] (dummy);
	 \path[] (G') edge[bend left] (dummy);
	 \path[] (S) edge[bend left] (dummy4);
	 \path (S') edge[bend right] (dummy4);
	 \path[->] (dummy) edge node[above]{\Cref{def:swarm-proto-composition}} (comp);
	 \path[] (wf) edge[bend right] (dummy2);
	 \path[] ($(bt) + (-9pt,-5pt)$) edge[bend right] (dummy2);
	 \path[->] (dummy2) edge node[below, near start, xshift = 2pt] {Thm.~\ref{thm:eventual_fidelity}} ($(ef.west) + (10pt,0)$);

	 \path[draw] ($(comp.south) + (-0pt, 0)$) edge[bend right] (dummy3);
	 \path[draw] ($(alg.north east) +(-15pt,0)$) edge[bend left] (dummy3);
	 \path[->] (dummy3) edge[sloped] node[below]{subscription} ($(wf.center) + (-27pt,1pt)$);
	 \path[->] (dummy4) edge node[above] (X) {Defs. \ref{def:machine-composition} \& \ref{def:machine-adaptation}} (comp);
	 \path[->] ($(X) + (0pt,-10pt)$) edge[sloped, bend left=10] node[below, near start] {Thm. \ref{lem:projection_composition_correct}} (ef);
  \end{tikzpicture}
  \caption{A synoptic view of our approach to swarm composition, with the main constructions and their dependencies.}
  \label{fig:framework}
\end{figure}

\begin{figure}[t]
  \centering
  \scalebox{\ifbool{ecoop}{0.75}{1}}{
    \begin{tikzpicture}[gt, scale=.8, transform shape, label distance=1mm, node distance=3mm and 22mm]
      \tikzset{pla/.style={circle,draw,fill=red!30}}
      \node (0) [pla]{$0\conc 0$};
      \node (begin) [left =0.5cm of 0] {};
      \node (1) [pla, right = of 0] {$1\conc 1$};
      \node (7)  [pla, below = of 1] {$3\conc 0$};
      \node (2) [pla,, right = of 1] {$2\conc 1$};
      \node (3)  [pla, right = of 2] {$0\conc 2$};
      \node (middle)  [right = of 3] {};
      \node (4)  [pla, above = of middle] {$0\conc 3$};
      \node (5)  [pla, below = of middle] {$3\conc 2$};
      \node (6)  [pla, right = of middle] {$3\conc 3$};

      \draw[->] (begin) to node[sloped, anchor=center, above=1mm]{}  (0);
      \draw[->] (0) to node[sloped, anchor=center, above=0.5mm]{$ \cstmcmd{T}{partReq}$}  (1);
      \draw[->] (1) to node[sloped, anchor=center, above=0.5mm]{$\cstmcmd{FL}{pos}$}  (2);
      \draw[->] (2) to node[sloped, anchor=center, above=0.5mm]{$ \cstmcmd{T}{partOK}$}  (3);
      \draw[->] (0) to node[sloped, anchor=center, below=0.5mm]{$  \cstmcmd{D}{closingTime}$}  (7);
      \draw[->] (3) to node[sloped, anchor=center, above=0.5mm]{$ \cstmcmd{A}{car}$}  (4);
      \draw[->] (3) to node[sloped, anchor=center, below=0.5mm]{$ \cstmcmd{D}{closingTime}$}  (5);
      \draw[->] (4) to node[sloped, anchor=center, above=0.5mm]{$  \cstmcmd{D}{closingTime}$}  (6);
      \draw[->] (5) to node[sloped, anchor=center, below=0.5mm]{$ \cstmcmd{A}{car}$}  (6);
  \end{tikzpicture}}
  \vspace{-2mm}%
  \caption{The composition  $\Gn{Warehouse} \conc \Gn{Factory}$ of  $\Gn{Warehouse}$ and $\Gn{Factory}$.}
  \label{fig:Composition}%
\end{figure}

\newcommand{\fmk}[1]{\tikz{\node[merge]{#1}}}

The pivotal notion upon which we develop our theory is
\emph{interfacing roles}, i.e., roles shared by protocols which
\quo{agree} on common event types.
For instance, the transport role $\rolename{T}$ is an interface
between the $\Gn{Warehouse}$ and $\Gn{Factory}$.
Our composition operator $\_ \conc \_$ (\fmk 1 in
\cref{fig:framework}) essentially acts as a synchronised product of
the behaviour of the protocols involved in the composition.
For instance, consider the composition
$\Gn{Warehouse} \conc \Gn{Factory}$ of $\Gn{Warehouse}$ and
$\Gn{Factory}$ shown in \cref{fig:Composition}.
All actions performed by the interfacing role $\rolename T$ (i.e.,
$\cstmcmd T {partReq}$ and $\cstmcmd T {partOK}$) are synchronised
 between $\Gn{Warehouse}$ and $\Gn{Factory}$ while actions of non-interfacing
roles (i.e., $\cstmcmd{A}{car}$ and $\cstmcmd{D}{closingTime}$)
are interleaved.
Note that ($i$) the interleaving of $\cstmcmd{A}{car}$ and
$\cstmcmd{D}{closingTime}$ shows that they are concurrent and ($ii$)
the loop in $\Gn{Warehouse}$ disappears from the composed protocol
since $\Gn{Factory}$ requires the delivery of just one part.

We develop a novel \emph{compositional} notion of well-formedness
for swarm protocols (\fmk 2).
This notion provides a sufficient condition for
statically verifying that a swarm is
{eventually faithful} to a protocol
(\fmk 3).
As part of this development, we refine the
semantics of machines w.r.t.~\cite{DBLP:conf/ecoop/KuhnMT23} by introducing a new \emph{branch-tracking} mechanism. %
Specifically, in our new machine semantics, each emitted event contains a pointer %
to the most recent \emph{branching}, \emph{joining}, or \emph{looping}
event (\cref{def:branching-event}) that precedes
it. This pointer establishes an \emph{event causality chain} and allows machines
to ignore events which are invalid because they are not caused by the expected branching/joining/looping event
(this is illustrated in \Cref{ex:bt-swarm-warehouse}).
Such a causality chain was not present nor needed in \cite{DBLP:conf/ecoop/KuhnMT23} 
-- but it is necessary for swarms to work correctly with our 
compositional well-formedness, which removes restrictions that were assumed in \cite{DBLP:conf/ecoop/KuhnMT23}.
Therefore, the new branch-tracking semantics allows us generalise the correctness results in~\cite{DBLP:conf/ecoop/KuhnMT23};
moreover, it allows us to reduce the size of subscriptions needed to correctly track protocol
choices (\fmk 4), thus increasing the efficiency of the swarm.
Crucially, %
well-formedness and branch tracking ensure that a swarm of machines projected from a protocol $\G$ is eventually faithful to $\G$ (\fmk 5).

Another contribution of this paper is the definition and implementation
of \Cref{alg:comp-swarm-subscriptions} (\fmk 6) to compute 
well-formed subscriptions for the composition of a set of input swarm protocols; %
crucially, \Cref{alg:comp-swarm-subscriptions} computes the subscriptions based on the individual structures of the input protocols
(e.g., $\Gn{Warehouse}$ and $\Gn{Factory}$)
without computing their composition (e.g., $\Gn{Warehouse} \conc \Gn{Factory}$) as it can be exponentially larger than the input protocols
(\fmk 7).

Finally, we develop a composition and adaptation of machines in a swarm, %
{%
  in such a way that the composed swarm (\fmk 8) is eventually faithful to the corresponding composition of swarm protocols.%
}%

\hide{
We now summarise our main contributions.

\begin{result}[Well-formedness of composed protocols]\label{result:comp-wf}%
  We introduce \Cref{alg:comp-swarm-subscriptions}
  (p. \pageref{alg:comp-swarm-subscriptions}) that provides an
  effective procedure to compute a subscription guaranteeing the
  well-formedness of the composition of some deterministic
  protocols\footnote{A protocol is deterministic when the event type
	 of each emission univocally determine the continuation of the
	 protocol.} if they are such that ($i$) every
  event type is emitted by only one role in any of the protocols and
  ($ii$) that different subterms of the protocols emit a same event
  type only due to concurrency.

  \Cref{alg:comp-swarm-subscriptions} and
  \Cref{lem:subscription_composition_WF} showing the correctness of
  our algorithm constitute our first main technical contribution.
\end{result}

As in~\cite{DBLP:conf/ecoop/KuhnMT23}, a canonical implementation of a
swarm protocol can be mechanically derived from the protocol itself by
projection. We however adapted the computational model of machines
in~\cite{DBLP:conf/ecoop/KuhnMT23} to incorporate a
\emph{branch-tracking} mechanism, which reduces the number of events
that a machine must observe to track protocol choices. Specifically,
each emitted event contains a pointer (akin to a \emph{causal link})
to the most recent \emph{branching}, \emph{joining}, or \emph{looping}
event that precedes it, thereby establishing a causality chain. When
processing events, a machine ignores those that do not point to the
expected event, effectively discarding events corresponding to invalid
executions of the protocol.
Our second main technical result ensures that correct implementation
of well-formed protocols can be obtained by projection.

\begin{result}[Eventual fidelity of projected machines]
  \label{result:wf-implies-ef}%
  \def\mproj{$\G\projshort{\subscription}{\role}$}
  Our notion of projection (cf.  \Cref{def:swarm-proto-projection})
  produces a machine \mproj\ for a role $\role$ out of a protocol $\G$ and a
  given subscription.
  More precisely, \mproj\ can emit an even having event type $\evt$ if
  $\G$ can reach $\G'$ with a sequence of event types \mproj\ does not
  subscribes to and $\role$ emits $\evt$ from $\G'$; likewise, \mproj\
  observes any event of type $\evt$ if it subscribes to $\evt$ and $\G$ can reach $\G'$ with a
  sequence of event types \mproj\ does not subscribes to and $\G'$
  can emit $\evt$.

  A swarm is eventually faithful to a protocol
  (\Cref{def:eventual-fidelity}) if the subscriptions are such that
  each machine in the swarm can disregard all and only the events that
  are  inconsistent with all the logs actually generated by the
  protocol.

  \Cref{thm:eventual_fidelity} ensures that eventual fidelity holds for subscriptions for which
   a swarm realises a well-formed protocol.
\end{result}

The third technical result extends compositionality to
implementations: a composed protocol does not need to be
re-implemented from scratch; it can rather be (automatically)
assembled from the existing realisations of the constituent
protocols.
Nonetheless, an existing machine may require adaptation to
handle some additional events (i.e., those events computed by
\Cref{alg:comp-swarm-subscriptions}).  We write
$\mathcal{A}\!\left( \M (\G_1,\dots,\G_n)\right)$ for the adaptation
of some existing machine in the context of the composition of
$\G_1,\dots,\G_n$.  We remark that such adaptation can be
automatically computed since it requires to extend the structure of
$\M$ to observe more events, while preserving its internal logic (The
full details are provided in \Cref{sec:machine-composition}).

\hsl[to be finished]

\begin{result}[Compositionality of swarms]\label{result:comp-verif}%
  We introduce a notion of machine composition \Cref{def:machine-composition}
  that allows the emission of an event of type $\evt$ in a composed machine, say $\M$,

  \Cref{def:machine-adaptation}

  \Cref{lem:projection_composition_correct}%

  Let  $\{\Gn{G}_i\}_{i \in I}$ be a set of protocols, and  $\{\S_i\}_{i \in I}$  a set of swarms such that $\S_i$ implements $\G_i$.  Define $\S$ as the swarm consisting of $\mathcal{A}\!\left(\M_i^j, (\G_1\mydots \G_n)\right)$ with  $\M_i^j \in \S_i$ for $i \in \{ 1\ldots n \}$.
  Then, $\S$ is eventually faithful to $\G_1 \conc \ldots \conc \G_n$.
\end{result}

\eMcomm[why are we repeating this?]{We implement the results above in the companion software artifact of this paper: a customised version of the Actyx toolkit~\cite{actyx}
that we discuss in \Cref{sec:implementation}.
}%
\todo{Maybe say something about the efficiency of the new mechanism}
}

\section{Composing Swarm Protocols}
\label{sec:swarm-protocols-composition}

In this section we introduce the composition operation on swarm protocols
(\Cref{sec:swarm-protocols-composition:definitions})
and our new notion of \emph{well-formedness} (\Cref{sec:well-formedness}).
We then provide an effective method for computing subscriptions that
guarantee the well-formedness of a protocol with respect to those
subscriptions (\Cref{sec:swarm-protocols-composition-computing-sub}
and \Cref{lem:subscription_composition_WF}).
First, as in~\cite{DBLP:conf/ecoop/KuhnMT23}, we formalise a \emph{swarm protocol}
(or simply \emph{protocol}) $\agt$ as a regular term from the
following coinductive grammar:\footnote{%
  A term is regular if it consists of finitely many \emph{distinct}
  subterms. This condition ensures that the language generated by the
  co-inductive grammar is finitely representable either using the
  so-called \quo{$\mu$ notation}~\cite{Pierce02} or as solutions of
  finite sets of equations~\cite{Courcelle83}. Also, there is a
  correspondence between these structures and finite-state automata;
  the interested reader is referred to~\cite{Courcelle83}.}%

\smallskip%
\centerline{\(
  \agt \;\stackrel{\text{co}}{\bnfdef}\; \sum_{i\in I} \cmdindx{i} \cdot \agt[i]
  \qquad \text{where $\arole[1]$, $\arole[2]$, etc. range over roles}
\)}%
\smallskip%

\noindent%
The protocol $\agt$ above specifies that (any machine implementing) role
$\arole[i]$ is expected to emit an event of type $\aeventtype[i]$ (for some $i \in I$)
and let the protocol continue as $\agt[i]$.
The summation operator $\sum$ represents a nondeterministic choice, so the order of summands is immaterial.
We write $\agt=\zero$ if $I$ is empty.
For readability we use the infix notation for summations; for
instance, for the protocol $\Gn{Warehouse}$ in
\Cref{fig:Warehouse} we write:

\smallskip%
\centerline{\(
\Gn{Warehouse}=
\cstmcmd{T}{partReq} \cdot \cstmcmd{FL}{pos} \cdot
\cstmcmd{T}{partOK} \cdot \Gn{Warehouse} \;+\; \cstmcmd{D}{closingTime}
\cdot \zero
\)}%

\subsection{Protocol Interface and Composition}
\label{sec:swarm-protocols-composition:definitions}

We introduce some convenient
notation: given a protocol
$\agt = \sum_{i\in I} \cmdindx{i} \mydot \agt[i]$, we write
$\cmd \in \agt$ if $\cmd$ occurs anywhere in $\agt$ (and likewise for
$\evt \in \agt$, and $\role \in \agt$). %
{%
}%

In \Cref{def:swarm-proto-interface} below we formalise when two protocols \emph{interface};
intuitively, this happens when event types common to the protocols
are emitted by common roles only.

\begin{definition}[Protocol interface]
  \label{def:swarm-proto-interface}%
  Two protocols $\agt$ and $\agt'$ are \emph{interfacing} iff for any $\cmd\in \G$ and
  $\cstmcmd{R'}{t} \in \G'$, it holds that $\role=\role'$. The  set of \emph{interfacing roles} is
  $\left\{\role \;\middle|\; \role \in \G \land \role \in
  \G'\right\}$. %
  When $\cstmcmd{R}{t} \in \G$ and $\role$ is an interfacing role, we say $\evt$ is an \emph{interfacing event type}.
\end{definition}

\Cref{def:swarm-proto-composition} below formalises a protocol composition
operator $\_ \conc \_$ that allows two protocols to run independently (items 1 and 2) -- except for the
\emph{synchronisations} imposed by the events from their interfacing roles, if any (item 3).

\begin{definition}[Protocol composition]\label{def:swarm-proto-composition}%
  Given protocols
  $\agt = \sum_{i\in I} \cmdindx{i}\mydot \agt[i]$ and
  $\agt' =\sum_{j \in J} \cstmcmd{\role'_j}{\evt'_j}\mydot \G'_j$ as well as a
  set of roles $\roleset$, we coinductively define:
\[
\begin{array}{rcl@{\hspace{4em}}r}
	\agt \concM{\roleset} \agt' &\stackrel{co}{\define}& \sum \{ \cmdindx{i}\cdot (\agt[i] \concM{\roleset} \agt') \sst i \in I, \role_i \notin \roleset\} & (1) \\
	& + &\sum \{ \othercstmcmd{\role'_j}{\evt'_j}\cdot (\agt \concM{\roleset} \G'_j) \sst j \in J, \role'_j \notin \roleset\} & (2) \\
	& + &\sum \{ \cmdindx i \cdot (\agt[i] \concM{\roleset} \G'_j) \sst i \in I, j \in J, \cmdindx i = \othercstmcmd{\role'_j}{\evt'_j}\} & (3)
\end{array}
\]
We let $\_\conc\_$ be left-associative and simply write $\G \conc \G'$
instead of $\agt \concM{\roleset} \agt'$
when $\roleset$ is the set of roles that $\G $ and $\G'$ interface on.%
\end{definition}
By construction, $\_\conc\_$ is an idempotent and commutative internal operation on
protocols.
Note that we do not introduce new protocol-level syntax for composed
protocols: more precisely, the result of $\G \conc \G'$ is a
protocol obtained by interleaving event types that $\G$ and
$\G'$ do not share while synchronising shared event types.\footnote{%
  This design decision allows us us to maintain backward compatibility
  with existing tools (see \Cref{sec:implementation}).
}

\begin{example}\label{ex:composition}
  The swarm protocol composition $\Gn{Warehouse} \conc \Gn{Factory}$ informally depicted
  in \Cref{fig:Composition} is obtained by expanding
  $\Gn{Warehouse} \concM{\{\rolename{T}\}} \Gn{Factory}$:
  \begin{align*}
	 \Gn{Warehouse} \conc \Gn{Factory} &\;=\;\gsumprefixnewnew[][\lgtn{partReq}][\rolename{T}]\gtpref \gsumprefixnewnew[][\lgtn{pos}][\rolename{FL}]\gtpref\gsumprefixnewnew[][\lgtn{partOK}][\rolename{T}]\gtpref  & &\\
												  & \qquad\quad\big( \gsumprefixnewnew[][\lgtn{car}][\rolename{A}]\gtpref\gsumprefixnewnew[][\lgtn{closingTime}][\rolename{D}]\gtpref\zero +\,\gsumprefixnewnew[][\lgtn{closingTime}][\rolename{D}]\gtpref\gsumprefixnewnew[][\lgtn{car}][\rolename{A}]\gtpref\zero \big)&\\
												  & \quad\;\;+\; \gsumprefixnewnew[][\lgtn{closingTime}][\rolename{D}]\gtpref\zero
    \ \\[-1.4cm]
  \end{align*}
\end{example}

In general, when two swarm protocols are composed, only those interfacing event
types emitted in both protocols in the same order are included in the
composition; if interfacing event types are missing in a component protocol or
emitted in different orders, then the composition excludes behaviour present in
the component protocols, as shown in \Cref{ex:behaviour_restriction}.\footnote{%
  Such behaviour restrictions are not consequential for our theoretical
  development. If undesired, they can be identified by checking whether the
  relevant transitions of $\G$ and $\G'$ remain present in %
  $\G \conc \G'$.%
}%

\begin{example}[Protocol composition and behaviour restrictions]\label{ex:behaviour_restriction}
  Consider the swarm protocol $\Gn{Warehouse}$ in \Cref{fig:Warehouse}, and\; %
  \(
  \Gn{Factory'} =
  \gsumprefixnewnew[][\lgtn{partOK}][\rolename{T}]\gtpref\gsumprefixnewnew[][\lgtn{partReq}][\rolename{T}]
  \gtpref \gsumprefixnewnew[][\lgtn{car}][\rolename{A}] \gtpref \zero
  \).

  By \Cref{def:swarm-proto-composition}, the composition $\Gn{Warehouse} \conc \Gn{Factory'}$
  yields\;
  $\gsumprefixnewnew[][\lgtn{closingTime}][\rolename{D}]\gtpref\zero$
  \;because the interfacing role $\rolename T$ in $\Gn{Factory'}$ emits
  an event of type $\lgtn{partOK}$ before those of type $\lgtn{partReq}$
  -- whereas $\Gn{Warehouse}$ follows the opposite order.
  As another example, letting $\roleset = \{\rolename{T}\}$, we have
  $\Gn{Warehouse} \concM{\roleset} \zero =
  \gsumprefixnewnew[][\lgtn{closingTime}][\rolename{D}]\gtpref\zero$
  -- because, by \Cref{def:swarm-proto-composition}, $\roleset$
  \quo{blocks} the actions of role $\rolename{T}$ in $\Gn{Warehouse}$ as it cannot
  synchronise with $\zero$.
\end{example}

\subsection{Well-Formedness of Swarm Protocols and Subscriptions}
\label{sec:well-formedness}

We now introduce our new notion of \emph{well-formedness} (\Cref{def:wf}),
which builds upon the following three properties:
\emph{confusion-freeness} (\Cref{def:confusion-freeness}),
\emph{causal consistency} (\Cref{def:causalconsistency}), and
\emph{determinacy} (\Cref{def:determinacy}).
Although these three properties have the same name as those
in~\cite{DBLP:conf/ecoop/KuhnMT23}, their definitions deviate
significantly to accommodate compositionality. %
Crucially, our well-formedness is more general than that
of~\cite{DBLP:conf/ecoop/KuhnMT23}
(except in some minor corner cases discussed later in \Cref{footnote:outliers}),
hence we do not lose expressiveness.

Before proceeding, we introduce some additional notation. %
Given $\agt =\sum_{i\in I} \cmdindx{i} \mydot \agt[i]$, the predicate
$\agt \xrightarrow{\cstmcmd{\role_i}{\evt_i}} \agt[i]$ holds for all $i \in I$
and indicates that some role $\role_i$ can emit an event of type $\evt_i$ causing the swarm protocol $\G$ to continue as $\G_i$.
We write $\agt \xrightarrow{\evt} \agt$ instead of $\agt \xrightarrow{\cstmcmd{\role}{\evt}} \agt$ when the role $\role$ is immaterial.
We sometimes write $\agt \xrightarrow{\cstmcmd{\role}{\evt}}$ when
there is $\agt'$ such that
$\agt \xrightarrow{\cstmcmd{\role}{\evt}} \agt'$, and similarly for
$\agt \xrightarrow{\evt}$. %
We write $\agt \subterm \agt'$ if $\agt$ is a \emph{subterm} of $\agt'$. 
We shorten
$\G \xrightarrow{\evt_1} \cdots \xrightarrow{\evt_n} \G'$
to $\G \xrightarrow{\evt_1 \ldots \evt_n} \G'$; therefore,
given a sequence of event types $\lgt = \evt_1 \ldots \evt_n$, we write $\G \xrightarrow{\lgt} \G'$
if $\G \xrightarrow{\evt_1 \ldots \evt_n} \G'$.

By \Cref{def:confusion-freeness} a protocol is \emph{confusion-free} if:
\eqref{item:confusion-free:evt-unique-role} every event type is
emitted by only one role;
\eqref{item:confusion-free:evt-deterministic} event types are deterministic,
i.e., an event emision leads to only one possible continuation;
and \eqref{item:confusion-free:concurrecy} the same event type can be
emitted by different subterms only due to concurrency, i.e., any
subterm accounts for a different interleaving of independent event
types.\footnote{%
  Item~\ref{item:confusion-free:concurrecy} of
  \Cref{def:confusion-freeness} can be expensive to check on an
  arbitrary $\G$, but is immediate when the appropriate
  $\G_1,\ldots,\G_n$ are given upfront (which is the common case when
  composing protocols).%
}

\begin{definition}[Confusion-Freeness]
  \label{def:invariance}%
  \label{def:confusion-freeness}
  A protocol $\agt$ is \emph{confusion-free} if:
    \begin{enumerate}
	 \item\label{item:confusion-free:evt-unique-role} For every $\evt \in \agt$, if
		$\cstmcmd{R}{\evt}, \cstmcmd{R'}{\evt} \in \G$, then
		$\rolename{R} = \rolename{R'}$.
	 \item\label{item:confusion-free:evt-deterministic} For every
	 	$\agt[1],\agt[2],\agt[3] \subterm \agt$ and
		$\evt \in \agt$, if
		$\agt[1] \xrightarrow{\evt} \agt[2]$ and
		$\agt[1] \xrightarrow{\evt} \agt[3]$ then $\agt[2] = \agt[3]$.
	 \item\label{item:confusion-free:concurrecy} There are
		$\G_1\mydots \G_n$ for $n \geq 1$ such that
		$\G=\G_1 \conc \ldots \conc \G_n$ and for every $\evt \in \agt$
		and $i \in \{1, \ldots, n\}$, there is at most one subterm
		$\G' \subterm \G_i$ such that $\G' \xrightarrow{\evt}$.
    \end{enumerate}
\end{definition}

In \Cref{def:concurrency} below we characterise four categories of event types in a confusion-free protocol:
(1) \emph{concurrent} event types (i.e., those that are causally independent, see \Cref{fig:joining-event-types});
(2) \emph{branching} event types (i.e., those that introduce decision points, see \Cref{fig:branching-event-types});
(3) \emph{joining} event types (i.e., those that synchronize multiple concurrent branches, see \Cref{fig:joining-event-types});
(4) \emph{looping} event types (i.e., those that are repeated due to loops, see \Cref{fig:looping-event-types}).
Notice that these categories are not mutually exclusive:
for instance, an event type can be both joining and branching, or neither.

\begin{definition}[Concurrent, Branching, Joining, and Looping Event Types]
  \label{def:concurrency}
  \label{def:simple-event}\label{def:branching-event}\label{def:joining-event}\label{def:looping-event}%
	Let  $\agt$ be a confusion-free protocol. An event type $\evt$ is:
	\begin{enumerate}
	\item {\emph{Concurrent with a $\evt' \neq \evt$}} if there exist
	  $\agt[1], \agt[2] \subterm \agt$ such that
	  $\agt[1] \xrightarrow{\evt} \xrightarrow{\evt'} \agt[2]$ and
	  $\agt[1] \xrightarrow{\evt'} \xrightarrow{\evt} \agt[2]$.%
	  \footnote{Note that our definition of ``concurrent events'' together with confusion-freeness (\Cref{def:confusion-freeness}) ensure that for all $\G_a,\G_b \subterm \G$ and $\evt_a,\evt_b \in \G$ where $\evt_a,\evt_b$ are concurrent and $\G_a \xrightarrow{\evt_a}\xrightarrow{\evt_b}\G_b$, we also have $\G_a \xrightarrow{\evt_b}\xrightarrow{\evt_a}\G_b$.}%
	\item {\emph{Branching with a $\evt' \neq \evt$ at $\G_1 \subterm \G$ }}
	  if there are {$\agt[2],\agt[3] \subterm \agt$
		 with $\agt[1] \xrightarrow{\evt}\agt[2]$,
		 $\agt[1] \xrightarrow{\evt'}\agt[3]$ such that $\G_2 \neq \G_3$} and $\evt, \evt'$ are not
	  concurrent.
			\item \emph{Joining for $\evt'$ and $\evt''$ at $\G_3$} if there are
			$\agt[1], \agt[2], \subterm \G$ such that
			$\agt[1] \xrightarrow{\evt'} \agt[3]$,
			$\agt[2] \xrightarrow{\evt''} \agt[3]$, and
			$\agt[3] \xrightarrow{\evt}$ where $\evt'$ and $\evt''$ are concurrent,
			but neither of them is concurrent with $\evt$.
            \item \emph{Looping} if there are $\G' \subterm \G$ and $\lgt$ such that $\G' \xrightarrow{\lgt}\xrightarrow{\evt} \G'$.
	    \end{enumerate}
\end{definition}

\tikzset{
	 initial/.style={state,initial by arrow, initial text={}, initial where=left}
  }
\begin{figure}[tb]
\begin{minipage}[t]{0.35\textwidth}
  \centering
  \begin{tikzpicture}[gt, scale=.85, transform shape, label distance=1mm, node distance=1mm and 9mm]
	\node (0) [initial,pla]{};
	\node[state] (1) [pla, above right = of 0] {};
	\node (2)  [pla, below right = of 0] {};
	\node (3) [pla, below right = of 1] {};
	\node (4)  [pla, right = of 3] {};

	\draw[->] (0) to node[sloped, anchor=center, above]{$ \cstmcmd{R'}{t'}$}  (1);
	\draw[->] (1) to node[sloped, anchor=center, above]{$ \cstmcmd{R''}{t''}$}  (3);
	\draw[->] (0) to node[sloped, anchor=center, below]{$ \cstmcmd{R''}{t''}$}  (2);
	\draw[->] (2) to node[sloped, anchor=center, below]{$ \cstmcmd{R'}{t'}$}  (3);
	\draw[->] (3) to node[sloped, anchor=center, above]{$ \cstmcmd{R}{t}$}  (4);
  \end{tikzpicture}
  \vspace{-4mm}
  \captionof{figure}{$\lgtn{t'}$ and $\lgtn{t''}$ are concurrent; $\evt$ is joining for $\lgtn{t'}$ and $\lgtn{t''}$.}%
  \label{fig:joining-event-types}
\end{minipage}
\begin{minipage}[t]{0.29\textwidth}
  \centering
  \begin{tikzpicture}[gt, scale=.85, transform shape, label distance=1mm, node distance=1mm and 9mm]
	\node (0) [initial,pla]{};
	\node (1) [pla, above right = of 0] {};
	\node (2)  [pla, below right = of 0] {};

	\draw[->] (0) to node[sloped, anchor=center, above]{$ \cstmcmd{R}{t}$}  (1);
	\draw[->] (0) to node[sloped, anchor=center, below]{$ \cstmcmd{R'}{t'}$}  (2);
  \end{tikzpicture}
  \vspace{-4mm}
  \captionof{figure}{$\evt$ is branching with $\evt'$.}
  \label{fig:branching-event-types}
\end{minipage}\hfill
\begin{minipage}[t]{0.32\textwidth}
	\centering
	\begin{tikzpicture}[gt, scale=.85, transform shape, label distance=1mm, node distance=1mm and 9mm]
			\node (0) [initial,pla]{ };
			\node (1) [pla, right = of 0] {};

			\draw[->, bend left] (0) to node[sloped, anchor=center, above, bend left]{$ \cstmcmd{R}{t}$}  (1);
			\draw[->, bend left] (1) to node[sloped, anchor=center, below, bend right]{$ \cstmcmd{R'}{t'}$}  (0);
		\end{tikzpicture}
	\vspace{-4mm}
	\captionof{figure}{$\evt$ and $\evt'$ are looping.}
	\label{fig:looping-event-types}
\end{minipage}\hfill
\end{figure}

A subscriptions $\subscription$ (\Cref{def:subscription} below) specifies which roles observe which event types; 
\Cref{def:causalconsistency} then defines when a swarm protocol $\G$ is causally consistent w.r.t.~ a subscription $\subscription$.
Intuitively, causal consistency enforces two key requirements.
First, each role must subscribe to the event types it emits: this will be necessary
later, when we define how to \emph{project} machines from $\G$ based on $\subscription$.\footnote{%
  As mentioned in \Cref{sec:overview:swarms}, and later formalised in
  \Cref{def:operational-semantics-machines}, if a machine playing role $\role$
  emits an event of type $\evt$, that machine will \emph{not} change state unless it
  has an accepting transition for $\evt$. Subscribing $\role$ to $\evt$ ensures that the projection of a
  swarm protocol onto its roles (using \Cref{def:swarm-proto-projection}) will
  produce machines that can change state by accepting every event type they
  emit.%
} %
Second, to follow the causal dependencies specified in the protocol, roles must subscribe to the event types directly preceding
the events they emit.

\begin{definition}[Subscription]
  \label{def:subscription}%
  A \emph{subscription} $\subscription$ is a mapping from roles to sets of event types;
  i.e., $\subscription(\role)$ is the set of event types that role $\role$ subscribes to.
  {%
  }%
\end{definition}

\begin{definition}[Causal-consistent swarm protocol]
  \label{def:causalconsistency}
  A protocol \emph{$\G$ is causal-consistent for a subscription $\subscription$}
  if, for all $ \evt', \cstmcmd{\role}{\evt} \in \G$:
  \begin{enumerate}
  \item $\evt \in \subscription (\role)$ and
  \item $\evt' \in \subscription(\role)$ if $\evt'$ and $\evt$ are not
	 concurrent and there is $\G_1\subterm \G$ such that
	 $\G_1 \xrightarrow{\evt'} \xrightarrow{\cstmcmd{\role}{\evt}}$.
  \end{enumerate}
\end{definition}

\begin{example}[Causal consistency]\label{ex:causal-consistency}%
  Consider the $\Gn{Factory}$ protocol in \Cref{fig:Factory}.
  By causal consistency, the subscription $\subscription(\rolename{A})$
  of the $\rolename{A}$ssembly robot must include event types $\lgtn{car}$ and
  $\lgtn{partOK}$ -- respectively by the first and second condition of \Cref{def:causalconsistency}.
\end{example}

\Cref{def:determinacy} below characterises subscriptions of \emph{determinate}
protocols, namely those where each role $\role$ subscribes to the events that
are crucial for tracking causality in a swarm protocol $\G$.
Intuitively, depending on the category of an event type $\evt$ in $\G$ (branching, joining, or looping), \Cref{def:determinacy}
require that ($i$) role $\role$ must subscribe to any event types that are branching with $\evt$ (clause \emph{Branching}), ($ii$) $\role$ must
subscribe to 
{%
  any concurrent event types that immediatly precede $\evt$ %
}%
(clause \emph{Joining}), and ($iii$) each role $\role$ must observe each loop iteration,
hence there must be at least one event type in each loop that all roles subscribe to (clause \emph{Looping}).
\Cref{def:determinacy} relies on the set $\newroles[\evt][\G][\subscription]$, which is the set of roles that
subscribe to event types in $\G$ that causally depend on $\evt$ according to $\subscription$.
Formally,
\begin{align*}
	\newroles[\evt][\G][\subscription] = \{ \role \sst &
	\text{there are } n \geq 0, \; \evt_0, \ldots, \evt_n, \; \lgt_1, \ldots, \lgt_n  \text{ such that: }
	\\&\evt_0 = \evt, \; \G \xrightarrow{\evt_0} \xrightarrow{\lgt_1} \xrightarrow{\evt_1} \cdots \xrightarrow{\lgt_n} \xrightarrow{\evt_n}
	\text{ and }
	\\& \evt_n \in \subscription(\role), \text{ and }\evt_i, \evt_{i+1} \text{ not concurrent for all }0 \leq i < n \}.
\end{align*}

\begin{definition}[Determinate protocol subscriptions]\label{def:determinacy}%
  A protocol \emph{$\G$ is determinate for subscription $\subscription$} if,
  for all $\G' \subterm \G$ and $\evt, \role \in \G$
  \begin{description}
  \item[Branching:] If $\evt$ is branching with $\evt'$ at $\G'$ and
    $\role \in \newroles[\evt][\G']$, then
    $\evt,\evt' \in \subscription(\role)$.
  \item[Joining:] If $\evt$ is joining for $\evt'$ and $\evt''$ at $\G'$ and
    $\role \in \newroles[\evt][\G']$, then
    $\evt, \evt', \evt'' \in \subscription(\role)$.
  \item[Looping:] If $\G' \xrightarrow{\lgt} \G'$ for some non-empty $\lgt$
    then there are $\evt' \in \lgt$ and $\G'' \subterm \G'$ such that
    $\G'' \xrightarrow{\evt'}$ and $\evt' \in \subscription(\role)$ for all
    $\role \in \newroles[\evt'][\G'']$. We call $\evt'$ looping in $\subscription$.
  \end{description}
\end{definition}

\begin{example}[Branching vs.~looping events]
  In \Cref{fig:Warehouse}, the loop
  $ \Gn{Warehouse} \xrightarrow{\lgtn{partReq}} \xrightarrow{\lgtn{pos}}
  \xrightarrow{\lgtn{partOK}} \Gn{Warehouse} $
  can be exited by event type $\lgtn{closingTime}$, which branches with
  $\lgtn{partReq}$. Consequently, subscribing to $\lgtn{partReq}$ and
  $\lgtn{closingTime}$ is required by clause \emph{branching} of
  \Cref{def:determinacy} -- and the subscription to $\lgtn{partReq}$ also
  satisfies the clause \emph{looping}.
  In general, if a protocol contains a loop with a branching that exits the
  loop, then each branching event type that continues the loop is also a
  looping event type in the subscription (by \Cref{def:looping-event}) and
  subscribing to it satisfies clause \emph{looping} of \Cref{def:determinacy}.
\end{example}

We now have all the ingredients to formalise when a swarm protocol is well-formed.

\begin{definition}[Well-formed swarm protocols]
  \label{def:termination}
  \label{def:wf}%
  A protocol $\agt$ is \emph{$\subscription$-well-formed} if it is
  confusion-free, causal-consistent in $\subscription$, and determinate for
  $\subscription$
  (\Cref{def:confusion-freeness,def:determinacy,def:causalconsistency}).
\end{definition}

\begin{example}[Well-formed protocol and subscription]
  The protocol $\Gn{Warehouse} \conc \Gn{Factory}$ in \Cref{ex:composition}
  and \Cref{fig:Composition} is confusion-free (by
  \Cref{def:confusion-freeness}). Now, take the assembly robot $\rolename{A}$ and a subscription $\subscription$
  such that $\Gn{Warehouse} \conc \Gn{Factory}$ is $\subscription$-well-formed
  (by \Cref{def:wf}): rule (1) of causal consistency
  (\Cref{def:causalconsistency}) requires
  $\lgtn{car} \in \subscription(\rolename{A})$, implying
  $\rolename{A} \in \newroles[\lgtn{partReq}][{\Gn{Warehouse} \conc \Gn{Factory}}]$;
  therefore, from the \emph{branching} rule of determinacy
  (\Cref{def:determinacy}), we have
  $\lgtn{partReq}, \lgtn{closingTime} \in \subscription(\rolename{A})$. The
  \emph{joining} and \emph{looping} rules of \Cref{def:determinacy} hold
  vacuously, as there are no joining event types or loops in
  $\Gn{Warehouse} \conc \Gn{Factory}$.%
\end{example}

\subsection{On Computing Subscriptions}
\label{sec:swarm-protocols-composition-computing-sub}
\label{sec:constructwfsub}

\noindent%
This section introduces \Cref{alg:comp-swarm-subscriptions} for computing subscriptions that guarantee
well-formedness (by \Cref{def:wf}) for a composition $\G_1 \conc \ldots \conc \G_n$ of given swarm protocols $\G_1 \ldots \G_n$. %
In general, for any confusion-free swarm protocol $\G$ there is at least one
subscription $\subscription$ such that $\G$ is $\subscription$-well-formed:
it is the \quo{total} subscription where each role in $\G$ subscribes to all event types in $\G$.
However, if such a \quo{total} subscription is used to project machines from
$\G$ (using \Cref{def:swarm-proto-projection} later on), then the resulting swarm would consist of complex machines that await
the emission of all events emitted by all other machines, leading to inefficient executions.
It is hence worthwhile to compute smaller subscriptions that preserve
well-formedness.
Importantly, for a composition $\G_1 \conc \ldots \conc \G_n$ of $n$ \emph{composable} protocols (by \Cref{def:composable} below),
our \Cref{alg:comp-swarm-subscriptions} operates on the individual input protocols without expanding their composition (\Cref{def:swarm-proto-composition}),
because the size of the expanded protocol may grow exponentially with $n$ (as shown later in \Cref{sec:implementation:experiments}).

\begin{definition}[Composable swarm protocols] 
  \label{def:composable}
  A protocol is \emph{sequential} if it
  does not contain concurrent events. A set of protocols is \emph{composable} if they
  are pairwise interfacing and each protocol is sequential and confusion-free.
\end{definition}

\Cref{alg:comp-swarm-subscriptions} takes a set $\{\G_i\}_{i\in I}$ of
composable protocols together with a set $\{\subscription_i\}_{i\in I}$ of
subscriptions, and generates a subscription $\subscription$ such that
$\G = \G_1 \conc \ldots \conc \G_n$ is $\subscription$-well-formed and
$\subscription_i \subseteq \subscription$ holds for all $i \in I$.
\begin{algorithm}[t]
  \begin{tabular}{@{}l l@{}}
    \textbf{Input:}
    &
      A set $\{\G_1 \mydots \G_n\}$  of composable protocols and a set of subscriptions $\{\subscription_1\mydots \subscription_n\}$
    \\
    \textbf{Output:}
    &
      A subscription $\subscription \supseteq \subscription_1\mydots \subscription_n$ such that $ \G_1 \conc \ldots \conc \G_n$ is $\subscription$-well-formed.
    \\
  \end{tabular}

   	{\small
	\begin{algorithmic}
		\STATE Let $\oldroles$ denote the roles that subscribe to some event type of $\G$ in $\subscription$.
		\STATE $\subscription  \define \bigcup_{i\leq n} \subscription_i $%
		\STATE $\ifr \define \{ \role \mid \exists i,j\leq n : i\neq j\land \role \in \G_i \land \role \in \G_j\}$ %
		\STATE $\concSet \define \{ \{ \evt_i,\evt_j \}\mid i, j\leq n,\, \cmdindx{i}\in \G_i,\, \cmdindx{j}\in \G_j,\, \role_i \notin \G_j,\, \role_{j} \notin \G_i \}$

		\REPEAT
		\FORALL{ $i\leq n$, $\G'_i,\, \G''_i \subterm \G_i$, and $\role,\, \role',\, \evt,\, {\evt'} \in \G_i$}
		\STATE \textbf{Causal Consistency:} If $\G'_i \xrightarrow{\cmd}$, then add $\evt $ to $\subscription(\role)$; If $\G'_i \xrightarrow{\evt'} \xrightarrow{\cmd}$, then add $\evt' $ to $ \subscription (\role)$.
 		\STATE \textbf{Branching:} If $\evt $ is branching with $\evt'$ at $\G'_i$ and $\role \in \oldroles[\G'_i]$, then add $\evt, \evt' $ to $ \subscription(\role)$.
 		\vspace{-1.2em}
 		\STATE \textbf{Joining:} For all $j\leq n$, $\G'_j \subterm \G_j$, and $\evt'' \in \G_j$ such that $\G'_i \xrightarrow{\evt'} \xrightarrow{\evt}$, $\G'_j \xrightarrow{\evt''} \xrightarrow{\evt}$, $\{\evt',\evt''\}\in \concSet$, $\{\evt',\evt\}\not\in \concSet$, $\{\evt'',\evt\}\not\in \concSet$, and $\role\in \oldroles[\G'_i]$: Add $\evt',\, \evt'',\, \evt $ to $ \subscription (\role)$.
 		\STATE \textbf{Interfacing:} If $\G'_i \xrightarrow{\cmd} \G''_i$ with $\role \in \ifr$ and $\role' \in \oldroles[\G''_i]$, then add $\evt $ to $ \subscription (\role')$. %
		\ENDFOR
		\UNTIL $\subscription$ no longer changes
		\FORALL{ $i\leq n$, $\G'_i \subterm \G_i$, such that $\G'_i \xrightarrow{\evt \logcat \lgt} \G'_i$ for some $\evt$ and $\lgt$, and $\role \in \oldroles[\G'_i]$}
		\STATE %
		\textbf{Looping:} If $\G'_i$ has no looping event type in $\subscription$, then add $\evt $ to $ \subscription (\role)$.
		\ENDFOR
	\end{algorithmic}}
\caption{Computing a subscription for a composition of protocols.}\label{alg:comp-swarm-subscriptions}
\end{algorithm}

The algorithm initialises three sets:
\begin{itemize}
\item $\subscription$, initially approximating the result as the union of the 
  given $\subscription_i$. Each $\subscription_i$ can be understood as the
  specification of the event types that each role in $\G_i$ is interested in.
\item $\ifr$ containing all interfacing roles.
\item $\concSet$ containing all potentially concurrent event types, i.e.,
  the pairs of event types belonging to different components that are emited
    by non-interfacing roles.%
\end{itemize}
The algorithm iteratively enlarges $\subscription$ by applying monotone
operations corresponding to the definitions of causal-consistency
(\Cref{def:causalconsistency}) and \emph{branching} and \emph{joining} of
determinacy (\Cref{def:determinacy}), together with the new operator
\emph{interfacing}, which is crucial for ensuring the correct synchronisation
of the protocols (as will be discussed below).
The iteration terminates when a fixed point is reached.
The final phase ensures that every role that is involved in a term that
contains a loop $\lgt$ is subscribed to (at least) one looping event of
$\lgt$.%

\longer{%
  \begin{example}[Key Insight]\label{ex:separating_interface}
    The composition of $\G_1$ and $\G_2$ below leads to a simple sequence with
    one interfacing event type.

    $$\G_1 = \cmdindx{1} \cdot \mycmd{IR}{i} \cdot \zero + \mycmd{R'_1}{t'_1} \cdot \zero%
    \qquad \G_2 = \mycmd{IR}{i} \cdot \cmdindx{2} \cdot \zero $$%
    $$\G_1 \conc \G_2 = \cmdindx{1} \cdot  \mycmd{IR}{i} \cdot \cmdindx{2}\cdot \zero + \mycmd{R'_1}{t'_1} \cdot \zero  $$

    \noindent%
    Here, $\evt_1$ and $\evt_2$ are separated in $\G_1\conc \G_2$ by the
    interfacing event type $\lgtn{i}$. Because of this, $\evt_1$ and $\evt_2$
    are not concurrent. Assume $\evt_2 \in \subscription (\role)$ for some
    role $\role$ and thus $\role \in \newroles[t_1][\G_1 \conc \G_2]$.
    According to determinacy (\Cref{def:determinacy}), it must hold
    $\evt_1 \in \subscription (\role)$. An application of the \emph{Interface}
    rule ensures $\lgtn{i}\in \subscription(\role)$ since
    $\role \in \oldroles[\G_2][\subscription] $. This means
    $\role \in \oldroles[\G_1][\subscription] $ and thus an application of the
    \emph{Branching} rule gives us $\evt_1\in \subscription(\role)$. The
    algorithm ensures determinacy.
  \end{example}
}

\subparagraph*{On the need of the \emph{interfacing} step in \Cref{alg:comp-swarm-subscriptions}.}
One of the requirements of well-formedness (\Cref{def:wf}) is \emph{determinacy} (\Cref{def:determinacy}).
Checking the determinacy of a composed swarm protocol for a subscription requires the set $\newroles[\evt][\G']$ %
for subterms $\G' = \G_1' \conc \ldots \conc \G_n'$ -- and computing that set would require the
explicit computation of the expanded protocol $\G'$ (using \Cref{def:swarm-proto-composition});
however, the size  of the expanded $\G'$ can be exponential in $n$.
To avoid this exponential blow-up, \Cref{alg:comp-swarm-subscriptions} uses the set $\oldroles[\G'_i]$ (where $\G'_i$ is a subterm of
one of the input protocols $\G_i$)
and its \emph{interfacing} step to over-approximate the set
$\newroles$ without explicitly expanding $\G = \G_1 \conc \ldots \conc \G_n$.
To see why this approach yields the desired over-approximation, recall that $\role \in \newroles[\evt][\G']$ holds if there exists a path $\G' \xrightarrow{\evt_0} \xrightarrow{\lgt_1} \xrightarrow{\evt_1} \cdots \xrightarrow{\lgt_n} \xrightarrow{\evt_n}$
such that $\evt_0 = \evt$, $\evt_n \in \subscription(\role)$, and every pair $\evt_i, \evt_{i+1}$ is not concurrent for $0 \leq i < n$.
Since $\evt_i$ and $\evt_{i+1}$ are not concurrent, they must either belong to the same concurrency-free input swarm protocol or they are separated by interfacing event types that enforce their ordering.
Consequently, there exists a path $\G' \xrightarrow{\evt'_0} \xrightarrow{\lgt'_1} \xrightarrow{\evt'_1} \cdots \xrightarrow{\lgt'_{n'}} \xrightarrow{\evt'_{n'}}$ satisfying the conditions above, where \(\evt'_1, \dots, \evt'_{n'-1}\) are interfacing event types and  \(\evt'_i\), \(\evt'_{i+1}\) occur within the same input protocol for \(0 \le i < n'\).
The \emph{interfacing} step propagates subscriptions along this path, that is, if $\evt'_{i+1} \in \subscription(\role)$, then also $\evt'_i \in \subscription(\role)$.
Starting from $\evt'_{n'} \in \subscription(\role)$, repeated application of the \emph{interfacing} step backwards along the path ensures that $\evt'_1 \in \subscription(\role)$.
Since $\evt'_1$ is in the same input protocol $\G'_i$ as $\evt$, we obtain $\role \in \oldroles[\G'_i]$.
\emph{(This argument is formally stated in \Cref{lem:subscribers-correctness} in \Cref{sec:app:proof_subscription_WF})}.

\begin{example}[Application of \Cref{alg:comp-swarm-subscriptions} to
  $\Gn{Warehouse}$ and $\Gn{Factory}$]\label{ex:applying-sub-gen-alg}
  We illustrate the application of \Cref{alg:comp-swarm-subscriptions} to the
  protocols $\Gn{Warehouse}$ and $\Gn{Factory}$ (depicted in
  \Cref{fig:Warehouse,fig:Factory}, respectively) and two empty subscriptions
  $\subscription_1 = \subscription_2 = \emptyset$.
  It is straightforward to check that these two protocols are composable, since
  they are sequential, confusion-free and they interface on $\arole[][T]$.
  Then, the algorithm initialises
  $\subscription = \subscription_1 \cup \subscription_2 = \emptyset$,
  $\ifr = \{\arole[][T]\}$, and
    \(%
    \concSet = \{ \{\lgtn{closingTime}, \lgtn{car}\}, \{\lgtn{pos},
    \lgtn{car}\}\}
    \).
    Note that $\concSet$ over-approximates the actual pairs of concurrent
    event types. Indeed, only $\lgtn{closingTime}$ and $\lgtn{car}$ are
    concurrent in the composition $\Gn{Warehouse} \conc \Gn{Factory}$ (see
    \Cref{fig:Composition}). Now the algorithm starts the iteration with the
    step \emph{causal consistency}.
    Starting with role $\arole[][T]$ and protocol $\Gn{Factory}$, %
    we first add the event types $\lgtn{partReq}$ and $\lgtn{partOK}$
    to $\subscription(\arole[][T])$ because $\Gn{Factory}$ has transitions
    where $\arole[][T]$ emits such event types.
    In $\Gn{Warehouse}$, the same role $\arole[][T]$ emits
    $\lgtn{partOK}$ after $\arole[][FL]$ emits $\lgtn{pos}$; consequently
    $\lgtn{pos}$ is added to $\subscription(\arole[][T])$. Subscriptions for
    the remaining roles, shown below, are obtained analogously.

\ifbool{ecoop}{\smallskip}{\medskip}%
  \centerline{\(%
    \subscription \;=\; \left\{
      \begin{array}{ll}
        \rolename{T} \mapsto \{ \lgtn{partReq}, \lgtn{partOK}, \lgtn{pos}\}, &\rolename{FL} \mapsto \{\lgtn{partReq}, \lgtn{pos}\},\\
         \rolename{D} \mapsto \{\lgtn{partOK}, \lgtn{closingTime}\}, & \rolename{A} \mapsto \{\lgtn{partOK}, \lgtn{car}\}
      \end{array}
    \right\}
    \)}%
  \ifbool{ecoop}{\smallskip}{\medskip}%

    \noindent%
    Then, the \emph{branching} step is executed. Note that $\lgtn{partReq}$ and
    $\lgtn{closingTime}$ are branching in $\Gn{Warehouse}$. Therefore, ($i$)
    $\lgtn{partReq}$ is added to the subscriptions of the roles already
    subscribed to event types of \changeOursNoMargin{$\Gn{Warehouse}$}, i.e., to $\rolename{D}$ and
    $\rolename{A}$; ($ii$) analogously, $\lgtn{closingTime}$ is added to
    $\rolename{T}$, $\rolename{FL}$, and $\rolename{A}$.
	 Step \emph{joining} is not actually applied in this example because its conditions are not satisfied.

Next, the \emph{interfacing} step adds $\lgtn{partOK}$ to all roles that have subscribed to some event that follows $\lgtn{partOK}$ in some of the protocols; hence,   $\lgtn{partOK}$ is added to $\rolename{FL}$ because it already subscribed to $\lgtn{partReq}$.

The fixed point of the iteration is:

\smallskip%
\centerline{\(%
  \subscription = \left\{
    \begin{array}{@{}l@{}}
      \rolename{T}, \ \ \rolename{FL}\ \mapsto \{ \lgtn{partReq}, \lgtn{partOK}, \lgtn{closingTime}, \lgtn{pos}\}
      \\
      \rolename{D} \mapsto \{\lgtn{partReq}, \lgtn{partOK}, \lgtn{closingTime}\},
      \quad
      \rolename{A} \mapsto \{\lgtn{partReq}, \lgtn{partOK}, \lgtn{closingTime}, \lgtn{car}\}\\
    \end{array}
  \right\}
  \)}%
  \smallskip%

  \noindent%
  In $\Gn{Warehouse}$, the event type $\lgtn{partReq}$ is looping in $\subscription$. However, it was already added to the relevant roles due to the branching rule.
  Thus, the \emph{looping} phase is not applied.%
\end{example}

\Cref{lem:subscription_composition_WF} below states that
\Cref{alg:comp-swarm-subscriptions} computes a subscription $\subscription$
ensuring well-formedness of the composed protocol. \short{\emph{(Proof in
    \Cref{sec:app:proof_subscription_WF}.)}}

\begin{theorem}[\ecoop{Correctness of \Cref{alg:comp-swarm-subscriptions}}]
    \label{lem:subscription_composition_WF}%
    If $\subscription$ is the subscription computed from the composable protocols
    $\G_1 \mydots \G_n$ and subscriptions $\subscription_1\mydots \subscription_n$
    using \Cref{alg:comp-swarm-subscriptions},

	 then $\G_1 \conc \ldots \conc \G_n$ is
	 $\subscription$-well-formed.
\end{theorem}
\longer{\section{Proof of \Cref{lem:subscription_composition_WF}} \label{sec:app:proof_subscription_WF}
\Cref{alg:comp-swarm-subscriptions} makes use of a key insight: the function $subscribers$ overapproximates $roles$.
We now outline the reasoning behind this insight and formalise it below in \Cref{lem:subscribers-correctness}. %
Suppose $\role \in \newroles$. By \Cref{def:determinacy} there is a path
$\G \xrightarrow{\evt_0}  \xrightarrow{\lgtn{l_1} \cdot \evt_1} \ldots  \xrightarrow{\lgtn{l_n} \cdot \evt_n}$
such that $\evt=\evt_0$, $\evt_n \in \subscription(\role)$, and $\evt_i , \evt_{i+1}$ are not concurrent for all $i < n$.
Whenever two consecutive (and thus non-concurrent) event types $\evt_i$ and $\evt_{i+1}$ along the path above belong to different input swarm protocols, the corresponding path segment $\lgt_{i+1}$ must include a sequence of interfacing event types $\evt'_1\mydots \evt'_m$ that enforce their order:
i.e., $\lgt_{i+1}=\lgt'_1 \cdot \evt'_1 \cdot \ldots \lgt'_{n'}\cdot  \evt'_{n'} \cdot  \lgt'_{n'+1}$.
We take the sequence $\evt_1 \ldots \evt_{n-1}$ and, for each such case, replace $\evt_i \cdot \evt_{i+1}$ with the sequence $\evt_i\cdot \evt'_1 \ldots \evt'_{n'}$. This yields a new sequence $\evt''_1 \ldots \evt''_{n''}$ where adjacent events from different protocols are guaranteed to be interfacing.
If we now discard any non-interfacing event type from $\evt''_1\ldots \evt''_{n''}$, we obtain a sequence $\evt'''_1 \ldots \evt'''_{n'''}$ consisting solely of interfacing event types.
Since $\evt_n \in \subscription(\role)$, by iteratively applying the \emph{Interfacing} rule we get $\evt'''_{n'''} \in \subscription(\role)$, $\evt'''_{n'''-1} \in \subscription(\role)$, $\ldots$, $\evt'''_{1} \in  \subscription(\role)$.
Since the first event $\evt$ and $\evt'''_1$ belong to the same protocol, say $\G_i$, we have $\role \in \oldroles[\G'_i]$ for some $\G'_i \subterm \G_i$ with $\G'_i \xrightarrow{\evt}$.
Therefore we conclude that whenever $\role \in \newroles$, then $\role \in \oldroles[\G'_i]$ (for some $\G'_i \subterm \G_i$ with $\G'_i \xrightarrow{\evt}$).

We express this as a technical lemma:
\begin{lemma}\label{lem:subscribers-correctness}
We apply \Cref{alg:comp-swarm-subscriptions}. Assume $\role \in \newroles[\evt][\G']$ for some $\G' \subterm \G=\G_1 \conc \ldots \conc \G_n$ and $\G'=\G'_1 \conc \ldots \conc \G'_n$ where $\G'_i \subterm \G_i$ for $i\leq n$. 
 It holds $\role \in \oldroles[\evt][\G'_i]$ for any $i\leq n$ with $\G'_i\xrightarrow{\evt}$.
\end{lemma}
\begin{proof}[Proof of \Cref{lem:subscribers-correctness}]
Since $\role \in \newroles[\evt][\G']$, 
there is a path $ \G' \xrightarrow{\evt_0} \xrightarrow{\lgtn{l_1} } \xrightarrow{\evt_1} \ldots  \xrightarrow{\lgtn{l_n} } \xrightarrow{\evt_n}$ in the composition $\G$ such that $\evt=\evt_0$, $\evt_n \in \subscription(\role)$, and $\evt_i, \evt_{i+1}$ are not concurrent for $i<n$.

We now show that we can assume WLOG that $\evt_{1}\ldots \evt_{n-1}$ are interfacing event types such that $\evt_{i},\, \evt_{i+1}$ occur in the same swarm protocol for  $i< n$.
Assume, towards contradiction, that this is not the case. This means that there are  $\evt_i,\evt_{i+1}$ that do not occur in the same swarm protocols and they are not both interfacing. Further, there are no event types in $\lgt_{i+1}$ that can be added to the sequence $\evt_{1}\ldots \evt_{n-1}$. After-all, if there were some we could add, then we could keep adding those event types to the sequence until all pairs from different protocols are interfacing.
Since there are no elements in $\lgt_{i+1}$ that can be added, this means there is no chain of pairwise non-concurrent event types in $\lgt_{i+1}$ that starts with $\evt_i$ and ends in $\evt_{i+1}$. 
Thus, we can rearrange the path (by moving event types from $\lgt_{i+1}$ directly before $\evt_i$ or directly after $\evt_{i+1}$) such that $\evt_i$ and $\evt_{i+1}$ are next to each other. 
According to the definition of swarm protocol composition, these two non-interfacing event types from different protocols are concurrent. This is a contradiction. 

This means whenever  $\evt_i,\evt_{i+1}$ are from different protocols, then they are interfacing. If $\evt_i$ is non-interfacing, this means $\evt_{i-1}$ and $\evt_{i+1}$ occur in the same swarm protocol as $\evt_i$ and thus they are non-concurrent with each other. This means we can safely remove $\evt_i$ from the sequence $\evt_{1}\ldots \evt_{n-1}$. We can safely remove all non-interfacing event types and then our WLOG assumption holds.

Since $\evt_{n-1},\, \evt_{n}$ occur in the same swarm protocol, say $\G_j$, and $\G$ has a path $\xrightarrow{\evt_{n-1}}\xrightarrow{\lgt_{n-1}}\xrightarrow{\evt_n}$ that means there are $\G',\G''\subterm \G_j$ such that $\G' \xrightarrow{\evt_{n-1}} \G''$, $\evt_n\in \G''$, and $\evt_{n}\in \subscription(\role)$.
It follows $\role \in \oldroles[\G''][\subscription]$. 
Recall, that $\evt_1 \ldots \evt_{n-1}$ are interfacing.
The \emph{Interfacing} rule ensures $\evt_{n-1}\in \subscription(\role)$.
We iterate this application of the \emph{Interfacing} rule to get $\evt_1 \in \subscription (\role)$.

Assume $\evt$ is not interfacing. This means it occurs only in one swarm protocol $\G_i$.
Then $\evt_1 \in \subscription (\role)$ together with $\evt_1$ occurs in $\G'_i$, ensures that $\role\in \oldroles[\G'_i]$.

If $\evt$ is interfacing, then $\evt$ and $\evt_1$ occur together in some swarm protocol. In this case, the \emph{Interfacing} rule ensures $\evt \in \subscription(\role)$ and thus $\role \in \oldroles[\G'_i]$ for any $i\leq n$ with $\G'_i\xrightarrow{\evt}$.
\end{proof}

Recall that we call those looping event type that satisfies the looping condition in Determinacy and are thus chosen as updating event types \emph{looping in $\subscription$}.

\begin{proof}[Proof of \Cref{lem:subscription_composition_WF}]
To prove that $\G=\G_1 \conc \ldots \conc \G_n$ is $\subscription$-well-formed, we verify the following conditions: (1) confusion-freeness, (2) causal consistency, and (3) determinacy.
It is easy to see that  confusion-freeness (\Cref{def:confusion-freeness}) holds: Property (1) holds since $\G_1 \conc \ldots \conc \G_n$ are confusion-free and interfacing. Properties (2) and (3) hold since the swarm protocol composition preserves these properties.
Causal consistency is enforced by the Causal Consistency rule.
We now verify that determinacy (\Cref{def:determinacy}) holds by examining each condition that determinacy requires.
\subparagraph{Branching.}
Assume $\evt$ is branching with $\evt'$ at $\G'\subterm \G$.
It holds $\G'=\G'_1 \conc \ldots \conc \G'_n$ where $\G'_i\subterm \G_i$ for $i\leq n$ and there is an $i\leq n$ where $\evt$ branches with $\evt'$ at $\G'_i$. Thus they are not in $\concSet$.
Since $\evt$ and $\evt'$ are not concurrent, but both are outgoing transitions of $\G'$, 
\Cref{def:swarm-proto-composition} (Swarm Protocol Composition) ensures there is some $\G_i$ with $i\leq n$ that they both occur in. 
For any  $\role \in \newroles[\evt][\G']$, \Cref{lem:subscribers-correctness} ensures $\role\in \oldroles[\G'_i]$ and thus the \emph{Branching} rule sets $\evt,\evt' \in \subscription(\role)$.
\subparagraph{Joining.}
Assume towards contradiction that the Joining condition is not satisfied.
Then there are  $\G_a, \G_b, \G_c,\G_d \subterm \G$ with  $\G_a \xrightarrow{\evt_a} \G_b \xrightarrow{\evt}\G_c$, $\G_d \xrightarrow{\evt_b} \G_b$,  $\evt_a$ and $\evt_b$ are concurrent but not concurrent with $\evt$ and $\role \in \newroles[\evt][\G_b]$ and  $\{\evt_a, \evt_b, \evt \} \not \subseteq \subscription(\role)$. 
It holds that since $\evt_a$ and $\evt$ are not concurrent, then it is not the case that: $\G_a \xrightarrow{\evt} \xrightarrow{\evt_a}\G_c$.
According to the composition of swarm protocols, this means that there is a $\G_i$ such that $\evt_a, \evt \in \G_i$ and thus $\{\evt_a,\evt \}\not\in \concSet$. 
Further, there is a $\G'_i \subterm \G_i$ such that $ \G'_i\xrightarrow{\evt_a} \xrightarrow{\evt}$.
Analogue, we argue there is a $\G'_j \subterm \G_j$ such that $\G'_j \xrightarrow{\evt_b}  \xrightarrow{\evt}$ and thus $\{\evt_b,\evt \}\not\in \concSet$. 
Since $\concSet$ is an overapproximation of concurrent event types, it holds  $\{\evt_a,\evt_b \}\in \concSet$.
This means there is a role $\role \in \newroles[\evt][\G_b]$ but there is no corresponding $\G'_j \subterm \G_j$ with $j\leq n$, $\G'_j \xrightarrow{\evt}$, and $\role \in \oldroles[\G'_j]$. This is a contradiction to \Cref{lem:subscribers-correctness}.
\subparagraph{Looping.}
Assume towards contradiction that the Looping condition is not satisfied.
Then there is a  $\G' \subterm \G$ with  $\G' \xrightarrow{\lgt} \G'$ wit $\lgt$ being non-empty and $\G'$ has no looping update.
Let $\G'=(\G'_1 \conc \ldots \conc \G'_n)$.
According to the composition of swarm protocols, this means that there is an $i\leq n$ such that $\G'_i \xrightarrow{\lgt'} \G'_i$ and $\epsilon \subset \lgt' \subseteq \lgt$.
If $\lgt$ contains an interface event type, then the interfacing rule ensures that this is a looping update.
If $\lgt$ contains no interface event type, then neither does $\lgt'$.
If $\G'_i$ contains branching event types, then at least one of those event types remains in $\lgt'$, even if it may be no longer branching. 
We argued in paragraph \emph{Branching} that the remaining one is in the necessary subscriptions to be a looping update. This is a contradiction.
If $\G'_i$ contains no branching event types, then Step 3 picks an event type $\evt\in \lgt'$ and ensures $\evt\in \subscription(\role)$ for all $\role \in \oldroles$.
Since the looping condition is not satisfied by $\evt$, that means $\role \in \newroles[\evt][\G']$ but not $\role \in \oldroles[\G'_i]$.
This is a contradiction, the argument is analogue to branching.
\end{proof}
}

\section{Realising Swarm Protocols as Swarms of Machines}
\label{sec:projecting-machines-comp-swarm-proto}

\noindent%
We now address the problem of correctly \emph{realising} a swarm protocol $\G$ (\Cref{def:realisation}), i.e.,
constructing a swarm of machines that interact as specified by $\G$. %
We begin by fixing the language for writing machines (\Cref{sec:formalisms:machines}).
Following the standard top-down approach for behavioral types, we rely on a \emph{projection} operation (\cref{sec:projection}) to derive correct machines for the different roles of a protocol. %
Then, we refine the swarm model introduced
in~\cite{DBLP:conf/ecoop/KuhnMT23} by adding a
\emph{branch-tracking} mechanism to the semantics of machines (\cref{sec:branch-tracking-machines}) in a swarm (\cref{sec:formalisms:swarms}).
This revised semantics is essential for (1) ensuring that the machines projected from a swarm protocol $\G$ are \emph{eventually faithful} to $\G$
(\Cref{sec:eventual-fidelity}), and (2)
enabling the composition of swarms (discussed in \cref{sec:swarm-composition}).

\subsection{Machines}
\label{sec:formalisms:machines}
A \emph{machine} models a swarm agent that processes a local log of
events to make decisions on the emission of new events.
We define a machine as a grammar term (\Cref{def:formalisms:machine}) with a corresponding finite-state automaton (\Cref{def:machine-state-transitions}).%

\begin{definition}[Machine]
  \label{def:formalisms:machine}
  A \emph{machine} is a \emph{regular term} from the
  coinductive grammar:

  \smallskip%
  \centerline{\(
		\afish  \stackrel{\text{co}}\bnfdef \kappa \cdot  \&_{i\in I} \inp[i] \afish[i] \qqand[where $\kappa$ is a finite set of event types dubbed the \emph{emitter set}]
  \)}%
  \smallskip%

    \noindent
		such that $\evt_i \neq \evt_j$  for all $i \neq j \in I$.
		We write \( \zero \) when \( \kappa = I = \emptyset \).
\end{definition}

\begin{definition}[Machine states and transitions]
\label{def:machine-state-transitions}
A machine $\afish$ from \Cref{def:formalisms:machine} %
can be described as the deterministic finite-state automaton where:
\begin{itemize}
	\item Each distinct subterm of \( \afish \) (finitely many since $\afish$ is regular) corresponds to a \textbf{state}.
\item The \textbf{initial state} is \( \afish \) itself.
  \item There is an \textbf{event acceptance transition}  $\afish[1] \xrightarrow{\inp} \afish[2]$
  whenever machine $\afish[1]$ has a branch $\inp$ leading to
  $\afish[2]$
  \item and an \textbf{event emission transition}
  $\afish' \xrightarrow{\outp[]} \afish'$ for any $\evt$ in the
  emitter set of $\afish'$.
\end{itemize}
\end{definition}

\begin{remark}
The grammar-based presentation of machines (\Cref{def:formalisms:machine}) and the automaton presentation (\Cref{def:machine-state-transitions}) are equivalent: each distinct subterm corresponds to a state, branches correspond to acceptance transitions, and emitter sets yield emission self-loops.%
\end{remark}

\begin{example}%
  By \Cref{def:machine-state-transitions},
  the automaton depicted in \cref{fig:Door} corresponds to the
  machine:\footnote{%
   For readability, in our examples we often use brackets and infix
   notation for the branching operator $\&$.%
  }

  \ifbool{ecoop}{\smallskip}{\medskip}\centerline{\(%
   \afish_\rolename{D} \;=\; \{\aeventtype[][closingTime]\} \cdot \Bigl( (\inp[][partReq] \inp[][partOK] \afish )\ \&\ (\inp[][closingTime] \zero ) \Bigr)
  \)}%
  \ifbool{ecoop}{\smallskip}{\medskip}%

  \noindent%
  Note that a closing-time event can be emitted only
  from $\afish_\rolename{D}$ (the initial state in the automaton of \cref{fig:Door})
  since $\aeventtype[][closingTime]$ belongs only to the emitter set of $\afish_\rolename{D}$.
  Also note that in \cref{fig:Door}, emitted events appear as \emph{self-loops}:
this is because %
(by \Cref{def:machine-state-transitions} above)
a machine remains in its current state while emitting an event;
a state transition may occur only when a machine accepts an event from its log.
\end{example}

\subsection{Projecting a Swarm Protocol onto its Roles}
\label{sec:projection}

By \Cref{def:swarm-proto-projection} below, the \emph{projection} of a swarm
protocol $\G$ onto a role $\role$ yields a machine that captures the behaviour
of $\role$ in $\G$. Specifically, the projected machine
retains the event types that are either emitted or subscribed to
by $\role$ -- and it disregards all other event types.

\begin{definition}[Projection]
	\label{def:swarm-proto-projection}
	We coinductively define the \emph{projection} of a $\subscription$-well-formed swarm protocol $\G$ onto a role $\role$  as:

	\ifbool{ecoop}{\smallskip}{\medskip}\centerline{\(
		\G \proterm \;\;\stackrel{\text{co}}{\define}\;\; \kappa \cdot  \& \Bigl\{  \evt ? (\G' \proterm) \;\;\Big|\;\;  \exists \lgtn{l}:
		\G \xrightarrow{\lgtn{l}} \xrightarrow{\evt} \G' \text{ and }
		\lgtn{l}\cap \subscription (\role)=\emptyset%
		\;\text{ and }\; \evt \in \subscription(\role)  \Bigr\}
		\)}%
	\ifbool{ecoop}{\smallskip}{\medskip}

	\noindent
\begin{tabular}{lrl}
	where& $\kappa = \bigl\{\evt \;\big|\; \exists \lgtn{l}:
	\G \xrightarrow{\lgtn{l}} \xrightarrow{\mycmd{R}{t}}
	\text{ and } \lgtn{l}\cap \subscription (\role)=\emptyset
	\bigr\}$.\\
\end{tabular}
\end{definition}

Observe that, by \Cref{def:swarm-proto-projection}, the initial state of the projected machine $\G \proterm$ has:
\begin{itemize}
\item A branch $\inp$ for every event type $\aeventtype$ that
  $\arole[]$ subscribes to (by $\subscription$) and can be reached from $\G$ after
  a (possibly empty) sequence of events of type $\lgt$ \emph{not} subscribed to by $\arole[]$.
\item An emitter set $\kappa$ consisting of all the event types that
  $\arole[]$ can emit in $\G$ after a (possibly empty)
  sequence of events of type $\lgt$ \emph{not} subscribed to by $\role$.%
\end{itemize}
\noindent%
Hence, the projection $\G \proterm$ may start by either
emitting an event of a type in $\kappa$ or proceed to
$\G' \proterm$ upon accepting an event of type $\evt$.
Note that an
event type $\evt$ that is \emph{not} subscribed to by $\role$ does \emph{not}
appear in $\G \proterm$, i.e., the projected machine ignores all events of such a type $\evt$.

\begin{example}[Projection]
\label{ex:projection}
Consider the swarm protocol $\G = \Gn{Warehouse} \conc \Gn{Factory}$ shown in
\Cref{fig:Composition} and the following subscription (where $E$ stands for
the set of all event types in $\G$):

\smallskip%
\(
  \subscription \;=\; \left\{
      \rolename{T} \mapsto E \setminus \{ \lgtn{car}\}, \;\;
      \rolename{D} \mapsto E \setminus \{ \lgtn{pos}, \lgtn{car}\}, \;\;
      \rolename{FL} \mapsto E \setminus \{ \lgtn{partOK}, \lgtn{car}\}, \;\;
      \rolename{A} \mapsto E \setminus \{ \lgtn{pos} \}
    \right\}
\)%
\smallskip%

\noindent
The projections of $\G$ are shown in \Cref{fig:ProjWHF}.
\end{example}

\begin{figure}[tb]
\ifbool{ecoop}{  \vspace{-5mm}%
  \begin{minipage}{0.55\linewidth}}{\centering}
	   \tabcolsep=0pt
  \begin{tabular}{r@{\quad$=$\quad}l}
    $\G \projshort{\subscription}{\rolename{T}}$ &
    \raisebox{-5mm}{%
    \begin{tikzpicture}[gt, scale = \ifbool{ecoop}{.6}{1}, node distance = 4mm and 15mm, font=\footnotesize]
        \node (0) [plb, minimum size = 4mm]{};
        \node (1) [below = 0.3cm of 0, minimum size = 4mm] {};
        \node (2) [plb, right = of 0, minimum size = 4mm] {};
        \node (4) [plb, right = of 2, minimum size = 4mm] {};
        \node (3) [plb, left = of 0, minimum size = 4mm] {};
        \node (5) [plb, right = of 4, minimum size = 4mm] {};
        \node (6) [plb, right = of 5, minimum size = 4mm] {};

        \draw[->] (1) to node[sloped, anchor=center, above=1mm]{}  (0);
        \path (0) edge[loop above] node{$\lgtn{partReq}!$}  (0);
        \path (0) edge[above] node{$ \lgtn{closingTime}?$}  (3);
        \path (0) edge[sloped, above] node{$ \lgtn{partReq}?$}  (2);
        \path (2) edge[sloped, above] node{$ \lgtn{pos}?$}  (4);
        \path (4) edge[loop above] node{$\lgtn{partOK}!$}  (4);
        \path (4) edge[sloped, above] node{$  \lgtn{partOK}?$}  (5);
        \path (5) edge[sloped, above] node{$\lgtn{closingTime}?$} (6);
    \end{tikzpicture}
    }%
    \\[1em]
    $\G \projshort{\subscription}{\rolename{D}}$ &
    \raisebox{-5mm}{%
    \begin{tikzpicture}[gt, scale = \ifbool{ecoop}{.6}{1}, node distance = 4mm and 15mm, font=\footnotesize]
        \node (0) [plb, inner sep=0pt, minimum size = 4mm]{};
        \node (1) [below =0.3cm of 0,inner sep=0pt, minimum size = 4mm] {};
        \node (2) [plb, right = of 0,inner sep=0pt, minimum size = 4mm] {};
        \node (3) [plb, left = of 0,inner sep=0pt, minimum size = 4mm] {};
        \node (4) [plb, right = of 2,inner sep=0pt, minimum size = 4mm] {};
        \node (5) [plb, right = of 4,inner sep=0pt, minimum size = 4mm] {};

        \draw[->] (1) to node[sloped, anchor=center, above=1mm]{}  (0);
        \path (0) edge[loop above] node{$\lgtn{closingTime}!$}  (0);
        \path (0) edge[above] node{$ \lgtn{closingTime}?$}  (3);
        \path (0) edge[sloped, above] node{$ \lgtn{partReq}?$}  (2);
        \path (2) edge[sloped, above] node{$  \lgtn{partOK}?$}  (4);
        \path (4) edge[loop above] node{$\lgtn{closingTime}!$} ();
        \path (4) edge[sloped, above] node{$\lgtn{closingTime}?$} (5);
    \end{tikzpicture}
    }%
    \\[1em]
    $\G \projshort{\subscription}{\rolename{FL}}$ &
    \raisebox{-5mm}{%
    \begin{tikzpicture}[gt, scale = \ifbool{ecoop}{.6}{1}, node distance = 4mm and 15mm, font=\footnotesize]
        \node (0) [plb, minimum size = 4mm]{};
        \node (1) [below =0.3cm of 0, minimum size = 4mm] {};
        \node (2) [plb, right = of 0, minimum size = 4mm] {};
        \node (3) [plb, left = of 0, minimum size = 4mm] {};
        \node (4) [plb, right = of 2, minimum size = 4mm] {};
        \node (5) [plb, right = of 4, minimum size = 4mm] {};

        \draw[->] (1) to node[sloped, anchor=center, above=1mm]{}  (0);
        \path (0) edge[above] node{$ \lgtn{closingTime}?$}  (3);
        \path (0) edge[sloped, above] node{$ \lgtn{partReq}?$}  (2);
        \path (2) edge[loop below] node{$\lgtn{pos}!$} ();
        \path (2) edge[sloped, above] node{$  \lgtn{pos}?$}  (4);
        \path (4) edge[sloped, above] node{$\lgtn{closingTime}?$} (5);
    \end{tikzpicture}
    }%
    \\[1em]
    $\G \projshort{\subscription}{\rolename{A}}$ &
    \raisebox{-5mm}{%
    \begin{tikzpicture}[gt, scale = \ifbool{ecoop}{.6}{1}, node distance = 4mm and 15mm, font=\footnotesize]
        \node (0) [plb, minimum size = 4mm]{};
        \node (1) [below =0.3cm of 0, minimum size = 4mm] {};
        \node (2) [plb, right = of 0, minimum size = 4mm] {};
        \node (3)  [plb, left = of 0, minimum size = 4mm] {};
        \node (4) [plb, right = of 2, minimum size = 4mm] {};
        \node (5) [plb, above right = of 4, minimum size = 4mm] {};
        \node (6) [plb, below right = of 4, minimum size = 4mm] {};
        \node (7) [plb, above right = of 6, minimum size = 4mm] {};

        \draw[->] (1) to node[sloped, anchor=center, above=1mm]{}  (0);
        \path (0) edge[above] node{$ \lgtn{closingTime}?$}  (3);
        \path (0) edge[sloped, above] node{$ \lgtn{partReq}?$}  (2);
        \path (2) edge[sloped, above] node{$  \lgtn{partOK}?$}  (4);
        \path (4) edge[loop above] node{$\lgtn{car}!$} ();
        \path (4) edge[sloped, above] node{$\lgtn{closingTime}?$} (5);
        \path (5) edge[loop below] node{$\lgtn{car}!$} ();
        \path (5) edge[sloped, above] node{$\lgtn{car}?$} (7);
        \path (4) edge[sloped, below] node{$\lgtn{car}?$} (6);
        \path (6) edge[sloped, below] node{$\lgtn{closingTime}?$} (7);
    \end{tikzpicture}
    }%
  \end{tabular}
  \vspace{-3mm}
  \caption{Projection of $\G = \Gn{Warehouse} \conc \Gn{Factory}$ (from \Cref{fig:Composition}) on roles $\rolename{T}, \rolename{D}, \rolename{FL}$, and $\rolename{A}$, using the subscription $\subscription$ defined in \Cref{ex:projection}.}
  \label{fig:ProjWHF}
 \ifbool{ecoop}{ \end{minipage}
 \hfill%
 \begin{minipage}{0.4\linewidth}}{\end{figure}\begin{figure}\centering}
  \begin{tabular}{r@{\quad$=$\quad}l}
    $\G \projshort{\subscription}{\rolename{T}}$ &
    \raisebox{-5mm}{%
    \begin{tikzpicture}[gt, scale = \ifbool{ecoop}{.6}{1}, node distance = 4mm and 15mm, font=\footnotesize]
        \node (0) [plb, minimum size = 4mm]{};
        \node (1) [below =0.3cm of 0, minimum size = 4mm] {};
        \node (2) [plb, above right = of 0, minimum size = 4mm] {};
        \node (4) [plb, below right = of 2, minimum size = 4mm] {};
        \node (3)  [plb, left = of 0, minimum size = 4mm] {};

        \draw[->] (1) to node[sloped, anchor=center, above=1mm]{}  (0);
        \path (0) edge[loop above] node{$\lgtn{partReq}!$}  (0);
        \path (0) edge[above] node{$ \lgtn{closingTime}?$}  (3);
        \path (0) edge[sloped, above] node{$ \lgtn{partReq}?$}  (2);
        \path (2) edge[sloped, above] node{$ \lgtn{pos}?$}  (4);
        \path (4) edge[loop above] node{$\lgtn{partOK}!$}  (4);
        \path (4) edge[sloped, above] node{$  \lgtn{partOK}?$}  (0);
    \end{tikzpicture}
    }%
    \\[1em]
    $\G \projshort{\subscription}{\rolename{D}}$ &
    \raisebox{-5mm}{%
    \begin{tikzpicture}[gt, scale = \ifbool{ecoop}{.6}{1}, node distance = 4mm and 15mm, font=\footnotesize]
        \node (0) [plb, minimum size = 4mm]{};
        \node (1) [below =0.3cm of 0, minimum size = 4mm] {};
        \node (2) [plb, right = of 0, minimum size = 4mm] {};
        \node (3)  [plb, left = of 0, minimum size = 4mm] {};

        \draw[->] (1) to node[sloped, anchor=center, above=1mm]{}  (0);
        \path (0) edge[loop above] node{$\lgtn{closingTime}!$}  (0);
        \path (0) edge[above] node{$ \lgtn{closingTime}?$}  (3);
        \path (0) edge[sloped, above, bend left] node{$ \lgtn{partReq}?$}  (2);
        \path (2) edge[sloped, below, bend left] node{$  \lgtn{partOK}?$}  (0);
    \end{tikzpicture}
    }%
    \\[1em]
    $\G \projshort{\subscription}{\rolename{FL}}$ &\hspace*{-7mm}
    \raisebox{-5mm}{%
    \begin{tikzpicture}[gt, scale = \ifbool{ecoop}{.6}{1}, node distance = 4mm and 15mm, font=\footnotesize]
        \node (0) [plb, minimum size = 4mm]{};
        \node (1) [left =0.5cm of 0, minimum size = 4mm] {};
        \node (2) [plb, above right = of 0, minimum size = 4mm] {};
        \node (3)  [plb, right = of 2, minimum size = 4mm] {};
        \node (4) [plb, below right = of 3, minimum size = 4mm] {};

        \draw[->] (1) to node[sloped, anchor=center, above=1mm]{}  (0);
        \path (0) edge[sloped, below] node{$ \lgtn{closingTime}?$}  (4);
        \path (0) edge[sloped, above] node{$ \lgtn{partReq}?$}  (2);
        \path (2) edge[loop above] node{$\lgtn{pos}!$} ();
        \path (2) edge[sloped, below, bend right] node{$  \lgtn{pos}?$}  (3);
        \path (3) edge[sloped, above, bend right] node{$  \lgtn{partReq?}$}  (2);
        \path (3) edge[sloped, above] node{$\lgtn{closingTime}?$} (4);
    \end{tikzpicture}
    }%
  \end{tabular}
  \short{\vspace{-1mm}}
  \caption{Projection of $\G = \Gn{Warehouse}$ (from \Cref{fig:Warehouse}) on roles $\rolename{T}, \rolename{D}$, and $\rolename{FL}$,
  using the subscription $\subscription$ defined in \Cref{ex:bt-swarm-warehouse}.}
  \label{fig:ProjWH}
  \short{\end{minipage}}
\end{figure}

Notice that in \Cref{def:swarm-proto-projection},
when constructing the projection $\G \proterm$
there might be multiple $\lgt$ and $\G_i$ such that
$\G \xrightarrow{\lgtn{l}} \xrightarrow{\evt_{i}} \G_i$,
meaning the projected machine can accept or emit $\evt_i$ and then reach one of potentially many different options for $\G_i \proterm$. %
Still, \Cref{lem:equalprojections} ensures that if $\G$ is well-formed, then $\G \proterm$ relates $\G$ to only one unique machine.%

\begin{lemma}\label{lem:equalprojections}
  If $\G$ is a $\subscription$-well-formed protocol, then $\G \proterm$ is uniquely defined for all
  $\arole[]$.
\end{lemma}
\begin{proof}
	  \emph{(Outline; full proof in \Cref{sec:app:proof_equalprojections})}
    The $\subscription$-well-formedness property ensures that the
		projection satisfies $\G' \proterm = \G'' \proterm$
		for all sub-terms $\G'$ and $\G''$ reachable from $\G$
		by event types not in $\subscription(\role)$%
  -- so it does
  not matter which one is chosen in the construction of $\G \proterm$. This is
  because the subscription $\subscription$ includes all event types that can
  influence the future behaviour of the projection -- so in
  \Cref{def:swarm-proto-projection}, the event types selected in $\lgt$ (which
  are not in $\subscription$)
  do not influence the shape of $\G_i \proterm$.
\end{proof}

\subsection{Branch-Tracking Machine Semantics}
\label{sec:branch-tracking-machines}

In this section we define how machines emit events and process their local log of events.
We fix a set of \emph{events}, ranged over by $\aevent$; each event has a
unique \emph{event type} $\evt$. An \emph{event log
  $\alog$} is a \ecoop{sequence $\mklog{\aevent[1],\aevent[2], \ldots}$ of
  pairwise distinct events}.
As anticipated in \Cref{sec:overview:approach}, we incorporate a new \emph{branch-tracking mechanism} to the machine semantics
originally proposed in \cite{DBLP:conf/ecoop/KuhnMT23} to establish an explicit causal link between events.
Such a link was not present nor needed in \cite{DBLP:conf/ecoop/KuhnMT23} 
-- but it is necessary for swarms to work correctly with our 
compositional well-formedness (\Cref{def:wf}), which removes restrictions assumed by the machine semantics in \cite{DBLP:conf/ecoop/KuhnMT23}.

The behaviour of a machine $\afish$ is determined by a \emph{log-processing
  function} $\hat\stchange$ (formalised in \Cref{def:branch-tracking-transition-function-machines} below).
Let $\typeset$ be the set of all event types. In defining $\hat\stchange$, we assume that machines are equipped with an additional element, which  does not appear in \cite{DBLP:conf/ecoop/KuhnMT23}:
A set of \emph{updating event types} $\mup{\M}\subseteq\typeset$.\footnote{We will see in \Cref{def:realisation} that this set is populated with branching, joining, and looping event types from a swarm protocol.}
If an event $\aevent$ has a type $\evt \in \mup{\M}$, we call $\aevent$ an \emph{updating event}.
We also add to each event $\ev$ a field $\ev.\lbjePointer$, which intuitively refers to the last updating event that caused $\ev$.
A machine $\M$ processes its log starting from the oldest event. When processing an event $\ev$ of type $\evt$, $\M$ performs a transition if:
\begin{enumerate}
\item There is an acceptance transition  $\afish \xrightarrow{\inp} \afish'$; i.e., $\afish$ accepts events of type $\evt$; and
\item $\ev.\lbjePointer$ coincides with the last updating event observed by $\M$; i.e., $\ev.\lbjePointer$ is the last updating event that causally preceded $\ev$.
\end{enumerate}
If the two conditions above hold, then $\M$ accepts
$\ev$ and moves to the next state $\M'$, i.e., the continuation of the event acceptance transition, from which the remaning events in the log are processed. Otherwise, $\M$ skips $\ev$ and processes the next event in the log.

More precisely, in order to remember previously-seen updating events and check the pointer $\ev.\lbjePointer$, each machine keeps
a partial mapping $\lbje : \typeset \rightharpoonup \eventset$ that
associates each event type $\evt$ with the \emph{last updating event} %
for the type $\evt$.
The mapping $\lbje$ is built by the machine while processing a log: initially, $\lbje$
is empty (i.e., $\emptyset$) and it is updated every
time a machine processes an updating event.
The updating of $\lbje$ is performed as follows. %
   Let $\mbranch{\M, \evt}$  be a partial mapping that, for each machine $\M$ and event type $\evt$, gives the set event types that continue along the same branch as $\evt$ until the next updating event;
specifically, if $\M \xrightarrow{\evt\,?} \xrightarrow{\evt_1\,?} \ldots \xrightarrow{\evt_n\,?}$ with $\evt_1,\ldots,\evt_{n-1}\notin \mup{\M}$ and $\evt_1,\ldots,\evt_n$ not concurrent with $\evt$, then $\evt_1,\ldots,\evt_n \in \mbranch{\M, \evt}$.
The processing of an event $\ev$ of type $\evt$ by a machine $\M$ that has an event acceptance transition
$\M \xrightarrow{\inp[][\evt]} \M'$ is described by the following relation:

  \smallskip%
  \centerline{\(
    \begin{array}{lcl@{\quad}l}
    \parSt &\xrightarrow{\ev}& \parSt[\M']  & \text{if $\ev.\lbjePointer = \lbje(\evt)$ and $\evt \notin \mup{\M}$}\\
    \parSt &\xrightarrow{\ev}& \parSt[\M'][{\lbje[\mbranch{\M,\evt} \myupd \ev ]}] & \text{if $\ev.\lbjePointer = \lbje(\evt)$ and $\evt \in \mup{\M}$.}
    \end{array}
  \)}%
  \smallskip%

\noindent%
where $\lbje[\typeSet \myupd \ev]$ denotes the update of $\lbje$, namely the
mapping that maps each event type in the set $\typeSet$ to event $\ev$, and behaves as $\lbje$ for other event types.

The mechanism through which a machine $\M$ processes a complete event log $\alog$,
or equivalently determines its current state based on $\alog$, is specified by \Cref{def:branch-tracking-transition-function-machines} below.

\begin{definition}[Log-processing function]
  \label{def:branch-tracking-transition-function-machines}  The \emph{log-processing function} $\hat\stchange$
is defined as $\hat\stchange(\M,\alog) = \stchangeBT(\M,\alog,\emptyset)$ where:

  \smallskip%
  \centerline{\(
  \begin{array}{r@{\;}c@{\;}l}
      \stchangeBT(\M,\emptylog, \lbje) &=& (\M, \lbje)\\
      \stchangeBT(\M,\aevent\logcat\alog, \lbje) &=&
        \begin{cases}
          \stchangeBT(\M',\alog, \lbje') &\text{if $\parSt \xrightarrow{\ev} \parSt[\M'][\lbje']$}\\
          \stchangeBT(\M,\alog, \lbje)   & \text{otherwise.}
        \end{cases}
  \end{array}
  \)}%
\end{definition}

Note that in \Cref{def:branch-tracking-transition-function-machines}, $\hat\stchange$ is defined in terms of the auxiliary
function $\stchangeBT$, which is in charge of constructing $\lbje$
while traversing the log.
Observe that $\stchangeBT$, and consequently $\hat\stchange$,
returns a pair $(\M', \lbje)$ where $\M'$ is the reached machine and
$\lbje$ is the mapping built after processing the log.
Based on the current state computed by the log-processing function \(\hat\stchange\), machines generate new events according to \Cref{def:operational-semantics-machines} below.

\begin{definition}[Operational semantics of machines]
  \label{def:operational-semantics-machines} 
  A machine $\M$ with a local log $\alog$ emits events according to the following rule:

  \smallskip%
  \centerline{\(%
  \begin{array}{c}
    \mathrule{
    \hat\stchange(\M,\alog) = (\kappa \cdot  \&_{i\in I} \inp[i] \afish[i],\  \lbje)
    \qquad
    \evt \in \kappa
    \qquad
    \ev \text{ has type } \evt
    \qquad
    \evP{\ev} = \lbje(\evt)
    }{
    (\M,\alog) \xRightarrow{\evt!} (\M,\alog\logcat\aevent)
    }{emit}
  \end{array}
  \)}%
\end{definition}

The transition
$(\M,\alog) \xRightarrow{\evt!} (\M,\alog\logcat\aevent)$ in \Cref{def:operational-semantics-machines} reads as:
\emph{``a machine $\M$ with a log $\alog$ emits an event $\ev$ of type $\evt$.''}
Observe that the new event $\ev$ is added to the end of the local log, and %
the transition is possible only if, after processing $\alog$, the machine reaches a state $\kappa \cdot \&_{i\in I} \inp[i] \afish[i]$ where the emission of events of type $\evt$ is enabled (i.e., $\evt \in \kappa$). Moreover, $\ev$ must be of type $\evt$ and must reference the last updating event for $\evt$.

\subsection{Swarms}\label{sec:formalisms:swarms}

A \emph{swarm} (of size $n$) is a pair $(\asysrt, \alog)$ where:
\begin{itemize}
\item $\asysrt$ maps indices $i \in \{1, \ldots, n\}$ to machines and their local
log, i.e. $\asysrt(i) = (\afish[i],\alog[i])$;
\item $\alog$ is a \emph{global log}, i.e., a sequence consisting of all
  events generated by the machines in $\asysrt$ that preserves their order of
  generation. More formally:
  \begin{itemize}
  \item Every log (global or local)
	 $\alog' = \aevent[1] \cdots \aevent[n]$ totally orders its events using a relation $<_{\alog'}$ defined as 
	 $\aevent[i] < _{\alog'}\aevent[j]$ iff $i < j$;
    \item The global log $\alog$ contains all and only the events of the local logs $\alog[i]$ (for all $i \in \{1, \ldots, n\}$) and preserves their respective ordering: $\bigcup_{i \in \{ 1, \ldots, n \}}\alog[i] = \alog$ and $ <_{\alog[i]} \subseteq <_{\alog}$.
  \end{itemize}
\end{itemize}
For brevity, we may write $\asysrt$ instead of $(\asysrt, \alog)$ (and call $\asysrt$ a swarm) if the global log $\alog$ is empty.

\begin{remark}
  \label{remark:global-log-trick}%
  The global log $\alog$ of a swarm $(\asysrt, \alog)$ does not exist in
  actual swarm implementations: it is a formalisation device which models the global ordering and
  non-deterministic propagation of events, as defined by swarm semantics formalised in
  \Cref{def:formalisms:swarm-semantics} below.
\end{remark}

\begin{definition}[Swarm semantics]
  \label{def:branch-tracking-swarm-semantics}
  \label{def:formalisms:swarm-semantics}
  The behaviour of a swarm is defined by the smallest relation closed under
  the following rules:

  \smallskip%
  \centerline{\(%
    \mathrule{
    \asysrt(i) = (\afish[i],\alog[i])
    \quad
    (\afish[i],\alog[i]) \xRightarrow{\aeventtype!} (\afish[i],\alog[i]\cdot \aevent)
    }{
    (\asysrt, \alog \cdot \alog') \red[{\aeventtype}!] (\upd \asysrt i {(\afish[i], \alog[i]\cdot \aevent)},  \alog \cdot \aevent \cdot \alog')
    }{Local}
  \qquad%
    \mathrule{
    \asysrt(i) = (\afish[i],\alog[i])
    \quad
    <_{\alog[i] }  \subset <_{\alog'} \subseteq <_{\alog}
  }{
    (\asysrt, \alog) \red[\tau] (\upd \asysrt i {(\afish[i], {\alog'})}, \alog)
  }{Prop}
  \)}%
  \smallskip%
\end{definition}

In \Cref{def:branch-tracking-swarm-semantics}, rule \textsc{Local}
formalises the emission of an event enabled at the
machine at $\asysrt(i)$ of the swarm $(\asysrt, \alog)$. The emitted
event $\aevent$ is appended to the local log of $\asysrt(i)$ and
non-deterministically merged in the global log respecting the order of
events previously generated by $\asysrt(i)$.
Rule \textsc{Prop} models the asynchronous
event log propagation among machines: \ecoop{for some machine $\afish[i]$ with a local log
$\alog[i]$ (that is strictly included in the global log $\alog$), the
rule non-deterministically selects a log
$ \alog' \subseteq \alog$ that contains $\alog[i]$ and transfers the events in $\alog'$ to
$\alog[i]$ while preserving their order.}
Note that during an execution, machines can emit events
independently, and their local log may only partially reflect the ``global state''
in the global log (due to the non-deterministic and asynchronous
propagation of events).

\Cref{def:realisation} below formalises the \emph{realisation} of a swarm protocol $\G$, i.e.,
a swarm whose machines interact according to the roles $\G$.

\begin{definition}[Realisation of a swarm protocol]
	\label{def:realisation}%
  A swarm $(\S, \emptylog)$ of size $n$
  \emph{realises} a protocol $\G$
with respect to a subscription $\subscription$ iff for all $i \in \{1, \ldots, n\}$, ($i$) $\S(i) = (\M[i], \emptylog)$, and $(ii)$
 $\M[i] = \G\projshort{\subscription}{\role}$ for some $\role \in \G$; and ($iii$) $\mup{\M[i]}$ is
 equal to the event types of $\G$ that are branching, joining, or looping in $\subscription$.
 We also say that \emph{$\M[i]$ realises $\role$}.
\end{definition}

\begin{example}[Swarm protocol realisation and execution]\label{ex:bt-swarm-warehouse}
  Let $E$ denote the set of all event types in the $\Gn{Warehouse}$%
  protocol (\Cref{fig:Warehouse}) and  $\subscription = \{ \rolename{T} \mapsto E, \rolename{FL} \mapsto E
  \setminus \{ \lgtn{partOK} \}, \rolename{D} \mapsto E \setminus \{
  \lgtn{pos} \} \}$.
  By \Cref{def:wf}, $\Gn{Warehouse}$ is $\subscription$-well-formed.
  By \Cref{def:realisation}, the swarm $\asysrt_\Gn{W}$ below realises $\Gn{Warehouse}$ w.r.t. $\subscription$:

  \smallskip%
  \centerline{\(%
     \asysrt_\Gn{W} \;=\; \left\{
      i_{{\rolename{T}}_1} \mapsto (\afish[{\rolename{T}}], \emptylog),\quad
      i_{{\rolename{T}}_2} \mapsto (\afish[{\rolename{T}}], \emptylog),\quad
      i_{{\rolename{FL}}} \mapsto (\afish[{\rolename{FL}}], \emptylog),\quad
      i_{{\rolename{D}}} \mapsto (\afish[{\rolename{D}}], \emptylog)
    \right\}
    \)}%
    \smallskip%

  \noindent%
  where $\afish[{\rolename{T}}] = \Gn{Warehouse}\projshort{\subscription}{\rolename{T}}$, $\afish[{\rolename{FL}}] = \Gn{Warehouse}\projshort{\subscription}{\rolename{FL}}$,
  and $\afish[{\rolename{D}}] = \Gn{Warehouse}\projshort{\subscription}{\rolename{D}}$ (shown in \Cref{fig:ProjWH}). Notice that the machine in
  	$\asysrt_\Gn{W}(i_{{\rolename{T}}_1})$ and
  	$\asysrt_\Gn{W}(i_{{\rolename{T}}_2})$ play the same role
  	$\rolename T$.
  By \Cref{def:branch-tracking-swarm-semantics}, a possible execution of the swarm $\asysrt_\Gn{W}$ is:
  
  \smallskip%
  \scalebox{0.7}{%
	 \(%
		\begin{array}{r@{\;}r@{\;}l@{\quad}c@{\;}r@{\;}c@{\;}l}
		  & (\asysrt_\Gn{W}, \emptylog) \xRightarrow{\lgtn{partReq}!} & (\asysrt_1, \lgname{partReq}_1) \quad && \asysrt_1 &=& \asysrt_\Gn{W}{[i_{{\rolename{T}}_1} \myupd (\afish[{\rolename{T}}], \lgname{partReq}_1)]}\\
		  & \xRightarrow\tau& (\asysrt_2, \lgname{partReq}_1) \quad && \asysrt_2  &=& \asysrt_1{[i_{{\rolename{FL}}_1} \myupd (\afish[{\rolename{FL}}], \lgname{partReq}_1)]}\\
		  & \xRightarrow{\lgtn{pos}!}& (\asysrt_3, \lgname{partReq}_1\logcat\lgname{pos}_2) \quad && \asysrt_3 &=& \asysrt_2{[i_{\rolename{FL}} \myupd (\afish[{\rolename{FL}}], \lgname{partReq}_1\logcat\lgname{pos}_2)]}\\
		  & \xRightarrow\tau& (\asysrt_4, \lgname{partReq}_1\logcat\lgname{pos}_2)\quad && \asysrt_4 &=& \asysrt_3{[i_{{\rolename{T}}_1} \myupd (\afish[{\rolename{T}}], \lgname{partReq}_1\logcat\lgname{pos}_2)]}\\
		  & \xRightarrow \tau & (\asysrt_5, \lgname{partReq}_1\logcat\lgname{pos}_2)\quad && \asysrt_5 &=& \asysrt_4{[i_{{\rolename{T}}_2} \myupd (\afish[{\rolename{T}}], \lgname{partReq}_1\logcat\lgname{pos}_2)]}\\
		  & \xRightarrow{\lgtn{partOK}!} & (\asysrt_6, \lgname{partReq}_1\logcat\lgname{pos}_2\logcat\lgname{partOK}_3) \quad && \asysrt_6 &=& \asysrt_5{[i_{{\rolename{T}}_1} \myupd (\afish[{\rolename{T}}], \lgname{partReq}_1\logcat\lgname{pos}_2\logcat\lgname{partOK}_3)]}\\
		  & \xRightarrow{\lgtn{partReq}!} & (\asysrt_7, \lgname{partReq}_1\logcat\lgname{pos}_2\logcat\lgname{partOK}_3\logcat\lgname{partReq}_4) \quad && \asysrt_7 &=& \asysrt_6{[i_{{\rolename{T}}_1} \myupd (\afish[{\rolename{T}}], \lgname{partReq}_1\logcat\lgname{pos}_2\logcat\lgname{partOK}_3\logcat\lgname{partReq}_4)]}\\
		  & \xRightarrow {\lgtn{partOK}!} & (\asysrt_8, \alog) \quad && \asysrt_8 &=& \asysrt_7{[i_{{\rolename{T}}_2} \myupd (\afish[{\rolename{T}}], \lgname{partReq}_1\logcat\lgname{pos}_2\logcat\lgname{partOK}_5)]}\\
		  & \xRightarrow \tau^*& (\asysrt_9, \alog) \quad && \asysrt_9 &=& \asysrt_8{[i_{{\rolename{T}}_1} \myupd (\afish[{\rolename{T}}], \alog),i_{{\rolename{T}}_2} \myupd (\afish[{\rolename{T}}], \alog),i_{{\rolename{FL}}} \myupd (\afish[{\rolename{FL}}], \alog), i_{{\rolename{D}}} \myupd (\afish[{\rolename{D}}], \alog)]}\\
		\end{array}
	 \)
  }
  \smallskip%

  \noindent%
  where, in states $\S_8$ and $\S_9$, the global log is
  $\alog =
  \lgname{partReq}_1\logcat\lgname{pos}_2\logcat\lgname{partOK}_3\logcat\lgname{partReq}_4\logcat\lgname{partOK}_5$
  and its events carry the following pointers:

  \smallskip%
  \centerline{\(%
	 \arraycolsep=2pt
		\begin{array}{rclllllll}
			\lgname{partReq}_1.\lbjePointer &=& \perp \text{(undefined)}&&&&&&\\
			\lgname{pos}_2.\lbjePointer &=& \lgname{partOK}_3.\lbjePointer &=&
			\lgname{partReq}_4.\lbjePointer &=& \lgname{partOK}_5.\lbjePointer
			&=& \lgname{partReq}_1
		\end{array}
  \)}%
  \smallskip%

  In each state of the swarm, we compute the state of a machine for its local log
  using $\hat\stchange$
  (\Cref{def:branch-tracking-transition-function-machines}), %
  which in turn uses the
  branching, joining, and looping events in $\alog$ to update the mapping
  $\lbje$ and only considers events carrying correct pointers.
  For this example, when the machines process the event
  $\lgname{partReq}_4$, they will update their mapping $\lbje$ to
  obtain:

  \smallskip%
  \centerline{\(%
	 \lbje[\lgtn{partReq},\, \lgtn{pos},\, \lgtn{partOk},\, \lgtn{closingTime}  \myupd \lgname{partReq}_4]
  \)}%
  \smallskip%

  \noindent%
  In this updated mapping, all event types of
  $\Gn{Warehouse}$ point to $\lgname{partReq}_4$, which means
  that the machines will
  only accept new events $\ev$ such that
  $\ev.\lbjePointer = \lgname{partReq}_4$.
  Observe that, in the execution above, the event $\lgname{partOK}_5$ was emitted by the machine with ID $i_{{\rolename{T}}_2}$
  in response to the event $\lgname{partReq}_1$.
  Hence, $\lgname{partOK}_5.\lbjePointer = \lgname{partReq}_1 \neq
  \lgname{partReq}_4$;
  i.e., the event $\lgname{partOK}_5$ is causally linked to $\lgname{partReq}_1$
  through branch tracking.
  Since the machines expect events,
  whose $\lbjePointer$ fields refer to $\lgname{partReq}_4$,
  the event $\lgname{partOK}_5$ is ignored by all
  machines after the logs propagate.
  Once this happens in state $\asysrt_9$,
  the machines reach the states:

    \smallskip%
  \centerline{\(%
  \begin{array}{rcl}
	  \hat\stchange(\afish[{\rolename{T}}], \changeNoMargin{\alog}) &=& (\cstmcmd{FL}{pos}.\cstmcmd{T}{partOK}.\Gn{Warehouse})\projshort{\subscription}{\rolename{T}}
    \\
	  \hat\stchange(\afish[{\rolename{D}}], \alog) &=& (\cstmcmd{T}{partOK}.\Gn{Warehouse})\projshort{\subscription}{\rolename{D}}
    \\
	  \hat\stchange(\afish[{\rolename{FL}}], \changeNoMargin{\alog}) &=& (\cstmcmd{FL}{pos}.\cstmcmd{T}{partOK}.\Gn{Warehouse})\projshort{\subscription}{\rolename{FL}}
  \end{array}
  \)}%
  \smallskip%
  
  \noindent%
  Here, the machines $\afish[{\rolename{T}}]$ and $\afish[{\rolename{FL}}]$
  expect the forklift to
  emit an event of type $\lgtn{pos}$, while $\afish[{\rolename{D}}]$
  (which \ecoop{does not subscribe} to $\lgtn{pos}$ in
  $\subscription$) is waiting for an event of type
  $\lgtn{partOK}$ that has the correct pointer and thus follows $\lgtn{pos}$ in the swarm protocol
  $\Gn{Warehouse}$.
  These states are consistent with each other: $(\cstmcmd{T}{partOK}.\Gn{Warehouse})\projshort{\subscription}{\rolename{D}}=( (\cstmcmd{FL}{pos}.\cstmcmd{T}{partOK}.\Gn{Warehouse})\projshort{\subscription}{\rolename{D}}$.%
  Thus they correspond
  to a valid run of the $\Gn{Warehouse}$ swarm protocol.
\end{example}

\begin{remark}
Running the swarm of \Cref{ex:bt-swarm-warehouse} with the semantics
of~\cite{DBLP:conf/ecoop/KuhnMT23} can lead to a state where
the machines irremediably diverge due to the absence of branch tracking
and the resulting lack of event causality tracking.
We illustrate this situation in \Cref{sec:app:non_bt_swarm}.
\end{remark}

\subsection{Eventual Fidelity of Swarm Protocol Realisations}
\label{sec:eventual-fidelity}

In this section, we address the problem of ensuring that a swarm behaves correctly with respect to a swarm protocol $\G$.
As our correctness criterion, we adopt the notion of
\emph{eventual fidelity} (\Cref{def:eventual-fidelity} below).
Intuitively, a swarm is eventually faithful to a protocol $\G$ if, in every execution, once all machines have received all messages,
they agree on the same \emph{consistent view} of the execution path of $\G$.
By consistent view, we mean that machines are able to determine which events belong to the effective execution path through $\G$ and disregard all other events that may have been emitted inconsistently.
We call such a view the \emph{effective log} (see \Cref{sec:app:effective-logs} for formal definitions).
When considering a global log $\alog$ generated by a swarm, the effective log w.r.t.~$\G$, written $\srteffof[\lg][\G]$, is the
sublog
of $\alog$ that
consists solely of those events of $\alog$ that follow a valid path through $\G$.
We also define the effective log with
respect to a specific machine $\M$, written $\srteffof[\lg][\M]$, as
the sublog of $\alog$ containing only the events that $\M$ accepts
(i.e., all the events not ignored by the log-processing function $\hat\stchange$).

To formalise the eventual fidelity of a swarm to a protocol $\G$,
we assign to every machine in the swarm some role of $\G$: we denote this by a mapping $\roleMap$ from machines to the roles they play. %
Eventual fidelity requires that if a machine $\M$ that plays role $\roleMap(\M)$ is provided with a global log $\lg$ produced by the swarm, then $\M$ accepts only those events which are in the global effective log and to which $\roleMap(\M)$ subscribes.
We denote such events by $\srteffof[\lg][\G] \downarrow_{\subscription(\roleMap(\M))}$, which is the projection of $\srteffof[\lg][\G]$ onto the event types that role $\roleMap(\M)$ subscribes to in $\subscription$.
\begin{definition}[Eventual fidelity]
  \label{def:eventual-fidelity}
   A swarm $(\S,\epsilon)$  of size $n$ is
  \emph{eventually faithful} to a protocol $\G$ with respect to a subscription $\subscription$ if
        for all  $\alog$ such that $(\S,\epsilon)\rightarrow^* (\S',\lg)$, for all %
        $\role \in \G$, and for all $i \in \{ 1...n \}$, we have
   $\srteffof[\lg][\G] \downarrow_{\subscription(\roleMap(\S(i)))} = \srteffof[\lg][\S(i)]$.
\end{definition}

\Cref{thm:eventual_fidelity} below ensures that a realisation of a swarm protocol $\G$ (i.e., a swarm 
consisting of machines projected from $\G$, by \Cref{def:realisation}) is eventually faithful to $\G$. The proof is in \Cref{sec:app:fidelity-proof}.

\begin{theorem}[Eventual fidelity of swarm protocol realisations]\label{thm:eventual_fidelity}
  If a swarm $\S$ realises a $\subscription$-well-formed protocol $\G$,
  then $\S$ is eventually faithful to $\G$ with respect to $\subscription$.
\end{theorem}

\begin{example}[Effective logs and eventual fidelity]
    Consider the swarm protocol $\Gn{Warehouse}$ (\Cref{fig:Warehouse}),
    its projected machine $\afish[{\rolename{D}}] = \Gn{Warehouse}\projshort{\subscription}{\rolename{D}}$,
    (\Cref{fig:ProjWH}),
    and the log $\alog $ from \Cref{ex:bt-swarm-warehouse}. %
    By \Cref{def:global-effective-log,def:machine-proto-role-effective-log}, the effective logs for $\Gn{Warehouse}$
    and $\M_\rolename{D}$ filtered from $\alog$ are:

  \smallskip%
  \centerline{\(%
    \begin{array}{rcl}
    \srteffof[\alog][\Gn{Warehouse}] &=& \lgname{partReq}_1\logcat\lgname{pos}_2\logcat\lgname{partOK}_3\logcat\lgname{partReq}_4
    \\
    \srteffof[\alog][\M_\rolename{D}] &=& \lgname{partReq}_1\logcat\lgname{partOK}_3\logcat\lgname{partReq}_4
    \end{array}
  \)}%
  \smallskip%

  \noindent%
  Therefore, we have $\srteffof[\alog][\Gn{Warehouse}] \downarrow_{\subscription(\rolename{D})} = \srteffof[\alog][\M_\rolename{D}]$:
  hence, the effective logs for the machine and the swarm protocol align, as required by eventual fidelity (\Cref{def:eventual-fidelity}).
\end{example}

\section{Composing Machines and Swarms}
\label{sec:swarm-composition}%
\label{sec:swarmcomposition}

\noindent%
In this section, we present an approach for realising a composition of swarm protocols
$\G_1 \conc \ldots \conc \G_n$ by composing swarms of pre-existing machines which realise the individual protocols $\G_i$ (for $i \in \{1, \ldots, n\}$). %
This is crucial for enabling the reuse of pre-existing machines and their deployment in new, larger swarms %
(that we implement in \Cref{sec:implementation}). %
We first formalise the composition of machines (\Cref{sec:machine-composition}); %
then, we introduce an automatic \emph{adaptation} mechanism (\Cref{sec:machine-adaptation}) for such pre-existing machines, enabling us to compose entire swarms (\Cref{alg:swarm_composition_short}) while guaranteeing eventual fidelity to the composed swarm protocol (\Cref{lem:projection_composition_correct}).

\subsection{Machine Composition}
\label{sec:machine-composition}%
\label{sec:swarmcompositionalgo}%

Similar to swarm protocol composition (\Cref{def:swarm-proto-composition}), in
\Cref{def:machine-composition} below we define the composition of two machines $\M$ and $\M'$ as a synchronised product. %
Intuitively, if a transition is common to $\M$ and $\M'$, %
the machines synchronise and the continuation of $\M \concM{} \M'$ is given by the composition of both corresponding continuations of $\M$ and $\M'$ (clause \emph{Common Events});
otherwise, the continuation of $\M \concM{} \M'$ is the composition of
the continuation of either $\M$ or $\M'$ with the other machine unchanged
(clauses \emph{Events Unique}).

\begin{definition}[Machine composition]
    \label{def:machine-composition}%
    Let $\machine \define \map \cdot \&_{i \in I} \evt_i ? \machine_i$ and $\machine' \define \map' \cdot \&_{j \in J} \evt_j' ? \machine_j'$ be
    machines. Let $\typeSet$ be the set of event types that occur in both $\M$ and $\M'$.
    The \emph{composition of $\machine$ and $\machine'$}, written $ \machine  \concM{} \machine'$, is defined as $\machine  \concM{\typeSet} \machine'$, which in turn is coinductively defined:
    
    \smallskip%
    \centerline{\(
    	\machine  \concM{\typeSet} \machine' \;\;\stackrel{co}{\define}\;\; \map'' \cdot \&_{k \in K}  \evt_k ? \machine_k
    \)}%
    \smallskip%

    \noindent%
     where $K $ is the smallest set such that for all $i \in I, j \in J$:
    \begin{description}
        \item[Common Events ($\evt_i \in \typeSet,\; \evt_i =\evt_j'$):] There is a $k\in K$ with
        $\evt_k = \evt_i $ and $ \machine_k   = \machine_i  \concM{\typeSet} \machine_j'$.
        \item[Events Unique to $\machine$ ($\evt_i \notin \typeSet$):] There is a $k\in K$ with $\evt_k=\evt_i$ and
        $ \machine_k   = \machine_i  \concM{\typeSet} \machine'$.
        \item[Events Unique to $\machine'$ ($ \evt_j' \notin \typeSet$):] There is a $k\in K$ such that
        $\evt_k=\evt_j'$ and
        $ \machine_k = \machine \concM{\typeSet} \machine_j'$.
      \end{description}
    Then, the emitter set $\map''$  is  $\map''=(\map\cup \map') \cap \{\evt_k \mid k\in K\}$.
\end{definition}

Note that by \Cref{def:machine-composition}, %
the acceptance transitions of $\machine \concM{} \machine'$ consist of the acceptance transitions
of both $\M$ and $\M'$; common transitions are synchronised via the set $\typeSet$, hence they are only enabled in a state if both $\M$ and $\M'$ can take them. %
Further, an event type $\evt$ belongs to the emitter set of
$\machine \concM{} \machine'$ if and only if ($i$) $\evt$ belongs to the
emitter set of either $\M$ or $\M'$, and ($ii$) either $\evt \in \typeSet$ and both machines accept $\evt$, or $\evt \not\in \typeSet$ and one of the machines accepts $\evt$.
This means a machine either implements an interfacing role and thus only emits events in $\typeSet$ or it is non-interfacing and thus only emits events not in $\typeSet$.
Also note that, as with swarm protocol composition (\Cref{def:swarm-proto-composition}), we do not introduce new syntax for composed machines: composed machines are treated as ordinary machines, to maintain backward compatibility with existing
tooling (that we extend in \Cref{sec:implementation}).

\begin{example}[Relating swarm protocol composition, projection, and machine composition]%
    Consider the swarm protocols $\Gn{Warehouse}$ and $\Gn{Factory}$ in \Cref{fig:Warehouse,fig:Factory}, and the subscription $\subscription$ from \Cref{ex:projection}. %
    Also consider the projected machines $\Gn{Warehouse} \projshort{\subscription}{\rolename{D}}$ in \Cref{fig:ProjWH}, and $\Gn{Factory} \projshort{\subscription}{\rolename{D}}{} = \lgtn{partReq}?\cdot \lgtn{partOK}?\cdot \zero$. If we compose these two machines (using \Cref{def:machine-composition}), then we derive the machine denoted as $\G \projshort{\subscription}{\rolename{D}}$ in \Cref{fig:ProjWHF}, and therefore we have:

    \smallskip%
    \centerline{\(%
    \Gn{Factory} \projshort{\subscription}{\rolename{D}}{} \,\concM{}\, \Gn{Warehouse} \projshort{\subscription}{\rolename{D}}{} \;\;=\;\; (\Gn{Factory}\conc\Gn{Warehouse}) \projshort{\subscription}{\rolename{D}}
    \)}%
    \smallskip%

    \noindent%
    Note that this holds because, in $\subscription$, the roles of both projected machines subscribe to the interfacing event types that synchronise the composed swarm protocols (by \Cref{def:swarm-proto-composition}).
\end{example}

\subsection{Machine Adaptation}
\label{sec:machine-adaptation}
In \Cref{def:machine-adaptation} we introduce a function $\mathcal{A}$
that \emph{adapts} a single machine $\M$, projected from a swarm protocol $\G_k$, so that it behaves correctly when used in a swarm
with other adapted machines that were projected from swarm protocols $\G_1, \ldots, \G_n$.
Intuitively, the adaptation function $\mathcal{A}$ restructures the transition graph of $\M$ in two steps (described in more detail below): (1) it adds automatically generated acceptance transitions to wait for extra synchronisation event types required by a composed subscription $\subscription$, and (2) it disables transitions whose continuations are incompatible with the other projections. Notably, $\mathcal{A}$ does not rewrite the logic of $\M$:
the adapted machine looks different from $\M$ (see \Cref{ex:machine-adaptation} below) but is still just the original automaton with a thin structural wrapper that delays progress until required synchronisation events are emitted by other machines, or disables incompatible paths.

\begin{definition}[Machine adaptation for swarm composition]
	\label{def:machine-adaptation}%
	Let $\G_1, \ldots, \G_n$ be swarm protocols with subscriptions
	$\subscription_1,\ldots,\subscription_n$.
	Let $\M$ be a machine that, for some $k \in \{1, \ldots, n\}$, realises role $\role$ in the swarm protocol $\G_k$
	for subscription $\subscription_k$.
	Let $\subscription$ be a well-formed subscription for $\G_1 \conc \ldots \conc \G_n$
	(obtained \longer{from $\subscription_1,\ldots,\subscription_n$}
	via~\Cref{alg:comp-swarm-subscriptions}).
	Then, we define the %
	\emph{machine adaptation $\mathcal{A}$} as:
	
	\smallskip%
	\centerline{\(
		\mathcal{A}\!\left(\M, (\G_i)_{i \in \{ 1, \ldots, n \}},\role, k, \subscription\right)%
		\;=\;%
		\left(%
		\G_1 \projshort{\subscription}{\role} \concM{} \ldots \concM{}%
		\left(\M \concM{} \G_k \projshort{\subscription}{\role}\right) %
		\concM{} \ldots \concM{} \G_n \projshort{\subscription}{\role}
		\right).
		\)}%
	\smallskip%

	\noindent%
	We also set $\mup{\mathcal{A}\!\left(\M, (\G_i)_{i \in \{ 1, \ldots, n \}}, \role, k, \subscription\right)}$
	as the set of branching, joining, or looping event types selected by \Cref{alg:comp-swarm-subscriptions}: %
	this overapproximates the set of branching/joining/looping event types of $\G_1 \conc \ldots \conc \G_n$.
\end{definition}

To better see the effect of the adaptation $\mathcal{A}$ in \Cref{def:machine-adaptation}, observe that by \Cref{alg:comp-swarm-subscriptions} we have $\subscription_k \subseteq \sigma$, hence the subscriptions $\subscription$ and $\subscription_k$ may differ; consequently, $\G_k \projshort{\subscription_k}{\role}$ and $\G_k\projshort{\subscription}{\role}$ may not coincide.
Additionally, $\arole[]$ might be an interfacing role; therefore, it might appear in other protocols besides $\G_k$. Hence, the input machine $\M$ must be adapted to play role $\arole[]$ while conforming to \emph{all} protocols where $\arole[]$ appears.
Consequently, the  adaption of $\M$ involves the following two steps (that were outlined above):
\begin{enumerate}
	\item  $\M$ is  composed with $\G_k\projshort{\subscription}{\role}$ to adapt it to the behaviour expected by the subscription $\subscription$ (which may be larger than $\sigma_k$). This adaptation may force $\M$ to await and accept additional events types, i.e., those in $\subscription(\arole[])$ that are not in $\subscription_k(\arole[])$.
	
	\item $\M$ is also composed with the projections obtained from the remaining protocols of the composition. This step may prune some transitions of $\M$ that are absent in the other protocols and add necessary transitions from the other protocols to ensure correct synchronisation in the composed swarm.
\end{enumerate}

\smallskip%

We extend the adaptation mechanism in \Cref{def:machine-adaptation} to whole swarms: \Cref{alg:swarm_composition_short} takes a set of swarms that respectively realise protocols $\G_1, \ldots ,\G_n$, and automatically generates a swarm that realises $\G_1 \conc \ldots \conc \G_n$. Specifically, \Cref{alg:swarm_composition_short} uses \Cref{alg:comp-swarm-subscriptions} to find a suitable subscription for the composition, and then it applies the adaptation function in \Cref{def:machine-adaptation}. (For a more detailed version of \Cref{alg:swarm_composition_short}, see  \Cref{alg:swarm_composition} in  \Cref{sec:app:swarm_composition_algorithm}.)

\begin{algorithm}[t]
	\noindent\textbf{Input:}
	\begin{minipage}[t]{0.9\linewidth}
		\begin{itemize}
			\item  $\G_1,\ldots,\G_n$ composable protocols, which are respectively realised by swarms $\S_1,\ldots,\S_n$ for the subscriptions $\subscription_1,\ldots,\subscription_n$.
			\item Mappings $\roleMap[1], \ldots, \roleMap[n]$ that respectively map machines to the role they implement, i.e.,  $\roleMapApp[k]{\S_k(i)}$ is the role realised by the machine at $\S_k(i)$.
		\end{itemize}
	\end{minipage}
	\noindent\textbf{Output:} \makebox[\linewidth][l]{A swarm $\S$ that realises $\G_1 \conc \ldots \conc \G_n$.}
\vspace*{-4mm}
{\small
	\begin{algorithmic}[1]
		\STATE Apply \Cref{alg:comp-swarm-subscriptions} to obtain a well-formed subscription $\subscription$ for $\G_1 \conc \ldots \conc \G_n$.
		\STATE For each swarm  $\S_k$, for each machine $\S_k(i)$: adapt $\S_k(i)$ into $\M'=\mathcal{A}\!\left(\M, (\G_i)_{i \in \{ 1, \ldots, n \}}, \roleMapApp[k]{\S_k(i)}, k, \subscription\right)$ according to \Cref{def:machine-adaptation}.
		\STATE Return a swarm built up from the   machines adapted in the previous step.
\end{algorithmic}}
\caption{Swarm composition (shortened version of \Cref{alg:swarm_composition} in \Cref{sec:app:swarm_composition_algorithm}).}\label{alg:swarm_composition_short}
\end{algorithm}

\ecoop{\begin{example}[Composing swarms by adapting their machines]\label{ex:machine-adaptation}

	 Recall the swarm $\S_\Gn{W}$ whose machines are projected from the swarm protocol $\Gn{Warehouse}$ and subscription $\subscription$ in \Cref{ex:bt-swarm-warehouse}.

In order to compose $\S_\Gn{W}$ with another swarm $\S_\Gn{F}$ projected from $\Gn{Factory}$
using \Cref{alg:swarm_composition_short}, we apply the adaptation in \Cref{def:machine-adaptation} to all machines in both $\S_\Gn{W}$ and $\S_\Gn{F}$, and we obtain a swarm $\S_{\Gn{W}\conc \Gn{F}}$ that realises the composed protocol $\Gn{Warehouse} \conc \Gn{Factory}$.
We illustrate this on the forklift machine $\afish[{\rolename{FL}}] = \Gn{Warehouse}\projshort{\subscription}{\rolename{FL}}$.
We first apply \Cref{alg:comp-swarm-subscriptions} (as in \Cref{ex:applying-sub-gen-alg}) to obtain a subscription $\subscription'$ such that $\Gn{Warehouse} \conc \Gn{Factory}$ is $\subscription'$-well-formed, and which contains the initial $\subscription$ and the interfacing event types needed for synchronising the machines that implement $\Gn{Warehouse}$ and $\Gn{Factory}$. By \Cref{def:machine-adaptation}, the machine adaptation of $\afish[{\rolename{FL}}]$ is:

    \ifbool{ecoop}{\smallskip}{\medskip}%
    \centerline{\(
        \mathcal{A}\!\left(\afish[{\rolename{FL}}], {\left(\Gn{Warehouse}, \Gn{Factory}\right), \rolename{FL}}, 1, \subscription'\right)%
        \;=\;%
             \left(\afish[{\rolename{FL}}] \concM{} \Gn{Warehouse} \projshort{\subscription'}{\rolename{FL}}\right) %
             \concM{} \Gn{Factory} \projshort{\subscription'}{\rolename{FL}}
    \)}
    \ifbool{ecoop}{\smallskip}{\medskip}%

\noindent%
where ``1'' is the index of $\Gn{Warehouse}$ in the tuple $(\Gn{Warehouse},\Gn{Factory})$.
The stages of the adaptation %
are depicted in \Cref{fig:ProjAdapt}, where the state numbers show the relationship between the original machine and its adaptations. The original $\afish[{\rolename{FL}}]$ is adapted with the projection of its original swarm protocol ($\Gn{Warehouse}$) using $\subscription'$ (computed with \Cref{alg:comp-swarm-subscriptions}) to obtain the machine $\afish[{\rolename{FL}}] \conc \Gn{Warehouse} \projshort{\subscription'}{\rolename{FL}}$: observe that this adapted machine has an additional state 3 which forces the machine to wait for an event of type $\lgtn{partOK}$ after $\lgtn{pos}$, between the original states 2 and 4. When $\afish[{\rolename{FL}}] \conc \Gn{Warehouse} \projshort{\subscription'}{\rolename{FL}}$ is further composed with $\Gn{Factory} \projshort{\subscription'}{\rolename{FL}}$, the adapted machine loses the possibility in state 4 of awaiting another event of type $\lgtn{partReq}$ and looping to state 2: %
hence, in state 4, the adapted machine can only await and process an event of type $\lgtn{closingTime}$.
Using \Cref{alg:swarm_composition_short}, this final machine is added to the adapted swarm $\S_{\Gn{W}\conc \Gn{F}}$.
Then, \Cref{alg:swarm_composition_short} repeats this adaptation for every machine in $\S_\Gn{W}$ and $\S_{\Gn{F}}$.
\end{example}}

  \begin{figure}[t]
  	\centering
  	\begin{tabular}{@{}c@{\qquad}c@{\qquad}c@{}}
  		\begin{minipage}{0.25\textwidth}
  			\centering
  			\begin{tikzpicture}[gt, scale = \ifbool{ecoop}{.55}{1}, node distance = 4mm and 15mm, font=\footnotesize]
  				\node (0) [plb]{1};
  				\node (1) [left =0.5cm of 0] {};
  				\node (2) [plb, above right = of 0] {2};
  				\node (3)  [plb, right = of 2] {4};
  				\node (4) [plb, below right = of 3] {5};
  				
  				\draw[->] (1) to node[sloped, anchor=center, above=1mm]{}  (0);
  				\path (0) edge[sloped, below] node{$ \lgtn{closingTime}?$}  (4);
  				\path (0) edge[sloped, above] node{$ \lgtn{partReq}?$}  (2);
  				\path (2) edge[loop above] node{$\lgtn{pos}!$} ();
  				\path (2) edge[sloped, below, bend right] node{$  \lgtn{pos}?$}  (3);
  				\path (3) edge[sloped, above, bend right] node{$  \lgtn{partReq?}$}  (2);
  				\path (3) edge[sloped, above] node{$\lgtn{closingTime}?$} (4);
  			\end{tikzpicture}
  			
  			{\small $\afish[{\rolename{FL}}]$}
  		\end{minipage}
  		&
  		\begin{minipage}{0.3\textwidth}
  			\centering
  			\begin{tikzpicture}[gt, scale = \ifbool{ecoop}{.55}{0.9}, node distance = 4mm and \ifbool{ecoop}{13}{15}mm, font=\footnotesize]
  				\node (0) [plb]{1};
  				\node (1) [left =0.5cm of 0] {};
  				\node (2) [plb, above right = of 0] {2};
  				\node (3)  [plb, right = of 2] {3};
  				\node (4)  [plb, right = of 3] {4};
  				\node (5) [plb, below right = of 4] {5};
  				
  				\draw[->] (1) to node[sloped, anchor=center, above=1mm]{}  (0);
  				\path (0) edge[sloped, below] node{$ \lgtn{closingTime}?$}  (5);
  				\path (0) edge[sloped, above] node{$ \lgtn{partReq}?$}  (2);
  				\path (2) edge[loop above] node{$\lgtn{pos}!$} ();
  				\path (2) edge[sloped, below] node{$  \lgtn{pos}?$}  (3);
  				\path (3) edge[sloped, below] node{$  \lgtn{partOK}?$}  (4);
  				\path (4) edge[sloped, above, bend right] node{$  \lgtn{partReq?}$}  (2);
  				\path (4) edge[sloped, above] node{$\lgtn{closingTime}?$} (5);
  			\end{tikzpicture}
  			
  			{\small $\afish[{\rolename{FL}}] \conc \Gn{Warehouse} \projshort{\subscription'}{\rolename{FL}}$}
  		\end{minipage}
  		&
  		\begin{minipage}{0.35\textwidth}
  			\centering
  			\begin{tikzpicture}[gt, scale = \ifbool{ecoop}{.55}{0.9}, node distance = 4mm and \ifbool{ecoop}{13}{15}mm, font=\footnotesize]
  				\node (0) [plb]{1};
  				\node (1) [left =0.5cm of 0] {};
  				\node (2) [plb, above right = of 0] {2};
  				\node (4) [plb, right = of 2] {3};
  				\node (5) [plb, right = of 4] {4};
  				\node (6) [plb, below right = of 5] {5};
  				
  				\draw[->] (1) to node[sloped, anchor=center, above=1mm]{}  (0);
  				\path (0) edge[sloped, below] node{$ \lgtn{closingTime}?$}  (6);        
  				\path (0) edge[sloped, above] node{$ \lgtn{partReq}?$}  (2);
  				\path (2) edge[loop above] node{$\lgtn{pos}!$} ();
  				\path (2) edge[sloped, below] node{$  \lgtn{pos}?$}  (4);
  				\path (4) edge[sloped, below] node{$  \lgtn{partOK}?$}  (5);
  				\path (5) edge[sloped, above] node{$\lgtn{closingTime}?$} (6);
  			\end{tikzpicture}
  			
  			{\small $(\afish[{\rolename{FL}}] \concM{} \Gn{Warehouse} \projshort{\subscription'}{\rolename{FL}}) %
  			\concM{} \Gn{Factory} \projshort{\subscription'}{\rolename{FL}}$}
  		\end{minipage}
  	\end{tabular}
  	
  	\caption{\ecoop{How \Cref{def:machine-adaptation} adapts the machine $\afish[{\rolename{FL}}]= \Gn{Warehouse}\projshort{\subscription}{\rolename{FL}}$ to $\Gn{Warehouse} \conc \Gn{Factory}$.}}
  	\label{fig:ProjAdapt}
  \end{figure}

We can now state the correctness of our swarm composition procedure. \emph{(Proof in \Cref{sec:app:proof_of_projection_composition_correct}.)}

\begin{theorem}[\ecoop{Correctness of swarm adaptation and composition}]
\label{lem:projection_composition_correct}%
Let $\S$ be a swarm obtained by applying \Cref{alg:swarm_composition_short} %
to \ecoop{swarms $\S_1,\ldots,\S_n$, protocols $\G_1, \ldots, \G_n$} and
\ecoop{subscriptions} $\subscription_1, \ldots, \subscription_n$. %
Then, $\S$ is eventually faithful to $\G_1 \conc \ldots \conc \G_n$.
\end{theorem}

The key insight in the proof of \Cref{lem:projection_composition_correct} is that the subscription $\subscription$  computed by \Cref{alg:comp-swarm-subscriptions} includes all the
events needed for the synchronization of the composed swarm. This ensures that each
adapted machine is a correct projection of the composed swarm protocol. Therefore, a machine adapted from the input
swarm $\S_i$ will execute without desynchronising and diverging from other machines adapted from another input swarm $\S_j$ (for $i, j \in \{1 \dots n\}$ such that $i \neq j$).

\section{Implementation and Experiments}
\label{sec:implementation}
We now present our implementation of the theory in \Cref{sec:swarm-protocols-composition,sec:projecting-machines-comp-swarm-proto,sec:swarm-composition}.
In \Cref{sec:implementation:branch-verif}
we discuss how we extend the open source Actyx toolkit~\cite{actyx} to support branch-tracking machine semantics
and machine adaptation (which are key elements of our compositional approach), 
and to statically verify our new, compositional well-formedness of swarm protocols.
Then, in \Cref{sec:implementation:experiments} we benchmark different strategies for
computing the subscriptions of composed swarms.
Our implementation is availble in the companion artifact of this paper \cite{furbach_2026_18459720}.

\subsection{Branch Tracking, Adaptations, and New Verification Facilities}
\label{sec:implementation:branch-verif}

\subparagraph{Branch Tracking and
  Adaptation.}%
The \texttt{machine-runner} library~\cite{machineRunner} of the Actyx
toolkit features a TypeScript API to program \emph{machine
  implementations}.
For instance, \Cref{fig:code_machine_door} yields the machine implementation
of a machine $\M$ (defined according to
\Cref{{def:formalisms:machine}}) for the role $\rolename D$ of our
$\Gn{Warehouse}$ use case.

The \texttt{machine-runner} library also helps instantiating and running machine implementations in a
swarm using the Actyx middleware~\cite{actyxMiddleware} which
implements the log propagation mechanisms described
in \Cref{sec:formalisms:swarms}.
Let us discuss our extension of \texttt{machine-runner}.

In order to support our new branch-tracking semantics
(\Cref{def:branch-tracking-transition-function-machines}), \emph{(1)}
we added a pointer $\ev.\lbjePointer$ (\quo{previous updating event})
to each event $\ev$, and \emph{(2)} modified \texttt{machine-runner}
to produce and access $\ev.\lbjePointer$.
We also added new methods that \longer{allow programmers to
}automatically adapt machine implementations according to
\Cref{def:machine-adaptation}.
This allows developers to re-use and deploy existing machine implementations in composed
swarms.
For instance, to adapt the machine implementation
\lstset{backgroundcolor={gray}}\lstinline[style=ES6]{door} of role
$\rolename{D}$ in \Cref{fig:code_machine_door} to the composed swarm
$\Gn{Warehouse} \conc \Gn{Factory}$, a programmer can simply invoke:%

\smallskip%
\centerline{%
\lstset{backgroundcolor={gray}}\lstinline[style=ES6]{adaptedDoor  =  adaptMachine('D', listOfProtocols, 0, subscriptions, [door, s0])}
}%
\smallskip%

\noindent%
The parameters to \lstset{backgroundcolor={gray}}\lstinline[style=ES6]{adaptMachine} are the machine role $\rolename{D}$, a list of the swarm protocols in the composition, the index of $\Gn{Warehouse}$ in that list, a well-formed (and automatically generated) subscription for the composition, and the original \lstset{backgroundcolor={gray}}\lstinline[style=ES6]{door} machine with its initial state.
\longer{
	Consider a machine implementation of the role $\rolename{D}$ from the protocol $\Gn{Warehouse}$ ($\G \projshort{\subscription}{\rolename{D}}$ in \Cref{fig:ProjWH}).
	To automatically adapt this implementation into a correct implementation of $\rolename{D}$ in a swarm for $\Gn{Warehouse} \conc \Gn{Factory}$,
	we invoke:

	\ifbool{ecoop}{\smallskip}{\medskip}%
	\lstset{backgroundcolor={gray}}\lstinline[style=ES6,float=t]{adaptMachine('mD', projectionOfCompositionD, eventTypesOfComposition, originalDMachine)}
	\vspace{-3mm}%
	
	\noindent%
	where the first argument \lstset{backgroundcolor={}}\lstinline[style=ES6,float=t]{'mD'} is an internal name of the machine,
	the second argument is the projection $(\Gn{Warehouse} \conc \Gn{Factory}) \projshort{\subscription}{\rolename{D}}$ (which is shown in \Cref{fig:ProjWHF} and can be obtained using the features discussed in \Cref{sec:implementation:new-static} below),
	and the last two arguments are the event types of the composition and the original machine implementing role $\rolename{D}$ in $\Gn{Warehouse}$, respectively. }

\subparagraph{New Static Verification Facilities.}
\label{sec:implementation:new-static}
The Actyx toolkit can statically verify several properties through the
\texttt{machine-check}~\cite{machineCheck} library.
More precisely, \texttt{machine-check} can verify ($i$) the
well-formedness (as defined in~\cite{DBLP:conf/ecoop/KuhnMT23}) of a
protocol $\G$ under a given subscription $\subscription$ and ($ii$)
whether the machine implementation of a role $\rolename R$ of $\G$
emits and accepts events as specified by
the projection of $\G \projshort \subscription {\rolename R}$.

We extended \texttt{machine-check} to verify our new well-formedness
of protocol compositions (\Cref{def:swarm-proto-composition,def:wf}),
and generate subscriptions for which a composition is well-formed
(\Cref{alg:comp-swarm-subscriptions}).
We also modified the \texttt{machine-check} projection code (which was based on \cite{DBLP:conf/ecoop/KuhnMT23}) 
to conform to our notion of projections in \Cref{def:swarm-proto-projection}.
This enables (1)  the static verification of machine implementations in composed swarm protocols, 
and (2) the automatic adaptation of previously-implemented machines to work correctly in a composed swarm.

\subsection{Experiments on Compositional Subscription Computation}
\label{sec:implementation:experiments}
A key element for the correct composition of protocols and swarms is
the availability of a subscription $\subscription$ suitable for composition
(required by \Cref{def:eventual-fidelity} and
\Cref{thm:eventual_fidelity}).
We analyse the performance and trade-offs of
\Cref{alg:comp-swarm-subscriptions}, which is instrumental for
computing subscriptions suitable for swarm composition
(cf. \Cref{sec:swarmcomposition}).

\Cref{alg:comp-swarm-subscriptions} may compute a subscription larger than necessary, 
which may lead to non-optimal performance in the deployed swarm.
Therefore, we compare \Cref{alg:comp-swarm-subscriptions} with a naive
algorithm (available in {\Cref{sec:app:inference_rules}}) that computes the \quo{exact} subscription by directly
applying \Cref{def:wf}.
More precisely, this algorithm accepts a set $\{\G_1 \mydots \G_n\}$  
of composable protocols and a set of subscriptions $\{\subscription_1\mydots \subscription_n\}$, 
expands the composed protocol according to
\Cref{def:swarm-proto-composition}
and applies \Cref{def:wf} to produce the smallest $\subscription \supseteq \subscription_1 \mydots \subscription_n$
for which $\G = \G_1 \conc \ldots \conc \G_n$ is $\subscription$-well-formed.
As mentioned in \Cref{sec:swarm-protocols-composition-computing-sub},
the expansion of the composed protocol makes the time complexity of
this algorithm exponential in the number of input protocols.

\subparagraph{Benchmark Selection and Experimental Setup.}%
We created a benchmark consisting of 454 sets of randomly-generated swarm protocols.
Each set consists of $n$ swarm protocols
$\G_1, \ldots, \G_n$ (with $n \in \{1, \ldots, 10\}$) with up to 9
roles each; each role emits up to 9 event types.
In order to cover a broad range of possible practical cases, the
random generation produces variety of branching and looping patterns.
Moreover, protocols $\G_i$ and $\G_{i+1}$ share an
interfacing role, for all $i \in \{ 1, \ldots, n-1 \}$.
All randomly-generated protocols follow a pattern: %
two interfacing event types $\evt_1$ and $\evt_2$ may be separated by
a random number of choices with non-interfacing event types, but
$\evt_1$ and $\evt_2$ are always enabled in the same order in all
protocols they occur in.

The experiments were conducted by invoking both \Cref{alg:comp-swarm-subscriptions} and our naive ``exact'' algorithm on empty
input subscriptions: this is to compute the smallest possible output subscription %
for each protocol composition $\G_1 \conc \ldots \conc \G_n$.
The execution times were measured using Criterion \cite{criteriondoc}; %
for each randomly-generated set of swarm protocols, %
Criterion invoked \Cref{alg:comp-swarm-subscriptions} and the ``exact'' algorithm at least 50 times
after 3 seconds of warm-up, reporting averages, medians, and standard deviations (including outliers).

\subparagraph{Results.}%
\Cref{fig:runtimes_refinement_2_edges} plots the execution time measured in our
experiments against the number of transitions in the composition $\G_1 \conc \ldots \conc \G_n$;
note that the number of transitions grows exponentially in $n$.
As expected, \Cref{alg:comp-swarm-subscriptions} is faster and more
scalable than the ``exact'' algorithm, which suffers from the exponential
blow-up of the composition.
For comparison, the generation of a subscription for a composition of
protocols with ${\sim}10^5$ transitions takes ${\sim}0.01$ seconds
using \Cref{alg:comp-swarm-subscriptions,} and ${\sim}10$ seconds 
using the ``exact'' algorithm.
Notably, the trend shown in \Cref{fig:runtimes_refinement_2_edges} persists even if we discount the time
needed to expand the composition of the input protocols before invoking the ``exact'' algorithm:
this is because the ``exact'' algorithm spends more time applying \Cref{def:wf} to generate 
a well-formed subscription for the expanded composition.

\begin{figure}[t]
\begin{minipage}[t]{0.55\linewidth}
    \pgfkeys{/pgf/number format/.cd,1000 sep={\,}}
\centering
\vspace{-6mm}%
\resizebox{1\linewidth}{!}{\begin{tikzpicture}[spy using outlines = {circle, size=2cm, magnification=5, connect spies}]
\begin{axis}[
    width=1.5\linewidth,
    height=5cm,
    ylabel = {\small Time ($\mu s$, logarithmic scale)},
    xlabel = {\small Number of transitions in the composition (logarithmic scale)},
    enlarge x limits = false,
    legend style = {at={(0.5,-0.26)}, anchor=north},
    legend image post style={mark size=2pt},
    every axis plot/.append style={
        fill opacity=0,
        thin},
    ymode=log,                %
    log basis y=10,           %
    xmode=log,                %
    log basis x=10,           %
    grid=both,                %
    minor tick num=9          %
    ]
\addplot[only marks, mark=+,mark size=0.9pt,draw=dtured,fill=dtured]
table[col sep=comma, x=number_of_edges, y=mean_point_estimate] {figures_experiments/exact_microseconds_execution_time.csv};
\addplot[only marks, mark=o,mark size=0.9pt,draw=lightblue2,fill=lightblue2]
table[col sep=comma, x=number_of_edges, y=mean_point_estimate] {figures_experiments/alg1_microseconds_execution_time.csv};
\legend{Exact,  \hyperref[alg:comp-swarm-subscriptions]{Alg. 1}} %
\end{axis}
\end{tikzpicture}}
\vspace{-3mm}%
\captionof{figure}{%
    Execution times for generating well-formed subscriptions, as a function of
    the number of transitions in the composition. %
    {\scriptsize(AlmaLinux 9.5 on Intel Xeon Gold 6226R, 32GB of RAM; averages produced using Criterion 0.7.0 \cite{criteriondoc} with sample size 50.)}%
    \label{fig:runtimes_refinement_2_edges}%
}
\end{minipage}
\hfill%
\begin{minipage}[t]{0.2\linewidth}
    \centering
\vspace{-8mm}%
\resizebox{1\linewidth}{!}{\begin{tikzpicture}
\begin{axis}[
    width=1.2\linewidth,
    scale only axis,
    height=5.0cm,
    boxplot/draw direction=y,
    xtick={1,2},
    xticklabels={\small Exact, \small \hyperref[alg:comp-swarm-subscriptions]{Alg. 1}},%
    x axis line style={opacity=0},
    enlarge y limits,
    ymajorgrids,
    cycle list name=DTU,
    every axis plot/.append style={semithick},
    ymax=1,
    ymin=0,
    ylabel={\small $E_\textit{frac}$},
]
\addplot+[draw=black,fill=dtured,
boxplot prepared={
	lower whisker=0.0930232558139534, lower quartile=0.16529362293403538,
	median=0.21272586602775279,
	average=0.22418626973013725,
	upper quartile=0.26489441930618396, upper whisker=0.4117647058823529,
},
] coordinates {(0,0.417391304347826) (0,0.4222222222222222) (0,0.4259259259259259) (0,0.4375) (0,0.4490740740740741) (0,0.4615384615384615) (0,0.4835164835164835) (0,0.5625) (0,0.6) (0,0.6071428571428571) (0,0.8333333333333334) (0,0.875)};

\addplot+[draw=black,fill=lightblue2,
boxplot prepared={
	lower whisker=0.1329937747594793, lower quartile=0.2381868131868131,
	median=0.28320105820105823,
	average=0.2994252299534819,
	upper quartile=0.3451011273209549, upper whisker=0.5,
},
] coordinates {(0,0.5195312499999999) (0,0.5220588235294118) (0,0.5234375) (0,0.5392156862745098) (0,0.5454545454545454) (0,0.5833333333333334) (0,0.59375) (0,0.6) (0,0.6428571428571429) (0,0.6483516483516484) (0,0.8333333333333334) (0,0.875)};

\end{axis}
\end{tikzpicture}}
\captionof{figure}{%
    Boxplot showing the average fraction of events in the computed subscriptions. (Lower is better.)%
    \label{fig:boxplotgeneral}%
}
\end{minipage}
\hfill
\begin{minipage}[t]{0.2\linewidth}
	\centering
\vspace{-8mm}%
\resizebox{1\linewidth}{!}{\begin{tikzpicture}
\begin{axis}[
    width=1.2\linewidth,
    scale only axis,
    height=5.0cm,
    boxplot/draw direction=y,
    xtick={1,2,3},
    xticklabels={\small Exact, \small \hyperref[alg:comp-swarm-subscriptions]{Alg. 1}, \small \cite{DBLP:conf/ecoop/KuhnMT23}},
    x axis line style={opacity=0},
    enlarge y limits,
    ymajorgrids,
    cycle list name=DTU,
    every axis plot/.append style={semithick},
    ymax=1,
    ymin=0,
    ylabel={\small $E_\textit{frac}$},
]

\addplot+[draw=black,fill=dtured,
boxplot prepared={
	lower whisker=0.2037037037037037, lower quartile=0.3875,
	median=0.4642857142857143,
	average=0.4761522661516901,
	upper quartile=0.55, upper whisker=0.7916666666666666,
},
] coordinates {(0,0.8) (0,0.8) (0,0.8) (0,0.8) (0,0.8) (0,0.8) (0,0.8125) (0,0.8125) (0,0.8148148148148148) (0,0.8333333333333334) (0,0.8333333333333334) (0,0.8333333333333334) (0,0.8333333333333334) (0,0.8333333333333334) (0,0.8333333333333334) (0,0.8333333333333334) (0,0.8333333333333334) (0,0.8333333333333334) (0,0.8333333333333334) (0,0.8571428571428571) (0,0.8666666666666666) (0,0.875) (0,0.875) (0,0.875) (0,0.875) (0,0.875) (0,0.875) (0,0.875) (0,0.9) (0,0.9166666666666666) (0,0.9285714285714286) (0,1.0) (0,1.0)};

\addplot+[draw=black,fill=lightblue2,
boxplot prepared={
	lower whisker=0.2037037037037037, lower quartile=0.3875,
	median=0.4642857142857143,
	average=0.4761522661516901,
	upper quartile=0.55, upper whisker=0.7916666666666666,
},
] coordinates {(0,0.8) (0,0.8) (0,0.8) (0,0.8) (0,0.8) (0,0.8) (0,0.8125) (0,0.8125) (0,0.8148148148148148) (0,0.8333333333333334) (0,0.8333333333333334) (0,0.8333333333333334) (0,0.8333333333333334) (0,0.8333333333333334) (0,0.8333333333333334) (0,0.8333333333333334) (0,0.8333333333333334) (0,0.8333333333333334) (0,0.8333333333333334) (0,0.8571428571428571) (0,0.8666666666666666) (0,0.875) (0,0.875) (0,0.875) (0,0.875) (0,0.875) (0,0.875) (0,0.875) (0,0.9) (0,0.9166666666666666) (0,0.9285714285714286) (0,1.0) (0,1.0)};

\addplot+[draw=black,fill=green,
boxplot prepared={
	lower whisker=0.4333333333333333, lower quartile=0.6333333333333333,
	median=0.7058823529411765,
	average=0.7081477348769393,
	upper quartile=0.7777777777777778, upper whisker=0.9835390946502058,
},
] coordinates {(0,0.375) (0,0.3777777777777777) (0,0.3833333333333333) (0,0.3928571428571428) (0,0.4) (0,0.4074074074074074)};

\end{axis}
\end{tikzpicture}}
\captionof{figure}{%
    Comparison of subscription sizes obtained using our definition of well-formedness and that of \cite{DBLP:conf/ecoop/KuhnMT23}. %
    \label{fig:boxplot-kmt-vs-compositional}%
}
\end{minipage}
\end{figure}

\Cref{fig:boxplotgeneral} compares the accuracy of the two algorithms,
measured as the fraction of event types involved in the generated
subscription $\subscription$.
Specifically, $E_\textit{frac}$ is the fraction of all event types
subscribed to by a role in $\subscription$ on average across all
roles.
For instance, $E_\textit{frac}=0.2$ means that, on average, a role
subscribes to one out of five event types, while
  $E_\textit{frac}=1.0$ means that all roles subscribe to every event
  type.
In general, a lower $E_\textit{frac}$ means less subscriptions and a
higher degree of concurrency, leading to a swarm with smaller and more
efficient machines.
The plot shows that
the \quo{exact} algorithm requires roles to subscribe to ${\sim}22.4\%$
of all event types in the input protocols.
\Cref{alg:comp-swarm-subscriptions} requires subscribing to
${\sim}29.9\%$ of the event types, which is quite close to the
\quo{exact} algorithm.
The difference occurs because \Cref{alg:comp-swarm-subscriptions} may not see that parts of the input protocols
become unreachable in the composition, and thus, may produce redundant subscriptions to some updating event types; also,
\Cref{alg:comp-swarm-subscriptions} adds interfacing event types to subscriptions (even when not required by well-formedness). This is necessary to ensure correct machine adaptation without expanding the composed protocol.

\subparagraph{Comparing Subscription Sizes with Those of \cite{DBLP:conf/ecoop/KuhnMT23}.}%
In \Cref{sec:overview:approach} we mentioned that the well-formedness definition of \cite{DBLP:conf/ecoop/KuhnMT23} may lead to larger subscriptions than necessary, especially compared to our \Cref{def:wf} and \Cref{alg:comp-swarm-subscriptions}.
We now quantify this claim experimentally. %
\Cref{fig:boxplot-kmt-vs-compositional} compares the sizes of minimal well-formed subscriptions obtained using our definition of well-formedness (\Cref{def:wf}) and that of \cite{DBLP:conf/ecoop/KuhnMT23} 
for 2093 randomly generated swarm protocols. %
Since \cite{DBLP:conf/ecoop/KuhnMT23} does not support concurrent events, these 2093 protocols contain no compositions and no concurrency.
As a consequence, the minimal subscriptions that ensure well-formedness (for both \cite{DBLP:conf/ecoop/KuhnMT23} and our \Cref{def:wf}) are expected to be larger than those in \Cref{fig:boxplotgeneral}.
The plot in \Cref{fig:boxplot-kmt-vs-compositional} shows that our \Cref{def:wf} and \Cref{alg:comp-swarm-subscriptions} require roles to subscribe to ${\sim}47.6\%$ of all event types in the input protocols on average,
whereas \cite{DBLP:conf/ecoop/KuhnMT23} requires ${\sim}70.8\%$.\footnote{%
	\label{footnote:outliers}%
	In \Cref{fig:boxplot-kmt-vs-compositional}, the ``exact'' and \hyperref[alg:comp-swarm-subscriptions]{Alg. 1} boxplot both contain two outliers with $E_\textit{frac} = 1$, greater than all values in the boxplot for \cite{DBLP:conf/ecoop/KuhnMT23}. %
	These are two very small protocols with a choice where a role $\role$ only appears in one of the branches.
	In this case, our \Cref{def:wf} requires $\role$ to subscribe to event types in all branches of the choice, whereas \cite{DBLP:conf/ecoop/KuhnMT23} only requires $\role$ to subscribe to the branch where $\role$ participates. As a result, our \Cref{def:wf} requires additional subscriptions for such cases -- but it also provides a stronger guarantee w.r.t.~\cite{DBLP:conf/ecoop/KuhnMT23}: a machine playing role $\role$ can determine if the swarm is following a branch of the swarm protocol where $\role$ is not involved.%
} %
\Cref{fig:boxplot-kmt-vs-compositional} also shows that, for non-concurrent protocols, \Cref{alg:comp-swarm-subscriptions} yields the same output subscriptions as the ``exact'' algorithm: this is because the protocols contain no interfacing event types, nor do they exhibit behaviour restrictions
like those in \Cref{ex:behaviour_restriction}.

\section{Related Work}
\label{sec:related work}
This work embraces the research path taken in~\cite{DBLP:conf/ecoop/KuhnMT23},
advocating
\emph{local-first cooperation}~\cite{lfc,localfirst} in the design,
analysis, and development of distributed applications modelled as swarms of machines,
using swarm protocols as a specification language.
We tackle the problem of compositional design and verification of swarms --
which was left open in~\cite{DBLP:conf/ecoop/KuhnMT23} and, to the best of our knowledge, is not
addressed elsewhere.
The importance of modularisation for managing software complexity has been well
established since early work~\cite{parnas1972criteria}:
separate parts of a system should be composed via well-defined
\emph{interfaces} that hide implementation details between components.
We accomplish this using \emph{interfacing roles} (\Cref{def:swarm-proto-interface})
whose emitted events act as the glue that binds different swarm
protocols and swarms into larger ones.%

\subparagraph{Technical Differences w.r.t.~\cite{DBLP:conf/ecoop/KuhnMT23}.}
Besides introducing our new compositionality results, in this work we simplify the theory of~\cite{DBLP:conf/ecoop/KuhnMT23}
aligning it to the common usages of the open source Actyx toolkit~\cite{actyx},
without reducing the expressiveness of the model:
\begin{itemize}
\item The grammar in \Cref{def:formalisms:machine} is as
	 in~\cite{DBLP:conf/ecoop/KuhnMT23}, except that we do not explicitly model the
	 ``commands'' a machine performs when emitting events:

	 we assume a 1:1 correspondence between events and commands, leaving the
	 latter implicit. Moreover, our transitions range over single event types
	 instead of finite non-empty sequences of event types.

  \item In \Cref{def:formalisms:swarm-semantics}, our rule
	 \textsc{Local} simplifies the corresponding rule
	 in~\cite{DBLP:conf/ecoop/KuhnMT23}: since we have a 1:1
	 correspondence between events and commands, the shuffling operator
	 of~\cite{DBLP:conf/ecoop/KuhnMT23} is unnecessary.
     Also, our rule \textsc{Global} simplifies the corresponding rule
	 in~\cite{DBLP:conf/ecoop/KuhnMT23}.\footnote{%
	 	More specifically, the closure property imposed on the sub-log $\alog'$
	 	in~\cite{DBLP:conf/ecoop/KuhnMT23} is not necessary in this work. This
	 	is because we do not need to ensure that for each event %
		$\aevent$ in the global log $\alog$ emitted by a machine of a swarm
	 	$\asysrt$, say $\asysrt(j)$, $\alog$ includes each event $\aevent'$ that
	 	precedes $\aevent$ in the local log of $\asysrt(j)$.}%
\end{itemize}

A key challenge of our compositional approach, compared to the
``monolithic'' approach of~\cite{DBLP:conf/ecoop/KuhnMT23}, is that %
in the monolithic setting, the eventual fidelity of a machine $\M$ is only ensured if $\M$ knows and accepts the \emph{guard event} for every command that precedes $\M$'s event emissions -- %
and this, in our one-event-per-command setting, would mean that $\M$ must know and accept all events that precede $\M$'s emissions. %
However, in our composed swarms, a machine $\M$ projected from one swarm protocol may encounter events produced by other machines and protocols that $\M$ did not originally anticipate.
To address this, our revised machine semantics (\Cref{sec:branch-tracking-machines})
tracks event causality more precisely, requiring a machine $\M$ that plays role $\role$ to accept only those events that can affect $\role$'s behaviour.
This is achieved via \emph{updating event types} and our  improved swarm semantics (\Cref{def:branch-tracking-swarm-semantics}): each event carries minimal causal information, i.e., a pointer to the last relevant updating event observed by the emitting machine.
This allows a machine to decide whether to accept or ignore an event.
The main result of \cite{DBLP:conf/ecoop/KuhnMT23} is that any swarm
$\S$ that realises a well-formed $\G$ is eventually faithful to $\G$; %
we generalise this result to composed swarms and protocols in
\Cref{sec:eventual-fidelity,sec:machine-adaptation}.

\subparagraph{Global Types and Behavioural Types.}%
Swarm protocols are inspired by \emph{global types}, a.k.a.~\emph{choreographies}~\cite{honda2008multiparty,honda2016multiparty}; they
describe multiparty interaction pattern from a global point of view,
and they can be projected onto roles to obtain local
specifications, against which one can check the correctness of
implementations.
{%
	The literature on global types includes compositional approaches: e.g., 
	partial choreographies that can be composed while preserving deadlock-freedom \cite{DBLP:conf/concur/MontesiY13},
	and choreography refinement to replace part of a choreography with a more complex one \cite{DELIGUORO2022100776}. %
}%
However, the highly dynamic and non-deterministic semantics of swarms is
significantly different from the semantics in global types literature. %
When compared to the literature on behavioural
types~\cite{gay2017behavioural,hlvccdmprt16,ancona2016behavioral}, our approach
and~\cite{DBLP:conf/ecoop/KuhnMT23} depart from the usual assumptions in several
key ways.
Standard behavioural typing approaches %
ensure all interacting components are aware of the current state of the distributed computation.
In contrast, we allow inconsistencies between machines to emerge, as long as consistency is \emph{eventually} achieved~\cite{burckhardt2014principles}:
e.g., machines can make transient discordant choices that are reconciled as the logs propagate in the swarm.
This is akin to data replication methods, where there is a trade-off between availability, consistency, and partition tolerance~\cite{CAP}. Several approaches have emerged to balance this trade-off, including conflict-free replicated data types~\cite{DBLP:conf/sss/ShapiroPBZ11}, cloud types~\cite{burckhardt2012cloud}, consistency contracts~\cite{sivaramakrishnan2015declarative}, invariants~\cite{gotsman2016cause, kaki2018safe, DBLP:journals/pvldb/BalegasDFRP18}, linearizability~\cite{wang2019replication}, and operational models for applications such as GSP~\cite{DBLP:conf/ecoop/BurckhardtLPF15, gotsman2017consistency};
a key difference in our approach lies in how we achieve eventual consistency, and how we introduce compositional mechanisms.
Moreover, in standard behavioural type sytems, the variability in the number of participants has been addressed via \emph{parametricity}, e.g.~in multiparty session types~\cite{ydbh10,dy11,mybh12,chsny19,jy20}:
this requires explicit handling of
the parameters of the protocol. 
Instead, swarm protocols do not restrict the
number of instances playing a given role.

In behavioural types literature, most compositional approaches require
component protocols to be aware that they are part of a composition: e.g.~hybrid types~\cite{gheriyoshida23} distinguish
inter-component interactions from intra-component interactions; other
approaches allow composition via open endpoints~\cite{DBLP:conf/concur/MontesiY13,stolze2023composable}.
An approach closer to ours is \emph{Participants-as-Interfaces}
(PaI)~\cite{DBLP:journals/jlap/BarbaneradH19,bdlt21,blt22a,barbanera2023multicompatibility,blt23}, where interfaces are defined by
roles that perform complementary behaviours. Like us, PaI allows viewing
a closed system as an open one -- but PaI uses complementary roles as \emph{forwarders} to connect protocols, whereas our interfacing roles simply play their roles and do not act as forwarders.

\subparagraph*{Choreographic Programming Languages \cite{M23}}\hspace{-0.5em}%
differ from swarm protocols in both goals and approach. Choreographic programs provide a global \emph{implementation} of communicating roles via point-to-point messages, with notions of compositionality inspired by programming languages theory (e.g., higher-order choreographic programs~\cite{ShenICFP23,CruzFilipeECOOP23}). Swarm protocols and projected machines, by contrast, serve as specifications only: machines are implemented and checked separately, e.g.,~using the Actyx toolkit. Moreover, machines operate under asynchronous, subscription-based event propagation, with possible event reordering and machine replication. These goals and semantics are significantly different from choreographic programming languages, and thus, they lead to a different notion of composition.

\section{   Conclusion and Future Work}
\label{sec:conclusion}
\label{sec:future-work}

We have presented novel results on the compositional verification and implementation of swarms and swarm protocols. Our results enable the compositional verification of swarm systems, as well as the reuse of the implementations of swarm machines in larger composed systems.
This enables a compositional swarm design process, as well as the seamless extension of existing swarm systems with new functionality.
We believe our novel compositionality results can be applied to other local-first distributed systems.
We have implemented our results in a custom version of the open source Actyx toolkit for swarm application development \cite{actyx}, thus enabling the compositional design and implementation of swarms in practice.
We have also streamlined the swarm protocol theory presented in~\cite{DBLP:conf/ecoop/KuhnMT23}, while aligning it with the observed use of Actyx tools in the real world.

\subparagraph{Future Work.}%
{%
	Swarms and machines have some similarity with \emph{actor systems} \cite{HewittActors1973,AghaActors1986} (asynchronous communication, replication)
	but differ significantly in their semantics (out-of-order event propagation). We are not aware of specific work on the composition of actor systems -- although the communicating finite-state machines used in \cite{bdlt21,blt22a,blt23} have analogies with actors (message buffers can be seen as a form of mailbox), hence their approach to compositionality may be applicable to actors. %
	This is a valuable area of study, and we will investigate whether our compositionality results can be applied to actor systems.%
}%

{%
	Currently, our method of composing protocols and swarms (\Cref{alg:comp-swarm-subscriptions,alg:swarm_composition_short})
	and our correctness result (\Cref{lem:projection_composition_correct}) are limited to compositions of \emph{sequential} protocols without internal concurrency: %
	i.e., we support compositions of the form $\G_1 \conc \ldots \conc \G_n$ where each $\G_i$ is a sequential protocol (by \Cref{def:composable}) -- and this allows us to
	reuse and compose all swarm protocols that are well-formed according to \cite{DBLP:conf/ecoop/KuhnMT23}, and all machines projected from such protocols. %
	However, we do not currently support arbitrary compositions of concurrent protocols, like e.g.~$(\G_1 \conc \G_2) \conc (\G_3 \conc \G_4)$. %
}%
To address this limitation, we plan to build upon our results
and generalise them to design a full-fledged
\emph{module system for swarm development}. 
We will also study better compositional methods for the computation of subscriptions and the minimisation of swarm machines,
in order to improve the precision of our \Cref{alg:comp-swarm-subscriptions} (measured in \Cref{sec:implementation:experiments})
without losing scalability for larger swarms.

\bibliography{main}
\newpage
\part*{Appendices}
\appendix

\section{Proof of \autoref{lem:equalprojections}}\label{sec:app:proof_equalprojections}

We introduce the concept of causal dependency between event types which is implicitly used by the function $\newroles$.
\begin{definition}[{Definition: causal dependency}]
An event type \( \evt' \) is causally dependent on \( \evt \) in \( \G \) if there exists a sequence of transitions
$$\G \xrightarrow{\evt_0}  \xrightarrow{\lgtn{l}_1 \cdot \evt_1} \ldots  \xrightarrow{\lgtn{l}_n\cdot \evt_n}$$ where 
$\evt_0=\evt$, $\evt_n=\evt'$, and 
$\evt_i, \evt_{i+1}$ are not concurrent for all $i<n$.
\end{definition}

A role \( \role \) is in the set \( \newroles \) if and only if there exists an event \( \evt' \) in \( \G \) that is causally dependent on \( \evt \) and that \( \role \) subscribes to.

\subsection*{Proof of \Cref{lem:equalprojections} by Contradiction}
Given a swarm protocol $\G$, a subscription $\subscription$ such that $\G$ is well-formed w.r.t. $\subscription$ and a role $\role$, we will prove that all projections of $\G$ onto $\role$ for $\subscription$ behave identically.

Assume, for contradiction, that there exist two distinct correct projections \( \M_1 \) and \( \M_2 \) such that:
\begin{itemize}
	\item[A:] \( \M_1 \) and \( \M_2 \) receive different log types, or
	\item[B:] They emit different event types after receiving the same log type.
\end{itemize}
We analyze these two cases separately.

\subsection*{Case A: Different Log Types.}  
There is a shortest sequence of event types $\evt_1, \dots, \evt_n, \evt$ such that:
\begin{eqnarray*}
	\M_1=\G \proterm \receiving{\cmdnrindx{1}} \G_1  \proterm \ldots \receiving{\cmdnrindx{n}} \G_n \proterm &\receiving{\cmdnr} &\G_{n+1}\proterm\label{eq:eq-proj-1a}\\
	\M_2=\G \proterm \receiving{\cmdnrindx{1}} \G'_1 \proterm  \ldots \receiving{\cmdnrindx{n}} \G'_n \proterm &\not\receiving \evt&\label{eq:eq-proj-1b}
\end{eqnarray*}
We show that this is not the case.
\paragraph*{Induction Hypothesis:}
It holds that for any pairs of sequences 
\begin{eqnarray}
	\G \proterm \receiving{\cmdnrindx{1}} \G_1  \proterm \ldots \receiving{\cmdnrindx{n}} \G_n \proterm &\receiving{\cmdnr} \G_{n+1}\proterm\label{eq:appendable1}\\
	\G \proterm \receiving{\cmdnrindx{1}} \G'_1 \proterm  \ldots \receiving{\cmdnrindx{n}} \G'_n \proterm&\label{eq:appendable}
\end{eqnarray}
there is a $\G'_{n+1}$ such that we can append a transition $\G'_n \proterm \receiving{\cmdnr} \G'_{n+1}\proterm$ to Sequence (\ref{eq:appendable}).
We apply an induction over the length of the sequences $n$:
\paragraph*{Induction Basis $n=0$:}
Consider:
\begin{eqnarray*}
	\G \proterm &\receiving{\cmdnr} \G_{2}\proterm\\
	\G \proterm &
\end{eqnarray*}
It follows immediately that we can append a transition $\G \proterm \receiving{\cmdnr} \G_{2}\proterm$, since it is the same projection.

\paragraph*{Induction Step $(n-1)\rightarrow n$:}
It follows from the definition of projections and from the Sequences \ref{eq:appendable1} and \ref{eq:appendable}, that $\evt_1, \dots, \evt_n, \evt$ is the shortest series of event types such that:
\begin{eqnarray}
	\G  \xrightarrow{\cmdnrindx{a,1}} \G_{a,1}   \ldots \   \xrightarrow{\cmdnrindx{a,k}} \G_{a,k} \xrightarrow{\cmdnrindx{1}} \G_1   \ldots \xrightarrow{\evt_n}\G_n &\xrightarrow{\cmdnrindx{c,1}  \ldots \cmdnrindx{c,l}} \G_{c,l}  \xrightarrow{\cmdnr} \G_{n+1}\label{eq:project:ib1}\\
	\G  \xrightarrow{\cmdnrindx{b,1}} \G_{b,1}   \ldots \xrightarrow{\cmdnrindx{b,m}} \G_{b,m} \xrightarrow{\cmdnrindx{1}} \G'_1  \dots \xrightarrow{\evt_n} \G'_n&\label{eq:project:ib2}
\end{eqnarray}
and the only event types in Sequence \ref{eq:project:ib1} and \ref{eq:project:ib2} that occur in $\subscription(\role)$ are $\evt_1\mydots \evt_n,\evt$.

We assume WLOG that Sequence \ref{eq:project:ib1} is minimal, meaning the maximal subsequence of Sequence \ref{eq:project:ib1} that can be removed while still maintaining a valid path in the swarm protocol with $\evt_1$ and $\evt$ was removed.
For readability, we write Sequence \ref{eq:project:ib1} as:
$$\G  \xrightarrow{\cmdnrindx{a,1}}    \ldots \xrightarrow{\cmdnrindx{a,k}} \xrightarrow{\cmdnrindx{1}}  \ldots \xrightarrow{\evt_n} \xrightarrow{\cmdnrindx{c,1}}    \ldots \xrightarrow{\cmdnrindx{c,l}}   \xrightarrow{\cmdnr} $$
We will now show that we can apply a series of changes to Sequence \ref{eq:project:ib1} that ensure that it has a prefix that is identical to Sequence \ref{eq:project:ib2}.

Let $x$ be the smallest index such that $\evt_{a,x}\neq \evt_{b,x} $ (assuming it exists). 

\subparagraph{Case A1.  There is a $\evt_{i,j}$ in Sequence \ref{eq:project:ib1}  such that it is the first occurence of event type $\evt_{b,x}$ after $\evt_{a,x}$:}
$$\G  \xrightarrow{\cmdnrindx{a,1}}    \ldots  \xrightarrow{\cmdnrindx{a,x}}    \ldots  \xrightarrow{\cmdnrindx{i,j}}\ldots  \xrightarrow{\cmdnr} $$
Since $x$ is the smallest index where the sequences diverge, it holds $\G_{a,(x-1)}=\G_{b,(x-1)}$.

\subparagraph{Argument for Concurrency.} Since both transitions $\evt_{a,x}$ and $\evt_{b,x}$ are outgoing from $\G_{a,(x-1)}=\G_{b,(x-1)}$ and $\role$ does not subscribe to at least one of the event types, it follows that either the event types $\evt_{a,x}$ and $\evt_{b,x}$ are concurrent with each other or they are branching. If they are branching, then determinacy requires that they do not satisfy condition $\role \in \newroles[\evt_{a,x}][\G_{a,(x-1)}]$ of \autoref{def:determinacy}.%
Thus, $\evt$ is not causally dependent on $\evt_{a,x}$: There is no subsequence of Sequence \ref{eq:project:ib1} starting at $\evt_{a,x}$ and ending at $\evt$ where each event type is non-concurrent with its successor (see definition of $\newroles$ in \Cref{def:determinacy}).
This means we can move $\evt_{a,x}$ and all events that are causally dependent on it after $\evt$.
This means $\evt_{a,x}$ can be removed. This is a contradiction to the sequence being minimal. Thus, the event types are concurrent.

Since $\evt_{b,x}$ is concurrent with $\evt_{a,x}$ and they are both outgoing from $\G_{a,(x-1)}=\G_{b,(x-1)}$, this means $\G_{a,(x-1)} \xrightarrow{\evt_{a,x}} \G_{a,(x)} \xrightarrow{\evt_{b,x}}$. 
Thus, $\evt_{b,x}$ and $\evt_{a,x+1}$ are both outgoing from $\G_{a,(x)}$. We can apply the Argument for Concurrency, which means they are concurrent.
We iteratively apply this argument to show that $\evt_{b,x}$ is concurrent with every event type between $\evt_{a,x}$ and $\evt_{i,j}$ in Sequence \ref{eq:project:ib1}.
It follows that we can move $\evt_{i,j}$ directly before  $\evt_{a,x}$ and still get a valid interleaving:
$$\G  \xrightarrow{\cmdnrindx{a,1}}    \ldots  \xrightarrow{\cmdnrindx{i,j}} \xrightarrow{\cmdnrindx{a,x}}    \ldots   \xrightarrow{\evt} $$

\subparagraph{Case A2. $\evt_{b,x}$ does not occur in Sequence \ref{eq:project:ib1} after $\evt_{a,x}$:}
Here, we can apply the Argument for Concurrency from Case A1 and show that $\evt_{b,x}$ is concurrent with every event type after $\evt_{a,x}$ in Sequence \ref{eq:project:ib1}.
This means we can append $\evt_{b,x}$ at the end of the sequence:
$$\G  \xrightarrow{\cmdnrindx{a,1}}  \ldots \xrightarrow{\cmdnrindx{a,x}}    \ldots \xrightarrow{\evt} \xrightarrow{\evt_{b,x}}.$$
Since it is concurrent with $\evt_{a,x}$ and every event type afterwards in Sequence \ref{eq:project:ib1}, we can move the appended $\evt_{b,x}$ directly before $\evt_{a,x}$.

\subparagraph{Conclusion.}
Each time we apply this change, the prefix of Sequence \ref{eq:project:ib1} that is equal to Sequence \ref{eq:project:ib2} gets one transition longer.
We keep applying these changes until $\G_{a,m}=\G_{b,m}$. If no such diverging index existed, then this holds initially and no application were necessary.
If the next transition is $\G_{a,m} \xrightarrow{\evt_{a,m+1}}$, then we move $\evt_{1}$ directly before $\xrightarrow{\evt_{a,m+1}}$ in Sequence \ref{eq:project:ib1}, using the same concurrency argument as above.

We have now adapted Sequence \ref{eq:project:ib1} such that it starts with the prefix $$\G  \xrightarrow{\cmdnrindx{b,1}} \G_{b,1}   \ldots \xrightarrow{\cmdnrindx{b,m}} \G_{b,m} \xrightarrow{\cmdnrindx{1}} \G'_1$$ followed by a sequence 
$$ \G'_1 \ldots \xrightarrow{\evt_{2}}\G''_2 \ldots \xrightarrow{\evt_{n}}\G''_n \xrightarrow{e}\G''_{n+1}.$$
Thus, we have constructed a sequence
with a prefix that is equal to Sequence \ref{eq:project:ib2} until $\G'_{1}$.
We remove the prefix and get the following sequences for the series $\evt_2 \dots \evt_n$ of length $(n-1)$:
\begin{eqnarray}
	\G'_1  \proterm \receiving{\evt_2}\ldots \receiving{\cmdnrindx{n}} \G''_n \proterm &\receiving{\cmdnr} \G''_{n+1}\proterm\\
	\G'_1 \proterm \receiving{\evt_2}  \ldots \receiving{\cmdnrindx{n}} \G'_n \proterm&
\end{eqnarray}
We now apply the induction hypothesis to conclude that there is a $\G'_{n+1}$ with $\G'_n \proterm \receiving{\cmdnr} \G'_{n+1}\proterm$. 
Since Sequence \ref{eq:appendablelast1} also ends in  $\G'_{n}$, we can append $\G'_n \proterm \receiving{\cmdnr} \G'_{n+1}\proterm$ to it:
\begin{eqnarray}
	\G \proterm \receiving{\cmdnrindx{1}} \G_1  \proterm \ldots \receiving{\cmdnrindx{n}} \G_n \proterm \; &\receiving{\cmdnr} \G_{n+1}\proterm\label{eq:appendablelast}\\
	\G \proterm \receiving{\cmdnrindx{1}} \G'_1 \proterm  \ldots \receiving{\cmdnrindx{n}} \G'_n \proterm.&\label{eq:appendablelast1}
\end{eqnarray}
\medspace

\subsection*{Case B: Emitting Different Event Types.}  
We have shown that projections of the same swarm protocol receive the same log types. 
It remains to show that they can always emit the same set of event types after receiving the same log type. This is easy to see since the event types a projection can emit are determined by the events from its own role $\role$ that it can receive:

Examine $\G_n \proterm$ from Case A. It holds that  $\G_n\proterm$ emits event type $\evt$ iff there is an $\cmd$ and some sequence $$\G_n \xrightarrow{\cmdnrindx{a,1}} \ldots \xrightarrow{\cmdnrindx{a,m}} \G''\xrightarrow{\evt} \G_{n+1}$$ where $\role$ does not subscribe to any $\evt_{a,1}\ldots \evt_{a,m}$. 
This is the case iff there is an $\cmd$ and $\G_n \proterm \receiving{\cmdnr}$. 
Since $\G'_n \proterm$ can receive the same event types, it follows $\G'_n \proterm \receiving{\cmdnr}$ and thus $\G'_n \proterm $ can emit $\evt$.

\section{Formal Definitions of Effective logs}
\label{sec:app:effective-logs}
To formally define eventual fidelity, we introduce the concept of \emph{effective global logs} (\Cref{def:global-effective-log}),
\ecoop{obtained by filtering out all events $\ev$ that have the wrong type} according to $\G$, or have the wrong pointer $\ev.\lbjePointer$;
to this end, \Cref{def:swarm-proto-branch-tracking-semantics} defines a transition system of a swarm protocol with branch tracking, analogue to the transition systems of machines in \Cref{def:branch-tracking-transition-function-machines}.

\begin{definition}[Global effective log]
  \label{def:swarm-proto-branch-tracking-semantics}
  \label{def:global-effective-log}
  Let $\ev$ be of type $\evt$, $\aevent.\lbjePointer = \lbje(\evt)$, and $\G \xrightarrow{\evt} \G'$. We extend the partial function $\xrightarrow{\ev}$ to swarm protocols as follows: %
  \emph{(1)} if $\evt$ is not updating then $\parSt[\G] \xrightarrow{\ev} \parSt[\G']$; %
  and \emph{(2)} if $ \evt$ is updating, then $\parSt[\G] \xrightarrow{\ev} \parSt[\G'][{\lbje[\branch(\evt) \myupd \ev ]}]$.

	We define the global effective log filtering function $\effof{\lg}{\G}$ as:

  \ifbool{ecoop}{\smallskip}{\medskip}%
  \centerline{\(
	  \begin{array}{rcl}
	  	\effof{\epsilon}{\G} &=& \epsilon
      \\
	  	\effof{\ev\cdot \lg}{\G} &=&
      \begin{cases}
        \ev \cdot \effof{\lg}{\G'}[\lbje'] &
	  	  \text{if }  \parSt[\G]  \xrightarrow{\ev} \parSt[\G'][\lbje']
        \\
        \effof{\lg}{\G}& \text{otherwise}
      \end{cases}
	  \end{array}
  \)}%
  \ifbool{ecoop}{\smallskip}{\medskip}%

	Given a log $\lg$, we define the \emph{global effective log for $\G$} as $\srteffof[\lg][\G] \define \effof{\lg}{\G}[\emptyset]$.
\end{definition}

\ecoop{Similarly to \Cref{def:swarm-proto-branch-tracking-semantics}, the effective log of a machine $\M$ (\Cref{def:machine-effective-log}) only includes the events in a log that $\M$ actually accepts.}

\begin{definition}[Effective Log of a Machine]
  \label{def:machine-effective-log}
  \label{def:machine-proto-role-effective-log}
	Using $\xrightarrow{\ev}$ from \Cref{def:branch-tracking-transition-function-machines}, we define the effective log filtering function $\effof{\lg}{\M}$ for a machine $\M$:

	\ifbool{ecoop}{\smallskip}{\medskip}%
  \centerline{\(
	  \begin{array}{rcl}
		  \effof{\epsilon}{\M} &=& \epsilon \\
		  \effof{\ev\cdot \lg}{\M} &=&
      \begin{cases}
        \ev \cdot \effof{\lg}{ \M'}[\lbje'] & \text{if }  \parSt  \xrightarrow{\ev}  \parSt[\M'][\lbje']\\
		    \effof{\lg}{\M} & \text{otherwise}
      \end{cases}
    \end{array}
  \)}%
  \ifbool{ecoop}{\smallskip}{\medskip}%

  Given a log $\lg$, we define the \emph{effective log for $\M$} as $\srteffof[\lg][\M]\define\effof{\lg}{\M}[\emptyset]$.

\end{definition}

\section{Proof of Eventual Fidelity}\label{sec:app:fidelity-proof}

In order to prove eventual fidelity, we require some technical results:
 If an event $\ev$ has type $\evt$, we say that $\role$ \emph{subscribes to} $\ev$ if $\evt \in \subscription(\role)$.
 
We define $(\parSt[\G], \epsilon)\define \parSt[\G]$. 
If $\parSt[\G] \xrightarrow{\ev} \parSt[\G'][\lbje']$, then $(\parSt[\G], \ev.\lg)\define (\parSt[\G'][\lbje'], \lg) $, otherwise $(\parSt[\G], \ev.\lg)\define (\parSt[\G], \lg)$.
We set $(\G,\lg)\define  (\parSt[\G][\initlbje],\lg)$.

We define the projection of a parameterized swarm protocol $\parSt[\G]$ as the projection of the non-parameterized version of the swarm protocol $\G$ with the same parameter: $\parSt[\G] \proterm \define \parSt[\G \proterm]$.

We write $\parSt[\G\proterm] \sim_\role \parSt[\G'\proterm][\lbje']$ if $\G \proterm = \G' \proterm$ and for any $\evt \in \subscription(\role)$ holds $\lbje(\evt)=\lbje'(\evt)$.
We call event types that are not updating \emph{simple event types}. These are the event types that don't update the branch tracking function $\lbje$.

\begin{lemma}\label{lem:equivprojections}
	Let $\G$ be $\subscription$-well-formed and $\projeffof{\lg}{\G}[\role]=\srteffof \downarrow_{\subscription(\role)}$. It follows $(\G  \proterm ,\lgname{l}) \sim (\G ,\lgname{l}) \proterm $.
\end{lemma}
\begin{proof}
	Let $\parSt[\G'\proterm][\lbje']\define (\G  \proterm ,\lgname{l}) $ and $\G''\proterm [\lbje'']\define (\G ,\lgname{l}) \proterm $.
	\Cref{lem:equalprojections} ensures $\G \proterm = \G' \proterm$.
	Assume towards contradiction that there is an event type $\evt\in \subscription(\role)$ with $\lbje'(\evt)\neq \lbje''(\evt)$. 
	Note that the update function for the parameter of a swarm protocol is the same as the update function for the parameter of a projection.
	According to $\projeffof{\lg}{\G}[ \role]=\srteffof \downarrow_{\subscription(\role)}$, any update of the parameter that is performed by the projection is also performed by the swarm protocol.
	This means there was an update of the swarm protocol parameter caused by some event $\ev_b$ of type $\evt_b$ with $\evt \in \branch(\evt_b)$, which was not performed by the projection.
	Thus, $\ev_b$ is a branching or joining 
	event that $\role$ does not subscribe to. This means the condition $\role \in \newroles[\evt_b][\G_c]$ of \Cref{def:determinacy} is not satisfied.
	Accordingly, it holds that $\evt$, which is in $\branch(\evt_b)$, is not subscribed to by $\role$, since that would ensure  $\role \in \newroles[\evt_b][\G_c]$.
\end{proof}
\newtheorem*{theorem*}{Theorem}
\begin{theorem*}[Eventual Fidelity - \Cref{thm:eventual_fidelity}]
	Every realization of a $\subscription$-well-formed swarm protocol $\G$ is eventually faithful to $\G$ and $\subscription$.
\end{theorem*}
\begin{proof}
	Assume towards contradiction that it does not hold and there is a  $(\S,\epsilon) \rightarrow^* (\S,\lg)$ and a longest prefix $\lg'$ with $\lg=\lgname{l'.e.l''}$ and $\lg'$  satisfies eventual fidelity with logs.
	This means there is a role $\role$ such that $\projeffof{\lg'}{\G}[\role]=\srteffof[\lg'][ \G] \downarrow_{\subscription(\role)}$ but $\projeffof{\lgname{l'.e}}{ \G}[\role]\neq \srteffof[\lgname{l'.e}][\G] \downarrow_{\subscription(\role)}$.
	Assume $ \ev $ has type $ \evt$. 
	There are two possibilities:
	
	\paragraph*{Case 1: $\ev\notin \projeffof{\lgname{l'.e}}{\G}[\role], \ev \in \srteffof[\lgname{l'.e}][\G] \downarrow_{\subscription(\role)}$.}
		This means there is a transition $(\G ,\lgname{l'})\stackrel{\lgname{e}}\rightarrow (\G ,\lgname{l'.e})$ and $\evt \in \subscription(\role) $.
		It follows that there is a transition %
		$(\G ,\lgname{l'})\proterm \stackrel{\evt}\rightarrow (\G ,\lgname{l'.e}) \proterm$.
		According to \Cref{lem:equivprojections}, it holds $(\G ,\lgname{l'})\proterm \sim_\role(\G \proterm ,\lgname{l'})$.
		Since $\evt \in \subscription(\role) $, this means that the function $\lbje$ of $(\G ,\lgname{l'})\proterm$ assigns the same value to $\evt$ as the corresponding function of $ (\G \proterm ,\lgname{l'})$.
		From this and $(\G ,\lgname{l'})\proterm \stackrel{\evt}\rightarrow (\G ,\lgname{l'.e}) \proterm$ follows $(\G \proterm ,\lgname{l'}) \stackrel{\evt}\rightarrow (\G \proterm,\lgname{l'.e})$.
		This means $\ev\in\projeffof{\lgname{l'.e}}{\G}[\role]$ which contradicts the assumption of Case 1.

	\paragraph*{Case 2: $\ev\in \projeffof{\lgname{l'.e}}{\G}[\role], \ev \notin\srteffof[\lgname{l'.e}][\G] \downarrow_{\subscription(\role)}$, $\evt$ is not joining and $\evP{\ev}\neq \mathtt{NULL}$.} 
	Here, $\ev$ is not an interface event that joins multiple concurrent branches together.
	Let $\evP{\ev}=\lgname{e'}$ and $ \lg'=\lg''.\ev'.\lg'''$.
	Since $\ev\in \projeffof{\lgname{l'.e}}{\G}[\role]$, we know that when $\ev$ was received by $\role$, it had the correct pointer value $\ev'$ in $\lbje$. This value could only have been added to the function when receiving $\ev'$.
	This means $\ev'\in \projeffof{\lgname{l'.e}}{\G}[\role]$. 
	From $\projeffof{\lg'}{\G}[\role]=\srteffof[\lg'][\G] \downarrow_{\subscription(\role)}$ follows $\ev' \in \srteffof[\lg'][\G] $ as well.
	
	Let $\role'$ be the role that emitted $\ev$. This role subscribes to any event type that directly precedes $\evt$ non-concurrently according to \Cref{def:causalconsistency}.
	This means when it emitted $\ev$, it had received an event $\ev_1$ that directly preceded $\ev$. That event was either branching or joining (which means it has to be $\ev'$ according to branch tracking) or it was some simple event. 
	Assume it was a simple event.
	According to branch tracking semantics, this ensures $\evP{\ev_1}=\ev'$ and thus $\ev_1$ occurs in $\lgname{l'}$ after $\ev'$.
	We apply this argument iteratively until we have reached $\ev'$. This gives us a sequence $\lgname{simple}=\ev_n\ldots \ev_1$ of simple events. 
	If $\ev_1$ was not simple, and thus it was $\ev'$, then $\lgname{simple}$ is the empty sequence.
	It holds $\lgname{simple}$ is a subsequence of $\lg'''$.
	According to the construction of $\lgname{simple}$, it holds $(\G, \lg''.\ev')\xrightarrow{\ev_n\ldots \ev_1} (\G, \lg''.\ev'.\ev_n\ldots \ev_1)$.
	
	Note that $\ev_n$ is a simple event that directly succeeds $\ev'$.
	If $\ev_n\notin \srteffof[\lgname{l'.e}][\G]$, then this means there was an earlier event $\ev'_n$ in $\lgname{simple}$ that is a direct successor of $\ev'$ and that $\ev'_n$ also pointed to $\ev'$. Since $\ev_n$ is simple, it holds that $\ev'_n$ has the same type as $\ev_n$. 
	We set $\ev_n\define\ev'_n$.
	We repeat this step until $\ev_n \in \srteffof[\lgname{l'.e}][\G]$.
	We apply this method in sequence to $\ev_{n-1}\ldots \ev_1$ as well.
	Afterwards, $\lgname{simple}\define\ev_n\ldots \ev_1$ is a subsequence of $\srteffof[\lgname{l'}][\G]$ that occurs after $\ev'$.
	Let $\srteffof[\lgname{l'}][\G]=\lg_{eff1}.\ev_1.\lg_{eff2}$.
	According to the construction of $\ev_1$, $\ev$ is a direct successor of $\ev_1$. Since $ \ev \notin \srteffof[\lgname{l'.e}][\G]$, it follows that there was another direct successor $\ev''$ of $\ev_1$ in $\lg_{eff2}$ that is not concurrent with $\ev$. 
	\begin{description}
		\item[Case 2A: $\evt$ is simple:] Note that $\ev''$ has the same type $\evt$ as $\ev$ (and thus $\role$ subscribes to it) and the same pointer. This means $\ev''\in \projeffof{\lgname{l'}} {\G}[ \role]$ and $\ev\in \projeffof{\lgname{l'}} {\G}[ \role]$. 
		Since they have the same pointer, that means between receiving $\ev''$ and $\ev$, $\role$ only received events that were either simple or concurrent with $\ev$. This means there is path $\evt^s_1 \ldots \evt^s_n$ of simple events types in the swarm protocol such that $ \evt^s_1, \evt^s_n = \evt$ and $n>1$. This is a contradiction to the looping condition of Determinacy.
		
		\item[Case 2B: $\evt$ is branching and not joining:] This means that $\ev''$ can also be another branching option to $\evt$ and thus can be another type then $\ev$. Note that $\ev''$ has the same pointer as $\ev$ (and thus $\role$ subscribes to it).
		This means $\ev''\in \projeffof{\lgname{l'}} {\G}[\role]$ and $\ev''$ is not concurrent with $\ev$. This means $\ev'$ is not the last branching event that precedes $\ev$ and thus $\ev$ pointer to $\ev'$ is not correct. It follows $\ev\notin \projeffof{\lgname{l'.e}}{\G}[\role]$ which is a contradiction to the assumption of Case 2.	
		
		\item[Case 2C: $\evt$ is looping in $\subscription$ and not joining and not branching:] Note that $\ev''$ has the same (updating) type as $\ev$ and the same pointer as $\ev$ (and thus $\role$ subscribes to it).
		This means $\ev''\in \projeffof{\lgname{l'}} {\G}[\role]$ and $\ev''$ is not concurrent with $\ev$. This means $\ev'$ is not the last updating event that precedes $\ev$ and thus $\ev$ pointer to $\ev'$ is not correct. It follows $\ev\notin \projeffof{\lgname{l'.e}}{\G}[\role]$ which is a contradiction to the assumption of Case 2.	
	\end{description}
	
	\paragraph*{Case 3: $\ev\in \projeffof{\lgname{l'.e}}{\G}[\role], \ev \notin\srteffof[\lgname{l'.e}][\G] \downarrow_{\subscription(\role)}$ and $\evt$ is joining and $\evP \ev \neq \mathtt{NULL}$.} 
	It holds $(\G \proterm ,\lgname{l'}) \xrightarrow{\ev}$ and $(\G \proterm ,\lgname{l'}) \sim_\role (\G ,\lgname{l'})\proterm$ and thus $(\G  ,\lgname{l'})\proterm \xrightarrow{\ev}$.
	According to the definition of projections, there is a log $\lg$ that $\role$ does not subscribe to, such that $(\G  ,\lgname{l'}) \xrightarrow{\lg . \ev}$. We assume WLOG that $\lg$ does not contain event types concurrent with $\evt$.
	Here, $\role$ subscribes to all event types that directly precede $\evt$ according to \Cref{def:determinacy} (Joining). This means $\lg$ is empty and thus $(\G  ,\lgname{l'})\proterm \xrightarrow{\ev}$. It follows $\ev \in\srteffof[\lgname{l'.e}][\G]$ which is a contradiction to $\ev \notin\srteffof[\lgname{l'.e}][\G]\downarrow_{\subscription(\role)}$.
	\paragraph*{Case 4: $\evP \ev = \mathtt{NULL}$.} 
	This case is analogue to $\evP{\ev}=\ev'$ since the cases only argue about the part of the log that follows $\ev'$. The only difference is that $\lg''$ is empty and $\ev'$ is replaced by a dummy event $\mathtt{NULL}$ that everyone received initially.
\end{proof}

\section{Swarm Composition Algorithm}\label{sec:app:swarm_composition_algorithm}
We now use \Cref{def:machine-adaptation} to introduce
\Cref{alg:swarm_composition}, which constructs a composed swarm from a set of input swarms $\S_1, \ldots \S_n$ (all having empty logs) by adapting each input swarm machine $\M$ for the composed swarm protocol $\G_1 \conc \ldots \conc \G_n$.
First, we obtain the well-formed subscription $\subscription$ for the overall composed swarm protocol (line~\ref{alg:line0});
Then, traversing the input swarms, we select each machine $\M$ and its role $\role$ (lines~\ref{alg:lineSelM} and \ref{alg:lineSelR});
finally, we adapt $\M$ into $\M'$ (line~\ref{alg:lineAdapt}) and we add $\M'$ to the output swarm $\S$ (line~\ref{alg:lineAddOutS}, where the operator $\_ \swarmAdd \_$ adds \ecoop{$\M'$} to the swarm $\S$ with a suitable unique index and an empty local log).

\begin{algorithm}[t]
	\begin{description}
		\item[Input:]
		\begin{itemize}
			\item Confusion-free, pairwise-interfacing swarm protocols $\G_1 \ldots \G_n$ without concurrent event types;
			\item Subscriptions $\subscription_1\ldots \subscription_n$; %
			\item Swarms $\S_1 \ldots \S_n$ and mappings $\roleMap[1] \ldots \roleMap[n]$ such that, for $i \in \{ 1, \ldots, n \}$:
			\begin{itemize}
				\item $\S_i$ has empty logs and realises $\G_i$ for subscriptions $\subscription_i$;
				\item $\roleMapApp[i]{j}$ is the role in $\G_i$ that is realised by the machine $\S_i(j)$.
			\end{itemize}
		\end{itemize}
		\item[Output:] 
		\begin{itemize}
			\item A subscription $\subscription \supseteq \subscription_i$ (for $i \in \{1, \ldots, n\}$)
			such that $\G_1 \conc \ldots \conc \G_n$ is $\subscription$-well-formed;
			\item A swarm $\S$ that \ecoop{realises $\G_1 \conc \ldots \conc \G_n$ and} is constructed by adapting $\S_1 \ldots \S_n$.
		\end{itemize}
	\end{description}
	\begin{algorithmic}[1]
		\STATE Construct $\subscription$ %
		using \Cref{alg:comp-swarm-subscriptions} (by \Cref{lem:subscription_composition_WF}).\label{alg:line0}
		\STATE Let $\S = \left(\emptyset, \emptylog\right)$. \hfill\emph{(Output swarm with no machines and empty global log)}
		\FORALL{$k \in \{ 1, \ldots, n \}$}
		\FORALL{$j \in \dom(\S_k)$}
		\STATE $\M \;\define\; \S_k(j)$ \hfill\emph{(Select the $j$\textsuperscript{th} machine $\M$ from swarm $\S_k$)}\label{alg:lineSelM}
		\STATE $\role \;\define\; \roleMapApp[k]{j}$ \hfill\emph{(Role of machine $\M$ in swarm protocol $\G_k$)}\label{alg:lineSelR}
		\STATE $\M' \;\define\; \mathcal{A}\!\left(\M, {(\G_i)_{i \in \{ 1, \ldots, n \}}, \role}, k, \subscription\right)$ \hfill\emph{(Adapt $\M$ using \Cref{def:machine-adaptation})}\label{alg:lineAdapt}
		\STATE $\S \;\leftarrow\; \M' \swarmAdd \S$ \hfill\emph{(Add adapted machine $\M'$ with empty log to output swarm)}\label{alg:lineAddOutS}
		\ENDFOR
		\ENDFOR
	\end{algorithmic}
	\caption{Swarm Composition.}\label{alg:swarm_composition}
\end{algorithm}

\section{Proof of \Cref{lem:projection_composition_correct}}\label{sec:app:proof_of_projection_composition_correct}
	We show that the constructed swarm is a realization of the composed swarm protocol.
		The algorithm adapts each input machine $\M=\G_k \projshort{\subscription_k}{\role}$ (which is a projection of of some $\G_k$ onto a role $\role$) using \Cref{def:machine-adaptation}. 
		We will show that each adaptation is a projection of the swarm protocol composition:
	$\ \mathcal{A}\!\left(\M, (\G_i)_{i \in \{ 1, \ldots, n \}}, \role, k, \subscription\right)=(\G_1 \conc \ldots \conc \G_n) \proterm$.
	
	\subsection{Technical Requirement}We begin by showing that $\M \conc \G_k \projshort{\subscription}{\role}=\G_k \projshort{\subscription}{\role}$ holds:
	Since $\subscription_k \subseteq \subscription$, any event type of $\G_k \projshort{\subscription_k}{\role} $ also occurs in $\G_k \projshort{\subscription}{\role} $.
	Thus any transition in the composition is either only a transition in $ \G_k \projshort{\subscription}{\role}$ or in $\M$ as well.
	Since $\G_k \projshort{\subscription_k}{\role} $ only differs from $\G_k \projshort{\subscription}{\role} $ by skipping some transitions, it does not restrict the behavior of $\G_k \projshort{\subscription}{\role} $ in the composition.
	This means that $ \G_k \projshort{\subscription}{\role}=\G_k \projshort{\subscription_k}{\role} \conc \G_k \projshort{\subscription}{\role}$ holds.
	The condition follows.

	\subsection{Correct Realisation}  
	The algorithm adapts each each input machine $\M=\G_k \projshort{\subscription_k}{\role}$ using \Cref{def:machine-adaptation}. 
	We will show that each adaptation is a projection of the swarm protocol composition:
	$\ \mathcal{A}\!\left(\M, (\G_i)_{i \in \{ 1, \ldots, n \}}, \role, k, \subscription\right)=(\G_1 \conc \ldots \conc \G_n) \proterm$.
	
	Recall 	       
	$$ \mathcal{A}\!\left(\M, {(\G_i)_{i \in \{ 1, \ldots, n \}}, \role}, k, \subscription\right)%
	\;=\;%
	\left(%
	\G_1 \projshort{\subscription}{\role} \concM{} \ldots \concM{}%
	\left(\M \concM{} \G_k \projshort{\subscription}{\role}\right) %
	\concM{} \ldots \concM{} \G_n \projshort{\subscription}{\role}
	\right).$$
	
	We set $\M_k\define	\left(\M \concM{} \G_k \projshort{\subscription}{\role}\right)$ and $\M_j\define \G_j \proterm$ for $j\neq k$. 
	It remains to show that $\M_1 \conc \ldots \M_n = (\G_1 \conc \ldots \conc \G_i )\proterm$ holds, since the condition $\ \mathcal{A}\!\left(\M, {(\G_i)_{i \in \{ 1, \ldots, n \}}, \role}, k, \subscription\right)=(\G_1 \conc \ldots \conc \G_n) \proterm$ follows immediately.
	We use  an induction over $i \leq n$:
	Let $\M_i^\role\define \M_1 \conc \ldots \conc \M_i$ for $i\leq n$.
	\paragraph*{Induction Basis $i=1$:} It holds $\M^\role_1=\M_1 = \G_1 \proterm$. This holds either by definition or due to $ \G_i \projshort{\subscription}{\role}=\G_1 \projshort{\subscription_1}{\role} \conc \G_1 \projshort{\subscription}{\role}$ for k=1, as shown above.
	\paragraph*{Induction Hypothesis:} It holds $\M_i^\role = (\G_1 \conc \ldots \conc \G_i )\proterm$.
	\paragraph*{Induction Step $i\rightarrow i+1$:}
	First, we show that for any transition $$ (\G_1 \conc \ldots \conc \G_{i+1} )\proterm  \receiving{\evt} (\G'_1 \conc \ldots \conc \G'_{i+1}) \proterm,$$ there is a corresponding transition $\M_{i+1}^\role   \receiving{\evt} (\G'_1 \proterm \conc \ldots \conc \G'_{i+1}\proterm) $.
	According to the composition of swarm protocols, the transition  $ (\G_1 \conc \ldots \conc \G_{i+1} )\proterm  \receiving{\evt} (\G'_1 \conc \ldots \conc \G'_{i+1}) \proterm$ satisfies one of these cases:
	\begin{description}
		\item[Case 1 ($\evt \not\in \G_{i+1}$):]
		It holds  $(\G_1 \conc \ldots \conc \G_{i}) \proterm \receiving{\evt} (\G'_1 \conc \ldots \conc \G'_{i}) \proterm$
		and thus $$(\G'_1 \conc \ldots \conc \G'_{i+1}) \proterm= (\G'_1 \conc \ldots  \conc \G'_i \conc \G_{i+1}) \proterm.$$
		According to the Induction Hypothesis, it holds $\M_i^\role \receiving{\evt} (\G'_1 \conc \ldots \conc \G'_{i}) \proterm$.
		We apply the machine composition and we derive $\M_i^\role \conc (\G_{i+1} \proterm) \receiving{\evt}  (\G'_1 \conc \ldots \conc \G'_{i}) \proterm \conc \G_{i+1}\proterm$ since $\evt$ does not occur in $\G_{i+1} \proterm$.
		It holds  $\M_{i+1}^\role=\M_i^\role \conc (\G_{i+1} \proterm)$ and thus there is a corresponding transition from $\M_{i+1}^\role$.
		\item[Case 2 ($\evt \in \G_{i+1}$ and $\evt \not\in (\G_1 \conc \ldots \conc \G_{i})$):]  	
		It holds $$(\G'_1 \conc \ldots \conc \G'_{i+1}) \proterm= (\G_1 \conc \ldots  \conc \G_i \conc \G'_{i+1}) \proterm.$$
		According to the Induction Hypothesis, it holds $\evt \not\in \M^\role_i $.
		We apply the machine composition ($\M^\role_{i+1}= \M_i^\role \conc (\G_{i+1} \proterm)$) and we derive 
		$$\M_i^\role \conc (\G_{i+1} \proterm) \receiving{\evt}  (\G_1 \conc \ldots \conc \G_{i}) \proterm \conc \G'_{i+1}\proterm.$$
		\item[Case 3 ($\evt \in \G_{i+1}$ and $\evt \in (\G_1 \conc \ldots \conc \G_{i})$):]
		The event type is interfacing.
		According to the definition of swarm protocol composition, it holds $\G_{i+1} \proterm \receiving{\evt}  \G'_{i+1} \proterm$ as well as $$(\G_1 \conc \ldots \conc \G_{i}) \proterm \receiving{\evt} (\G'_1 \conc \ldots \conc \G'_{i}) \proterm.$$
		We apply the machine composition and we derive $\M_i^\role \conc (\G_{i+1} \proterm) \receiving{\evt}  (\G'_1 \conc \ldots \conc \G'_{i}) \proterm \conc \G'_{i+1}\proterm$.
	\end{description}
	\medskip
	
	Having shown that for any transition of $(\G_1 \conc \ldots \conc \G_{i+1} )\proterm$, there is a corresponding transition in $\M_{i+1}^\role$,
	we will now show that,
	for any transition $\M_{i+1}^\role   \receiving{\evt} (\G'_1 \proterm \conc \ldots \conc \G'_{i+1}\proterm)$ , there is a corresponding transition $ (\G_1 \conc \ldots \conc \G_{i+1} )\proterm  \receiving{\evt} (\G'_1 \conc \ldots \conc \G'_{i+1}) \proterm$.
	We examine the different cases for $\M_{i+1}^\role   \receiving{\evt} (\G'_1 \proterm \conc \ldots \conc \G'_{i+1}\proterm)$:
	\begin{description}
		\item[Case A ($ \G_{i+1} \proterm \receiving{\evt} \G'_{i+1} \proterm$ and $\evt \not\in \M_i^\role$):]
		It holds $$ (\G_1 \conc \ldots \conc \G_{i} )\proterm \conc  \G_{i+1} \proterm \receiving{\evt} (\G_1 \conc \ldots \conc \G_{i} )\proterm \conc  \G'_{i+1} \proterm.$$
		There is a logtype $\lgt$ of minimal length such that $\lgt \cap \subscription(\role) =\emptyset$ and $\G_{i+1} \receiving{\lgt}\receiving{\evt} \G'_{i+1}$.
		Note that $\lgt$ does not contain any interfacing event types since $\lgt \cap \subscription(\role) =\emptyset$ and we constructed $\subscription$ in such a way that all interfacing event types  are subscribed to if other subscribed event types are causally dependent on them. Since $\lgt$ is minimal, it only contains event types that $\evt$ is causally dependent on.
		Since $\lgt$ does contains no interfacing event types, it holds $ (\G_1 \conc \ldots \conc \G_{i+1} )\proterm \receiving{\evt} (\G_1 \conc \ldots \conc \G_{i}  \conc  \G'_{i+1}) \proterm$.
		\item[Case B ($ \M_i^\role \receiving{\evt}$ and $\evt \not\in  \G_{i+1} \proterm$):]
		This case is analogue to the previous one.
		\item[Case C ($ \G_{i+1} \proterm \receiving{\evt}\G'_{i+1} \proterm$ and $\M_i^\role \receiving{\evt} (\G'_1 \conc \ldots \conc \G'_{i}) \proterm$):]
		Note that $\evt$ is an interfacing event type.
		There are $\lgt_1,\lgt_2$ such that $(\lgt_1 \cup \lgt_2) \cap \subscription=\emptyset$, $ \lgt_1\cap \lgt_2 =\emptyset$,  $ \G_{i+1} \xrightarrow{\lgt_1} \xrightarrow{\evt}\G'_{i+1} $, and $$ (\G_1 \conc \ldots \conc \G_{i} ) \xrightarrow{\lgt_2} \xrightarrow{\evt} (\G'_1 \conc \ldots \conc \G'_{i}) .$$  Note that $\lgt_1$ and $\lgt_2$ are not subscribed to and we can assume they are minimal. Thus they do not contain any interfacing event types. This means it holds
		$$ (\G_1 \conc \ldots \conc \G_{i+1} ) \xrightarrow{\lgt_1 .\lgt_2} \xrightarrow{\evt} (\G'_1 \conc \ldots \conc \G'_{i+1}) $$ and thus
		$ (\G_1 \conc \ldots \conc \G_{i+1} ) \proterm \receiving{\evt} (\G'_1 \conc \ldots \conc \G'_{i+1}) \proterm $.
	\end{description}
	Since the projections have the same transitions and the emitter set is determined by the transitions, the projections emit the same event types. Thus they are equal:
	$\M_1 \conc \ldots \conc \M_n =(\G_1 \conc \ldots \conc \G_n) \proterm$. The constructed swarm is almost a realisation of the $\subscription$-well-formed composed protocol $\G_1 \conc \ldots \conc \G_n$.
	The only realisation requirement not met concerns the set of updating events: for each machine $\M$, $\M.\JB$ is an over-approximation (rather than the exact set).
	However, this over-approximation is sound for our purposes: for any $\evt$ in that over-approximation, if $\G'\xrightarrow{\evt}$ and $\role \in \newroles$, then $\evt \in \subscription(\role)$.
	Moreover, all machines share the same set $\M.\JB$.
	Therefore, the proof of eventual fidelity in \Cref{sec:app:fidelity-proof} still applies under this over-approximation. 
	Thus,  \Cref{thm:eventual_fidelity} can still be applied for the over-approximated set of updating event types and the constructed swarm is eventually faithful to the $\subscription$-well formed composed protocol $\G = \G_1 \conc \ldots \conc \G_n$.
Thus, \Cref{lem:projection_composition_correct} is proven.

\section{Inference Rules for Verifying Well-Formedness}\label{sec:app:inference_rules}
\mkfun[\colorFun]{\concset}{conc}{}
\mkfun[\colorFun]{\sterms}{subterms}{}
\mkfun[\colorFun]{\composable}{composable}{}
\mkfun[\colorFun]{\bw}{bw}{}
\mkfun[\colorFun]{\jf}{jf}{}
\mkfun[\colorFun]{\rolesInfRule}{roles}{}

We define rules that allow us to check the $\subscription$-well-formedness of a swarm protocol compositions.
To decide whether $\agt[1] \conc \cdots \conc \agt[n]$ is $\subscription$-well-formed we rely on a derivation environment consisting of the sets $C$, $S$, and $V$, where:
\begin{itemize}
	\item $C$ is a set of unordered pairs of event types containing an overapproximation of the concurrent event types of $\agt[1] \conc \cdots \conc \agt[n]$.
	\item $S$ is the set of subterms of $\agt[1] \conc \cdots \conc \agt[n]$.
	\item $V$ is a set of subterms of $\agt[1] \conc \cdots \conc \agt[n]$. When we build a derivation bottom-up, we collect in this set the subterms that we have ``visisted'' and recursively propagate this set between conclusions and premises of the rules.
\end{itemize}

Furthermore, the inference rules make use of the auxiliary definitions in \Cref{def:auxruleswfwconc} below. 

\begin{definition}[Auxiliary definitions for inference rules.]\label{def:auxruleswfwconc}
	\phantom{Phanotom text to force new line. find better solution.}
	\begin{enumerate}[A.]
		\item The predicate $\composable(\{\agt[1], \ldots,  \agt[n]\})$ is true \emph{iff} $\{\agt[1], \ldots,  \agt[n]\}$ is composable according to \Cref{def:composable}.
		\item For a set of swarm protocols $SP = \{\agt[1], \ldots, \agt[n]\}$, we denote by $\concset(SP)$ an overapproximation of the set of concurrent event types in $\agt[1]\conc \cdots \conc \agt[n]$ given by:
		\[
			\concset(SP) = \{ \{\evt', \evt''\} \; | \; \exists \agt[i], \agt[j] \in SP.\; \evt' \in \agt[i] \land \evt' \notin \agt[j] \land \evt'' \in \agt[j] \land \evt'' \notin \agt[i] \}
		\]
		\item The set of subterms of a swarm protocol $\agt = \agsumnewnew$, is given by:
		\[
		\sterms(\agsumnewnew) = \{ \agsumnewnew \} \cup \bigcup_{i \in I} \sterms(\agt[i])
		\]
		\item The set of event types branching with $\evt$ at $\agt$ according to a set of concurrent event types $C$ is given by:
		\[
			\bw(\evt, \agt, C) = \{ \evt' \; | \agt \xrightarrow{\evt}\agt' \land \agt \xrightarrow{\evt'} \agt'' \land \evt \neq \evt' \land \agt' \neq \agt'' \land \{ \evt, \evt' \} \notin C \}
		\]
		\item The set of event types for which $\evt$ is joining at $\agt$ according to a set of concurrent event types $C$ and a set of subterms $S$ is given by:
		\begin{align*}
			\jf(\evt, \agt, C, S) = \{ \evt' \; |& \; \exists \agt[1], \agt[2] \in S.\; \exists \evt'' \in \agt[2].\; \evt' \neq \evt'' \land \agt[1] \xrightarrow{\evt'} \agt \land \agt[2] \xrightarrow{\evt''} \agt \\
			&\land \{\{\evt, \evt'\}, \{\evt, \evt''\}\} \not\subseteq C \land \{ \evt', \evt'' \} \in C \} \\
		\end{align*}
		\item The set of roles that subscribe to event types in $\agt$ that causally depend on $\evt$ according to the $\subscription$ and the set of concurrent event types $C$ is given by:
			\begin{align*}
				\rolesInfRule(\evt, \agt, \subscription, C) = \{ \role \; | \; &
				\text{there are } n \geq 0, \; \evt_0, \ldots, \evt_n, \; \lgt_1, \ldots, \lgt_n  \text{ such that: }
				\\&\evt_0 = \evt, \; \G \xrightarrow{\evt_0} \xrightarrow{\lgt_1} \xrightarrow{\evt_1} \cdots \xrightarrow{\lgt_n} \xrightarrow{\evt_n}
				\text{ and }
				\\& \evt_n \in \subscription(\role), \text{ and } \{ \evt_i, \evt_{i+1} \} \notin C \text{ for all } 0 \leq i < n \}.
			\end{align*}
	\end{enumerate}
\end{definition}

\begin{definition}[Inference rules for checking well-formedness of swarm protocol compositions]\label{def:ruleswfwconc}
	\phantom{Phanotom text to force new line. find better solution.}
	\item We denote by $WF(\agt[1] \conc \cdots \conc \agt[n])$ and $DCC(\agt[1] \conc \cdots \conc \agt[n])$ the sets of subscriptions with respect to which $\agt[1] \conc \cdots \conc \agt[n]$ is, respectively, well-formed, and determinate and causal-consistent.
	We think of the task of proving that $\agt[1] \conc \cdots \conc \agt[n]$ is $\subscription$-well-formed as the same as showing that $\subscription \in WF(\agt[1] \conc \cdots \conc \agt[n])$.
	To check this, we use the judgment $\wfInf{}{\subscription}{\emph{WF}(\agt[1] || \agt[2])}$ defined by the following inference rules:\\\\
	{\allowdisplaybreaks
		\resizebox{\textwidth}{!}{$\displaystyle
			\begin{array}{c}
				\inferrule*[right={\emph{[Init]}}]{
					\emph{\composable}(\{\agt[1], \ldots,  \agt[n]\})\\
					\wfInf{\concset(\{\agt[1], \ldots, \agt[n]\}),\sterms(\agt[1]\conc \cdots \conc \agt[n]), \emptyset}{\subscription}{\emph{DCC}(\agt[1]\conc \cdots \conc \agt[n])}
				}{
					\wfInf{}{\subscription}{\emph{WF}(\agt[1]\conc \cdots \conc \agt[n])}
				}\\\\\\
				\inferrule*[right={\emph{[Loop]}}]{
					\agt \in V\\
					\exists \evt', \lgt \in \agt.\; \exists \agt' \subterm \agt.\; \agt \xrightarrow{\lgt} \agt \land \evt' \in \lgt \land \agt' \xrightarrow{\evt'} \land \forall \role' \in \rolesInfRule(\evt', \agt', \subscription, C).\; \evt' \in \subscription(\role')
				}{
					\wfInf{C, S, V}{\subscription}{\emph{DCC}(\agt)}
				}\\\\\\
				\inferrule*[right={\emph{[Term]}}]{
					\agt = \agsumnewnew\\
					\forall i \in I .\; \vcenter{\hbox{\begin{tabular}{l}
								\phantom{$\land\;\;$}$\evt_i \in \subscription(\role_i)$ \\
								$\land\;\; \forall \role' \in \{ \role \; | \; \agt[i] \xrightarrow{\cstmcmd{\role}{\evt}} \land \{ \evt_i, \evt \} \notin C \}.\; \evt_i \in \subscription(\role')$ \\
								$\land\;\; \bw(\evt_i, \agt, C) \neq \emptyset \Rightarrow \forall \role' \in \rolesInfRule(\evt_i, \agt, \subscription, C).\; \bw(\evt_i, \agt, C) \cup \{\evt_i\} \subseteq \subscription(\role')$  \\
								$\land\;\; \jf(\evt_i, \agt, C, S) \neq \emptyset \Rightarrow \forall \role' \in \rolesInfRule(\evt_i, \agt, \subscription, C).\; \jf(\evt_i, \agt, C, S) \cup \{\evt_i\} \subseteq \subscription(\role')$  \\
								$\land\;\; \wfInf{ C, S, V \cup \{\agt\} }{\subscription}{\emph{DCC}(\agt[i])}$ 
					\end{tabular}}}
				}{
					\wfInf{C, S, V}{\subscription}{\emph{DCC}(\agt)}
				}
			\end{array}
			$}}
\end{definition}

\section{Non-branch-tracking Swarm} \label{sec:app:non_bt_swarm}
The swarm semantics from \cite{DBLP:conf/ecoop/KuhnMT23} (\Cref{def:non-bt-swarm-semantics} below)
are similar to the semantics of \Cref{def:branch-tracking-swarm-semantics},
but does not involve branch-tracking.
These semantics ensures that a swarm
realising a swarm protocol $\G$ that is well-formed w.r.t. 
a subscription $\subscription$ according to the well-formedness 
definition from \cite{DBLP:conf/ecoop/KuhnMT23} is eventually faithful 
to $\G$ and $\subscription$. 

\begin{definition}[Non-Branch-Tracking Swarm Semantics (from \cite{DBLP:conf/ecoop/KuhnMT23})]
	\ecoop{Given a mapping $m$, let $m[a\myupd b]$ denote the mapping identical to $m$ except $m[a\myupd b](a)=b$.}
	\label{def:non-bt-swarm-semantics}
	We write
	$(\afish,\alog) \noBTMachineStep[\aeventtype !] (\afish,\alog\logcat\aevent)$ when machine $\afish$ emits an event $\aevent$ of type $\aeventtype$
	from state $\stchange(\afish,\alog)$, thus reaching state
	$\stchange(\afish, \alog \cdot \aevent)$.
	The behaviour of a swarm is defined by
	the smallest relation closed under
	the following rules:%
	
	\ifbool{ecoop}{\smallskip}{\medskip}\centerline{\small\(%
		\begin{array}{c}
			\mathrule{
				\asysrt(i) = (\afish[i],\alog[i])
				\quad
				(\afish[i],\alog[i]) \noBTMachineStep[{\aeventtype}!] (\afish[i],\alog[i]\cdot \aevent)
				\quad
				<_{\alog[i]} \subseteq <_{\alog}
			}{
				(\asysrt, \alog \cdot \alog') \noBTSwarmStep[{\aeventtype}!] (\upd \asysrt i {(\afish[i], \alog[i]\cdot \aevent)},  \alog \cdot \aevent \cdot \alog')
			}{Local}
			\ifbool{ecoop}{\quad}{ \\ \\}
			\mathrule{
				\asysrt(i) = (\afish[i],\alog[i])
				\quad
				<_{\alog[i] }  \subset <_{\alog'} \subseteq <_{\alog}
			}{
				(\asysrt, \alog) \noBTSwarmStep[\tau] (\upd \asysrt i {(\afish[i], {\alog'})}, \alog)
			}{Prop}
		\end{array}
		\)}%
\end{definition}

Well-formedness as defined in \Cref{def:wf}, however,
necessitates branch-tracking to ensure eventual fidelity. 
Using well-formedness as defined in \Cref{def:wf} together 
with the swarm semantics in \Cref{def:non-bt-swarm-semantics},
leads to swarm behaviours where machines can enter 
conflicting states, from where they cannot recover consistency.
This is shown in the following example.

\begin{example}[\ecoop{Incompatibilty between well-formedness in \Cref{def:wf} and swarm semantics in \cite{DBLP:conf/ecoop/KuhnMT23}}]
  \label{ex:non-bt-swarm-warehouse}
  Consider $\Gn{Warehouse}$ in \Cref{fig:Warehouse}.
  Let $E$ denote the set of all event types in $\Gn{Warehouse}$ and let $\subscription = \{ \rolename{T} \mapsto E, \rolename{FL} \mapsto E \setminus \{ \lgtn{partOK} \}, \rolename{D} \mapsto E \setminus \{ \lgtn{pos} \} \}$.
  The projections are given in \Cref{fig:ProjWH}.
  By \Cref{def:wf}, $\Gn{Warehouse}$ is $\subscription$-well-formed. \ecoop{Consider swarm $\S_\Gn{W}$ with:

  \ifbool{ecoop}{\smallskip}{\medskip}%
  \centerline{\(
  \asysrt_\Gn{W} \;=\; \left\{
      i_{{\rolename{T}}_1} \mapsto (\afish[{\rolename{T}}], \emptylog),\quad 
      i_{{\rolename{T}}_2} \mapsto (\afish[{\rolename{T}}], \emptylog),\quad 
      i_{{\rolename{FL}}} \mapsto (\afish[{\rolename{FL}}], \emptylog),\quad 
      i_{{\rolename{D}}} \mapsto (\afish[{\rolename{D}}], \emptylog)
    \right\}
  \)}}%
  \ifbool{ecoop}{\smallskip}{\medskip}%

  \noindent%
  where $\afish[{\rolename{T}}] = \Gn{Warehouse}\projshort{\subscription}{\rolename{T}}$, $\afish[{\rolename{FL}}] = \Gn{Warehouse}\projshort{\subscription}{\rolename{FL}}$,
  and $\afish[{\rolename{D}}] = \Gn{Warehouse}\projshort{\subscription}{\rolename{D}}$. Observe that the machines with IDs $i_{{\rolename{T}}_1}$ and $i_{{\rolename{T}}_2}$ both enact role $\rolename{T}$ (transport).
  Since all machines are obtained by projecting $\Gn{Warehouse}$ onto its roles using $\subscription$, the swarm $\asysrt_\Gn{W}$ realises the protocol $\Gn{Warehouse}$ w.r.t.~$\subscription$ (as defined in \cite{DBLP:conf/ecoop/KuhnMT23}).
  A possible execution of $\asysrt_\Gn{W}$ by \Cref{def:non-bt-swarm-semantics} (from \cite{DBLP:conf/ecoop/KuhnMT23}) is:

  \ifbool{ecoop}{\smallskip}{\medskip}%
  \noindent%
  \scalebox{0.7}{%
  \begin{minipage}{1\linewidth}
  \(%
  \begin{array}{r@{\;}r@{\;}l@{\quad}c@{\;}r@{\;}c@{\;}l}

    & (\asysrt_\Gn{W}, \emptylog) \noBTSwarmStep[{\lgtn{partReq}}!]& (\asysrt_1, \lgname{partReq}_1) \quad && \asysrt_1 &=& \asysrt_\Gn{W}{[i_{{\rolename{T}}_1} \myupd (\afish[{\rolename{T}}], \lgname{partReq}_1)]}\\
    & \noBTSwarmStep[\tau]& (\asysrt_2, \lgname{partReq}_1) \quad && \asysrt_2  &=& \asysrt_1{[i_{{\rolename{FL}}_1} \myupd (\afish[{\rolename{FL}}], \lgname{partReq}_1)]}\\
    & \noBTSwarmStep[{\lgtn{pos}}!]& (\asysrt_3, \lgname{partReq}_1\logcat\lgname{pos}_2) \quad && \asysrt_3 &=& \asysrt_2{[i_{\rolename{FL}} \myupd (\afish[{\rolename{FL}}], \lgname{partReq}_1\logcat\lgname{pos}_2)]}\\
    & \noBTSwarmStep[\tau]& (\asysrt_4, \lgname{partReq}_1\logcat\lgname{pos}_2)\quad && \asysrt_4 &=& \asysrt_3{[i_{{\rolename{T}}_1} \myupd (\afish[{\rolename{T}}], \lgname{partReq}_1\logcat\lgname{pos}_2)]}\\
    & \noBTSwarmStep[\tau]& (\asysrt_5, \lgname{partReq}_1\logcat\lgname{pos}_2)\quad && \asysrt_5 &=& \asysrt_4{[i_{{\rolename{T}}_2} \myupd (\afish[{\rolename{T}}], \lgname{partReq}_1\logcat\lgname{pos}_2)]}\\
    & \noBTSwarmStep[{\lgtn{partOK}}!]& (\asysrt_6, \lgname{partReq}_1\logcat\lgname{pos}_2\logcat\lgname{partOK}_3) \quad && \asysrt_6 &=& \asysrt_5{[i_{{\rolename{T}}_1} \myupd (\afish[{\rolename{T}}], \lgname{partReq}_1\logcat\lgname{pos}_2\logcat\lgname{partOK}_3)]}\\
    & \noBTSwarmStep[{\lgtn{partReq}}!]& (\asysrt_7, \lgname{partReq}_1\logcat\lgname{pos}_2\logcat\lgname{partOK}_3\logcat\lgname{partReq}_4) \quad && \asysrt_7 &=& \asysrt_6{[i_{{\rolename{T}}_1} \myupd (\afish[{\rolename{T}}], \lgname{partReq}_1\logcat\lgname{pos}_2\logcat\lgname{partOK}_3\logcat\lgname{partReq}_4)]}\\
    & \noBTSwarmStep[{\lgtn{partOK}}!]& (\asysrt_8, \lgname{partReq}_1\logcat\lgname{pos}_2\logcat\lgname{partOK}_3\logcat\lgname{partReq}_4\logcat\lgname{partOK}_5) \quad && \asysrt_8 &=& \asysrt_7{[i_{{\rolename{T}}_2} \myupd (\afish[{\rolename{T}}], \lgname{partReq}_1\logcat\lgname{pos}_2\logcat\lgname{partOK}_5)]}\\
    & \noBTSwarmStep[\tau]& (\asysrt_9, \lgname{partReq}_1\logcat\lgname{pos}_2\logcat\lgname{partOK}_3\logcat\lgname{partReq}_4\logcat\lgname{partOK}_5) \quad && \asysrt_9 &=& \asysrt_8{[i_{{\rolename{D}}} \myupd (\afish[{\rolename{D}}], \lgname{partReq}_1\logcat\lgname{pos}_2\logcat\lgname{partOK}_3\logcat\lgname{partReq}_4\logcat\lgname{partOK}_5)]}\\
    & \noBTSwarmStep[{\lgtn{closingTime}}!]& (\asysrt_{10}, \alog) \quad && \asysrt_{10} &=& \asysrt_9{[i_{{\rolename{D}}} \myupd (\afish[{\rolename{D}}], \alog)]}\\
    & \noBTSwarmStep[\tau]^*& (\asysrt_{11}, \alog) \quad && \asysrt_{11} &=& \asysrt_{10}{[i_{{\rolename{T}}_1} \myupd (\afish[{\rolename{T}}], \alog),i_{{\rolename{T}}_2} \myupd (\afish[{\rolename{T}}], \alog),i_{{\rolename{FL}}} \myupd (\afish[{\rolename{FL}}], \alog)]}\\
  \end{array}
  \)%
  \end{minipage}
}%
\ifbool{ecoop}{\smallskip}{\medskip}%

\noindent%
where the global log $\lg=\lgname{partReq}_1\logcat\lgname{pos}_2\logcat\lgname{partOK}_3 \logcat \lgname{partReq}_4 \logcat \lgname{partOK}_5 \logcat \lgname{closingTime}_6$
has propagated to every machine. %
(Recall that by \Cref{def:non-bt-swarm-semantics}, $\tau$-transitions
denote log propagation.)

Observe that in state $\S_5$, the machine with ID $i_{{\rolename{T}}_1}$ emits an event $\lgname{partOK}_3$ as a response to $\lgname{partReq}_1$;
however, the event does not immediately propagate to machine $i_{{\rolename{T}}_2}$, which (in state $\S_7$) responds to $\lgname{partReq}_1$ with another event $\lgname{partOK}_5$.
As the log $\alog$ propagates to $i_{{\rolename{T}}_2}$ (in state $\S_{11}$), the machine sees that its event $\lgname{partOK}_5$ is invalidated by the earlier $\lgname{partOK}_3$.

However, the machine $i_{\rolename{D}}$ (for the door) reaches a different conclusion: based on its local log in state $\S_9$, it interprets the event $\lgname{partOK}_5$ as a response to
$\lgname{partReq}_4$, when it was in fact emitted as a response to $\lgname{partReq}_1$. %
Consequently, the machines in $\S$ do not reach a consensus on which events are valid and conform to
an execution of $\Gn{Warehouse}$ --- despite having been projected from $\Gn{Warehouse}$ for $\subscription$. Specifically, in the final swarm state $\S_{11}$ reached after the propagation of the log $\lg$,
the two machines $\afish[{\rolename{T}}]$ (for the transport) with IDs $i_{{\rolename{T}}_1}$ and $i_{{\rolename{T}}_2}$ agree on the projection of the same state of the swarm protocol $\Gn{Warehouse}$
--- whereas $\afish[{\rolename{D}}]$ with ID $i_{\rolename{D}}$ is in a completely different state:

\ifbool{ecoop}{\smallskip}{\medskip}%
  \centerline{\(
      \stchange(\afish[{\rolename{T}}], \alog) \;=\; (\cstmcmd{FL}{pos}.\cstmcmd{T}{partOK}.\Gn{Warehouse})\projshort{\subscription}{\rolename{T}}%
      \quad\text{and}\quad%
      \stchange(\afish[{\rolename{D}}], \alog) \;=\; \zero
\)}%
\ifbool{ecoop}{\smallskip}{\medskip}%

\noindent%
The machines $\afish[{\rolename{T}}]$ expect the forklift to emit an event of type $\lgtn{pos}$, while $\afish[{\rolename{D}}]$ believes
the door has closed and the protocol has ended.
The swarm cannot recover from this conflict.\footnote{%
  \label{footnote:additional-subscriptions}%
  To avoid the issue in \Cref{ex:non-bt-swarm-warehouse} with the swarm semantics \Cref{def:non-bt-swarm-semantics} (from \cite{DBLP:conf/ecoop/KuhnMT23}),
  we should also subscribe role $\rolename{D}$ in $\subscription$ to events of type $\lgtn{pos}$. %
  The additional subscriptions needed by Def.~\ref{def:non-bt-swarm-semantics} would increase with swarm composition,
  leading to inefficient swarms: we measure these additional subscriptions in § 6.1.%
}%
\end{example}

\end{document}